\pdfoutput=1
\documentclass[12pt, letter]{article}

\usepackage{wrapfig}
\usepackage{amsfonts}
\usepackage{amssymb}
\usepackage{amsmath,tabu}
\usepackage{amsthm}
\usepackage{comment}
\usepackage{diagbox}
\usepackage{graphics}
\usepackage{epstopdf}
\usepackage{pst-all,psfrag}
\usepackage[unicode]{hyperref}
\usepackage{indentfirst}
\usepackage{mathtools}
\usepackage{subcaption}
\usepackage{multicol}

\usepackage{natbib}

\usepackage{verbatim}
\usepackage{eurosym}
\usepackage[utf8]{inputenc}
\usepackage[english]{babel}
\usepackage{chngcntr}
\usepackage{apptools}
\usepackage{tabularx}
\usepackage{setspace}
\usepackage{enumitem}
\usepackage{bbm}

\usepackage[textsize=footnotesize]{todonotes}

\usepackage{accents}

\usepackage[section]{placeins}

\AtAppendix{\counterwithin{theorem}{section}}

\newtheorem{theorem}{Theorem}
\newtheorem{proposition}[theorem]{Proposition}

\newtheorem{lemma}[theorem]{Lemma}

\newtheorem{assumption}{Assumption}

\newtheorem{corollary}[theorem]{Corollary}
\newtheorem{remark}[theorem]{Remark}

\newtheorem{definition}[theorem]{Definition}
\newtheorem{example}[theorem]{Example}

\numberwithin{theorem}{section}

\newcommand{\eps}{\varepsilon}

\newcommand{\U}{\mathbf U}

\newcommand{\V}{\mathbf V}
\newcommand{\W}{\mathbf W}
\renewcommand{\v}{\mathbf v}
\renewcommand{\u}{\mathbf u}
\newcommand{\w}{\mathbf w}
\newcommand{\x}{\mathbf x}
\newcommand{\y}{\mathbf y}
\newcommand{\ba}{\boldsymbol\alpha}
\newcommand{\bb}{\boldsymbol\beta}
\newcommand{\be}{\boldsymbol \eta}
\newcommand{\bz}{\boldsymbol \zeta}
\newcommand{\bg}{\boldsymbol \gamma}

\newcommand{\s}{\mathfrak s}

\newcommand{\p}{\mathfrak p}
\newcommand{\q}{\mathfrak q}

\newcommand{\dd}{\mathrm d}

\newcommand{\E}{\mathbb E}

\newcommand{\T}{\mathsf T}

\begin{document}

\def\spacingset#1{\renewcommand{\baselinestretch}%
{#1}\small\normalsize} \spacingset{1}

  \title{\bf High-dimensional canonical correlation analysis\thanks{The authors would like to thank Tim Bollerslev, Grigori Olshanski, Alexei Onatski, Andrew Patton, editors, and referees for valuable comments and suggestions. The authors also would like to thank Thomas Barker for excellent research assistance. }}
  \author{Anna Bykhovskaya \\
    Department of Economics\\ Duke University\\ \url{anna.bykhovskaya@duke.edu}\\
    \\
    and \\
    \\
    Vadim Gorin\thanks{Gorin's work was partially supported by NSF grant DMS-2246449.} \\
    Departments of Statistics and Mathematics\\ University of California at Berkeley\\ \url{vadicgor@gmail.com}}
    \date{August 29, 2026}
  \maketitle





\begin{abstract}
This paper studies high-dimensional canonical correlation analysis (CCA) with an emphasis on the vectors that define canonical variables. The paper shows that when two dimensions of data grow to infinity jointly and proportionally, the classical CCA procedure for estimating those vectors fails to deliver a consistent estimate. This provides the first result on the impossibility of identification of canonical variables in the CCA procedure when all dimensions are large. As a countermeasure, the paper derives the magnitude of the estimation error, which can be used in practice to assess the precision of CCA estimates. Applications of the results to cyclical vs.~non-cyclical stocks and to a limestone grassland data set are provided.
\end{abstract}


\newpage
\spacingset{1.9} 

\section{Introduction}

\subsection{Background}\label{section_intro_background}
Canonical correlation analysis (CCA) is a classical statistical analysis method used to find a common structure between two data sets. It was first introduced in \citet{harold1936relations} and remains in active use today. CCA can be viewed as a generalization of principal component analysis (PCA) from one set of variables to two: in PCA the goal is to find a signal (or factor) in a large matrix with a large amount of noise, while CCA searches for a common signal among two large matrices with a large amount of noise. We refer the reader to the textbooks \citet{thompson1984canonical,gittins1985canonical,anderson1958introduction,muirhead2009aspects} for introductions to CCA.

There are numerous applications of CCA in the social and natural sciences: in genomics (e.g., \citet{witten2009extensions} and references therein), neuroscience (e.g., \cite{wang2020finding,zhuang2020technical}), ecology (e.g., \citet{gittins1985canonical}, \citet[Chapter 23]{simon1998assessing}), econometrics (e.g., \citet{breitung2013canonical,andreou2019inference,choi2021canonical,franchi2023estimating} for factor analysis and \citet{johansen1988} for cointegration analysis), finance and financial econometrics (e.g., in the contexts related to maximally predictable portfolios, as in \citet{lo1997maximizing}, \citet{goulet2023maximally}, or \citet{firoozye2023canonical}). 

Most of the above settings ideally require working with very large (across all dimensions) data sets. However, the high-dimensional machinery for CCA (in contrast to that for PCA) is still very much underdeveloped. In particular, while most scenarios crucially require estimates of canonical vectors (or variables), the literature has not provided any theoretic guarantees for the consistency of such estimates in general settings when we do not know anything about their structure and do not impose further restrictions such as sparsity, smoothness, etc. Thus, new tools and theory suitable for high-dimensional settings remain vital. In this paper we aim to provide such techniques. Our key result quantifies the inconsistency in CCA by providing exact formulas to measure how far the sample canonical vectors are from their population counterparts. These formulas are easy to evaluate and are readily applicable for practitioners seeking to assess the quality of their estimates.

\subsection{High-dimensional CCA setup}
For two random vectors $\u\in\mathbb{R}^K$ and $\v\in\mathbb{R}^M$, the first goal of CCA is to find deterministic vectors $\ba\in\mathbb{R}^K,\,\bb\in\mathbb{R}^M$ that maximize the correlation between $\u^\T\ba$ and $\v^\T\bb$, where $^\T$ is the matrix transposition. The maximal correlation value is called the largest canonical correlation, and the corresponding $\u^\T\ba$ and $\v^\T\bb$ form the first pair of canonical variables. Other canonical correlations can be found iteratively. It is also known that squares of all canonical correlations can be found as the eigenvalues of $K\times K$ matrix $(\E \u \u^\T)^{-1} (\E\u \v^\T) (\E\v \v^\T)^{-1} (\E\v \u^\T)$ or as the eigenvalues of the ${M\times M}$ matrix $(\E\v \v^\T)^{-1} (\E\v \u^\T) (\E\u \u^\T)^{-1} (\E\u \v^\T)$, see Appendix \ref{Section_CCA_master}. By multiplying the corresponding eigenvectors of the former and the latter matrices by $\u^\T$ and $\v^\T$, respectively, we obtain all the pairs of canonical variables. We call the setup with two random vectors a population formulation.


In contrast, in the sample formulation, two vectors are replaced by two matrices: $K\times S$ matrix $\U$ and $M\times S$ matrix $\V$. The finite sample analogues of the canonical correlations, the  \textit{sample} canonical correlations, are computed by maximizing the sample correlations between the $S$-dimensional vectors $\U^\T\widehat \ba$ and $\V^\T\widehat \bb$. Equivalently, their squares are eigenvalues of $(\U\U^\T)^{-1}\U\V^\T(\V\V^\T)^{-1}\V\U^\T$. The canonical correlation vectors are eigenvectors of the preceding product of matrices and another similar expression. The corresponding $\U^\T\widehat \ba$ and $\V^\T\widehat \bb$ are called sample canonical variables.

Assume that the columns of $\mathbf U$ and $\mathbf V$ are independent samples from the joint law of $(\mathbf u,\mathbf v)$. Then, when $S$ is large while $K$ and $M$ are fixed, the sample covariances of $\U$ and $\V$ as well as their cross-covariances are consistent estimators of their population counterparts. Therefore, $(\U\U^\T)^{-1}\U\V^\T(\V\V^\T)^{-1}\V\U^\T$ is a consistent estimate of  $(\E \u \u^\T)^{-1} (\E\u \v^\T) (\E\v \v^\T)^{-1} (\E\v \u^\T)$, and we obtain consistent estimates of the squared correlation coefficients and corresponding vectors. The distribution of the sample canonical correlations and variables and their asymptotics in the fixed $K$ and $M$ regime have attracted the attention of many researchers, starting from \citet{hsu1939distribution,hsu1941limiting} and  \citet{constantine1963some} and extending to more recent \citet{anderson1999asymptotic}.

The properties of CCA when all three dimensions $S,K,M$ are large and comparable are less understood. The remarkable seminal result in this direction is \citet{wachter1980limiting}, which studies the case when $\u$ and $\v$ are uncorrelated and columns $\U$ and $\V$ are independent samples of $\u$ and $\v$, respectively. \citet{wachter1980limiting} shows that despite $\U$ and $\V$ being independent, the empirical distribution of the sample canonical correlations has a nontrivial limit as $S,K,M\to\infty$ with $S/K\to \tau_K$, $S/M\to\tau_M$; see also \citet{bouchaud2007large} for a rediscovery. In particular, the majority of the correlations are bounded away from $0$. The conceptual conclusion is that in the large $S,K,M$ setting the sample canonical correlations \emph{do not} deliver consistent estimates of the population canonical correlations, which in this case all equal zero due to $\u$ and $\v$ being independent. These results have been extended very recently in \citet{bao2019canonical}, \citet{yang2022limiting}, and \citet{ma2023sample} to the case when some (but finitely many as the dimensions grow) population correlations are nonzero. Conceptually similar results are obtained: the sample squared canonical correlations are always larger and bounded away from their population counterparts. However, there is an explicit dependence between the sample and population canonical correlations, and therefore, knowing the former allows one to reconstruct the latter. Hence, it is possible to identify nontrivial population canonical correlations. No progress has been achieved so far to tackle the challenging properties of the corresponding vectors of canonical variables in the regime in which $S,K,M$ are large and proportional to each other. 

To address the problem of estimating CCA vectors, we first assume that $\U$ and $\V$ are Gaussian and independent across $S$ samples of $\u$ and $\v$ but allow $\u$ to be correlated with $\v$. We start with $\u$ and $\v$ having exactly one nonzero canonical correlation and refer to it as the \emph{signal}, while the remaining uncorrelated parts of $\u$ and $\v$ are viewed as the \emph{noise}. We characterize (in terms of the parameters $S,K,M$ and the covariances between coordinates of $\u$ and $\v$) when we can detect a nonzero canonical correlation, i.e.,~when there is enough signal in the data compared to the pure noise case. Next, we develop formulas for sample $\widehat \ba$ and $\widehat \bb$, which allow us to show that, for finite ratios $S/K,\,S/M$, one \emph{cannot} consistently estimate true values $\ba$ and $\bb$. The estimated canonical variables $\U^\T\widehat \ba$ and $\V^\T\widehat \bb$ lie on cones around the true population canonical variables $\U^\T\ba$ and $\V^\T\bb$. We provide explicit formulas for the width of these cones, which are simple to use. In particular, they show that the cones shrink as we increase the ratios $S/K,\,S/M$ and that ultimately consistency is restored in the limit. We also establish asymptotic Gaussianity for the fluctuations of the cone widths.

Next, we provide various generalizations from the i.i.d.~Gaussian setting with one signal: we accommodate any distribution as long as its first four moments match the Gaussian moments and we allow for correlated realizations, treating the signal part (i.e., the signal is not independent across $S$) and the noise part (i.e., the noise is not independent across $S$) separately. We do not cover the case of simultaneously correlated structures for the signal and noise in this paper. We remark that such setting could be of interest for the cointegration testing, where the no-signal (or no-cointegration) case has been studied in recent works such as \citet{onatski_ecta, onatski2019extreme} and \citet{BG1, BG2}.


Finally, we extend the machinery to allow for multiple vectors $\ba$ and $\bb$, i.e.,\ multiple nonzero canonical correlations in the population setting. In other words, we allow the signal to be of any finite rank, and the vectors $\ba$ and $\bb$ are replaced by matrices.

\subsection{Other related literature}
\subsubsection{CCA vs.~PCA and factor models}
CCA is closely connected to the literature on factors. First, some methods for factor estimation rely on CCA, because the correlation between different sets of variables can be fundamentally related to the presence and properties of factors. For instance, CCA helps distinguish between common and group-specific factors (e.g., \citet{andreou2019inference,choi2021canonical}) or infer the dimension of non-stationary subspace (e.g., \citet{franchi2023estimating}).

Another point of view on the interconnection between factors and CCA comes from the interpretation of canonical variables as common explanatory factors between two data sets. The main approach for finding factors is via PCA and modifications (cf.\ \citet{bai2008large,stock2011dynamic,fan2016overview} and references therein). Despite many similarities in the spirit of the modeling and results between CCA and PCA,  the former is more challenging than PCA. This is perhaps an inherent feature of the CCA procedure, which involves matrix inversion, a much more complicated operation to carry out than simply calculating a product.\footnote{The reader is invited to compare the difficulty of the answers in our Theorem \ref{Theorem_master} with that of the parallel Theorems 2.9 and 2.10 in \citet{benaych2012singular} for the PCA setting.}


In statistics and probability factors (or signals) are often modeled as a spiked random matrix, a term going back to \citet{johnstone2001distribution}. A spiked random matrix is a sum of a full-rank noise matrix and a small-rank signal matrix. Often this is embedded in the PCA setting: one observes a matrix $C=A+B$, where $A$ is treated as a low-rank signal and $B$ is treated as noise, and tries to reconstruct $A$ through singular values and vectors of observed $C$. All dimensions of the matrices are assumed to be large here. In econometrics $A$ represents a product of factors and their loadings, while $B$ is composed of idiosyncratic errors.

The spiked covariance model, which begins with a positive-definite covariance matrix $\Omega=\Sigma+\sigma^2 I_N$, can be viewed as a special case of the spiked random matrix framework. In this model $\Sigma$ is a low-rank signal matrix, $I_N$ is the $N\times N$ identity matrix and $\sigma^2$ is a real parameter that quantifies the noise strength. One observes $N\times S$ matrix $X$, whose columns are i.i.d.\ samples of a random vector with covariance $\Omega$. The goal is to reconstruct $\Omega$ (equivalently, $\Sigma$) from the sample covariance matrix $\frac{1}{S}X X^\T$. For such a setting, \citet{baik2005phase} and \citet{baik2006eigenvalues} discovered the phenomenon now known as the BBP phase transition: when the signal is small while the strength of the noise is large, one cannot detect the presence of the signal matrix $\Sigma$ from eigenvalues of $\frac{1}{S}X X^\T$ , while for a large signal, the largest eigenvalues of $\frac{1}{S}X X^\T$ connect to those of $\Sigma$ through explicit formulas. The corresponding sample eigenvectors are inconsistent
estimators of the population spike directions, in parallel to what we observe in the CCA setting; see, among others,
\citet{paul2007asymptotics,johnstone2009consistency,nadler2008finite,bloemendal2016principal,dingYang2021}.

More generally, the first-order behavior of eigenvectors under finite-rank perturbations has been characterized
for additive and multiplicative deformations, spiked covariance
matrices, and rectangular signal-plus-noise models; see
\citet{benaych2011eigenvalues,benaych2012singular,
capitaine2013subordination}. Second-order fluctuations were also studied, including
\citet{hachemEtAl2013subspace,wangFan2017eigenstructure,capitaineDonatiMartin2021outlierEigenvectors,baoDingWang2021singular,baoDingWangWang2022pca}.

If the rank of the signal matrix $\Sigma$ is not restricted to be small, a popular approach for estimating it from $\frac{1}{S}X X^\T$ is known as ``shrinkage''. While this approach does not yield consistent estimators, it is optimal in a certain sense, see \citet{ledoit2022power} for a review.

The phase transition analogous to the BBP has also been observed in the econometrics literature on factor models. When factors are weakly influential, their identification becomes possible only if the signal's strength is above a threshold (see \citet{onatski2012asymptotics} and references therein). Moreover,  when factors become stronger, estimation consistency is restored.

\subsubsection{CCA vs.~$F$-matrix}
Closely related to the CCA framework is a spiked $F$-matrix model, explored recently in, e.g., \citet{gavish2023matrix,hou2023spiked}. An $F$-matrix, arising from the $F$-test, involves the ratio of two large matrices. In the benchmark case of no spikes (no signal), there exists a direct correspondence between an $F$-matrix and the random matrices used in CCA. When signals are present in the data, a more delicate connection between non-central $F$-matrices and CCA has been established in recent work, \citet{bai2022limiting,zhang2023limiting}. However, this connection applies only to eigenvalues, not eigenvectors, leaving it unclear how to extend these results to the vectors of canonical variables, which are the focus of our study.

\subsubsection{CCA vs.~sparse and regularized CCA}
Although a precise measure of CCA inconsistency was previously unavailable, researchers recognized that sample canonical correlations and variables degrade in quality when the sample-to-dimension ratios, $S/K$ and $S/M$ are small. To address this, it is common --- drawing from other high-dimensional problems --- to impose additional constraints, such as sparsity, on canonical variables. For instance, \citet{gao2015minimax,gao2017sparse} study CCA under sparsity, while \citet{tuzhilina2023canonical} explore other regularization techniques. These approaches can simplify interpretation and reduce dimensionality, yet often involve complex computations and may sacrifice desirable properties. In contrast, unconstrained sample canonical variables are maximum likelihood estimators, offering optimality guarantees.

Practically, it is vital to verify whether imposed assumptions, like sparsity, fit the data --- a topic still under debate. \citet{giannone2021economic} note that many macroeconomic, microeconomic, and financial data sets tend to be dense. Likewise, in our financial data example (Section \ref{Section_stocks}), most coefficients in the leading canonical variables are positive and similar in magnitude, making sparsity implausible. This highlights the need for methods, like ours, that require only minimal assumptions.

\subsection{Outline of the paper} The paper is organised as follows. Section \ref{Section_setting} discusses the setting of i.i.d.~normal vectors. Generalizations to the cases with nonnormal errors, correlated observations, and multiple signals are in Section \ref{Section_general}. Section \ref{Section_empirical} provides empirical illustrations. Section \ref{Section_conclusion} concludes. Proofs, additional discussions,  simulations, and data details are in supplementary appendices.

\section{Basic framework}
\label{Section_setting}

\subsection{Population setting}\label{Section_basic_population_setting}
Let $\u=(u_1,\dots,u_K)^\T$ and $\v=(v_1,\dots,v_M)^\T$ be $K$- and $M$- dimensional random vectors with zero means and nondegenerate covariance matrices, where $\,^\T$ here and below denotes matrix transposition. Without loss of generality, we assume $K\le M$ throughout the paper.

\begin{assumption}\label{ass_basic} The vectors $\u$ and $\v$ satisfy the following:
\begin{enumerate}[label=(\arabic*), ref=\ref{ass_basic}.(\arabic*)]
\item\label{ass_basic_gauss} The vectors $\u$ and $\v$ are jointly Gaussian with zero means.
\item\label{ass_basic_indep_set} There exist nonzero deterministic vectors $\ba\in\mathbb{R}^K,\,\bb\in\mathbb{R}^M$ such that
\begin{enumerate}
  \item For any $\bg\in\mathbb{R}^K$, if $\u^\T\ba$ and $\u^\T\bg$ are uncorrelated, then $\v$ and $\u^\T\bg$ are also uncorrelated;
  \item For any $\bg\in\mathbb{R}^M$, if $\v^\T\bb$ and $\v^\T\bg$ are uncorrelated, then $\u$ and $\v^\T\bg$ are also uncorrelated.
\end{enumerate}
\end{enumerate}
\end{assumption}

The pair $(\u^\T\ba, \v^\T\bb)$ represents the signal in the data, while the remaining part is treated as the noise, which is uncorrelated with the signal. Ultimately, we are interested in the signal and would like to filter out the influence of the noise.
\begin{example}\label{ex_10}
\emph{Assumption \ref{ass_basic} is satisfied with $\ba=(1,0,\ldots,0)^\T,\,\bb=(1,0,\ldots,0)^\T$ if $(u_1,\dots,u_K,v_1,\dots,v_M)^\T$ is a mean zero Gaussian vector such that, for each $2\le k\le K$, the coordinate $u_k$ is uncorrelated with $u_1$ and with $v_1,\dots, v_M$ and, for each $2\le m \le M$, the coordinate $v_m$ is uncorrelated with $v_1$ and with $u_1,\dots,u_K$. In other words, the only possible nonzero correlations are between $u_1$ and $v_1$, between $u_k$ and $u_{k'}$, $2\le k,k'\le K$, and between $v_m$ and $v_{m'}$, $2\le m,m'\le M$. These correlations can be arbitrary.}

\emph{Any other example can be obtained from the preceding by a change of basis.}
\end{example}

\begin{definition} \label{Def_squared_correlation_coef} Let $C_{uu}:=\E (\u^\T\ba)^2$, $C_{vv}:=\E (\v^\T\bb)^2$, $C_{uv}=C_{vu}:=\E (\u^\T\ba) (\v^\T\bb)$ and $r^2$ be the squared correlation coefficient:
\vskip-1cm
$$
r^2:=\frac{C_{uv}^2}{C_{uu} C_{vv}}.
$$
\end{definition}

The number $r^2$ and the vectors $\ba$ and $\bb$ of Assumption \ref{ass_basic} can be read from the covariance structure of $\u$ and $\v$, as the following lemma explains. We omit the proof.
\begin{lemma} \label{Lemma_signal_population}
 If $r^2>0$, then the number $r^2$ equals the single nonzero eigenvalue of the $K\times K$ matrix
 $(\E \u \u^\T)^{-1} (\E\u \v^\T) (\E\v \v^\T)^{-1} (\E\v \u^\T)$ and the single nonzero eigenvalue of the ${M\times M}$ matrix $(\E\v \v^\T)^{-1} (\E\v \u^\T) (\E\u \u^\T)^{-1} (\E\u \v^\T)$. Vectors $\ba$ and $\bb$ are the corresponding eigenvectors of the former and the latter matrices, respectively.
\end{lemma}

\subsection{Sample setting} \label{Section_basic_sample_setting}

Let $\W=\bigl(\U^\T,\V^\T\bigr)^\T$ be a $(K+M)\times S$ matrix composed of $S$ independent samples of $\bigl(\u^\T,\v^\T\bigr)^\T$. The matrix $\W$ represents observed data, which comes from the population setting \ref{Section_basic_population_setting}. We are interested in finding the signal  $(\u^\T\ba, \v^\T\bb)$ or vectors $(\ba,\bb)$. In the sample setting they come from the squared sample canonical correlations and their corresponding vectors. The vectors for the largest correlation represent the sample analogues of $\ba$ and $\bb$. The following definition is motivated as a sample version of Lemma \ref{Lemma_signal_population}. As before, $K\le M$.

\begin{definition} \label{Definition_sample_setting}
The squared sample canonical correlations $\lambda_1\ge \lambda_2\ge \dots\ge \lambda_K$ are defined as eigenvalues of the $K\times K$ matrix $(\U \U^\T)^{-1} \U \V^\T (\V \V^\T)^{-1} \V \U^\T$ or as the $K$ largest eigenvalues of the $M\times M$ matrix $(\V \V^\T)^{-1} \V \U^\T (\U \U^\T)^{-1} \U \V^\T$. We set $\widehat \ba$ to be the principal eigenvector of the former matrix, and we set $\widehat \bb$ to be the principal eigenvector of the latter matrix.

 The sample canonical variables are defined as $\widehat \x=\U^\T \widehat \ba$ and $\widehat \y =\V^\T \widehat \bb$. We also set $\x=\U^\T \ba$ and $\y =\V^\T \bb$.
\end{definition}

\subsection{Results}
In the sample setting of Section \ref{Section_basic_sample_setting}, we assume that $K$ and $M$ depend on $S$ in such a way that $S/K\to \tau_K$, $S/M\to\tau_M$, and $r^2\to \rho^2$. We define six constants depending on $\tau_K,\tau_M,\rho^2$:
\vskip-1.5cm
\begin{align}
\label{eq_rhoc} \rho_c^2&=\frac{1}{\sqrt{(\tau_M-1)(\tau_K-1)}},\\
\label{eq_lambda_plus} \lambda_\pm&=\left(\sqrt{\tau_M^{-1}(1-\tau_K^{-1})}\pm \sqrt{\tau_K^{-1}(1-\tau_M^{-1})}  \right)^2,\\
\label{eq_zrho}  z_\rho&=\frac{\bigl(  (\tau_K-1)\rho^2  + 1 \bigr) \bigl(  (\tau_M-1) \rho^2 + 1\bigr)}{\rho^2 \tau_K \tau_M },\\
\label{eq_sx} \s_x&= \frac{(1-\rho^2)(\tau_K-1)}{(\tau_M-1)(\tau_K-1) \rho^2 - 1 } \cdot \frac{ (\tau_M-1) \rho^2 + 1 }{(\tau_K-1) \rho^2 + 1},\\
\label{eq_sy} \s_y&= \frac{(1-\rho^2)(\tau_M-1)}{(\tau_M-1)(\tau_K-1) \rho^2 - 1 } \cdot \frac{ (\tau_K-1) \rho^2 + 1 }{(\tau_M-1) \rho^2 + 1}.
\end{align}

The constants play the following role. First, by \cite{wachter1980limiting} and \cite{Johnstone_Jacobi}, most of the squared sample canonical correlations should belong to the $[\lambda_-,\lambda_+]$ interval. \citet{wachter1980limiting} shows that as $S\to\infty$, the empirical distribution of the squared sample canonical correlations between \emph{independent} data sets converges to the Wachter distribution of density $\omega_{\tau_K,\tau_M}(x)=\frac{\tau_K}{2\pi} \frac{\sqrt{(x-\lambda_-)(\lambda_+-x)}}{x (1-x)} \mathbf 1_{[\lambda_-,\lambda_+]}(x)$, see Appendix \ref{Section_Wachter_specialization} for more details and Fig.~\ref{Fig_spike} for an illustration. This fact remains true for our setting, as follows from Lemma \ref{Lemma_interlacing} in Appendix.

Second, \citet{bao2019canonical} show that the largest squared canonical correlation undergoes a phase transition depending on the value of $\rho^2$: if $\rho^2>\rho_c^2$, then the largest squared canonical correlation is close to $z_\rho>\lambda_+$; otherwise, it is close to $\lambda_+$. In the former case, this squared canonical correlation can be clearly seen on the histogram as an outlier from the Wachter distribution, see Fig.~\ref{Fig_spike}. We focus on the case $\rho^2>\rho_c^2$ and examine how close the sample canonical variables corresponding to the largest correlation are to their population counterparts by analyzing the angles between them. The constants $\s_x, \s_y$ provide the answers.

\begin{figure}[t]
    \centering
        \includegraphics[width=0.6\textwidth]{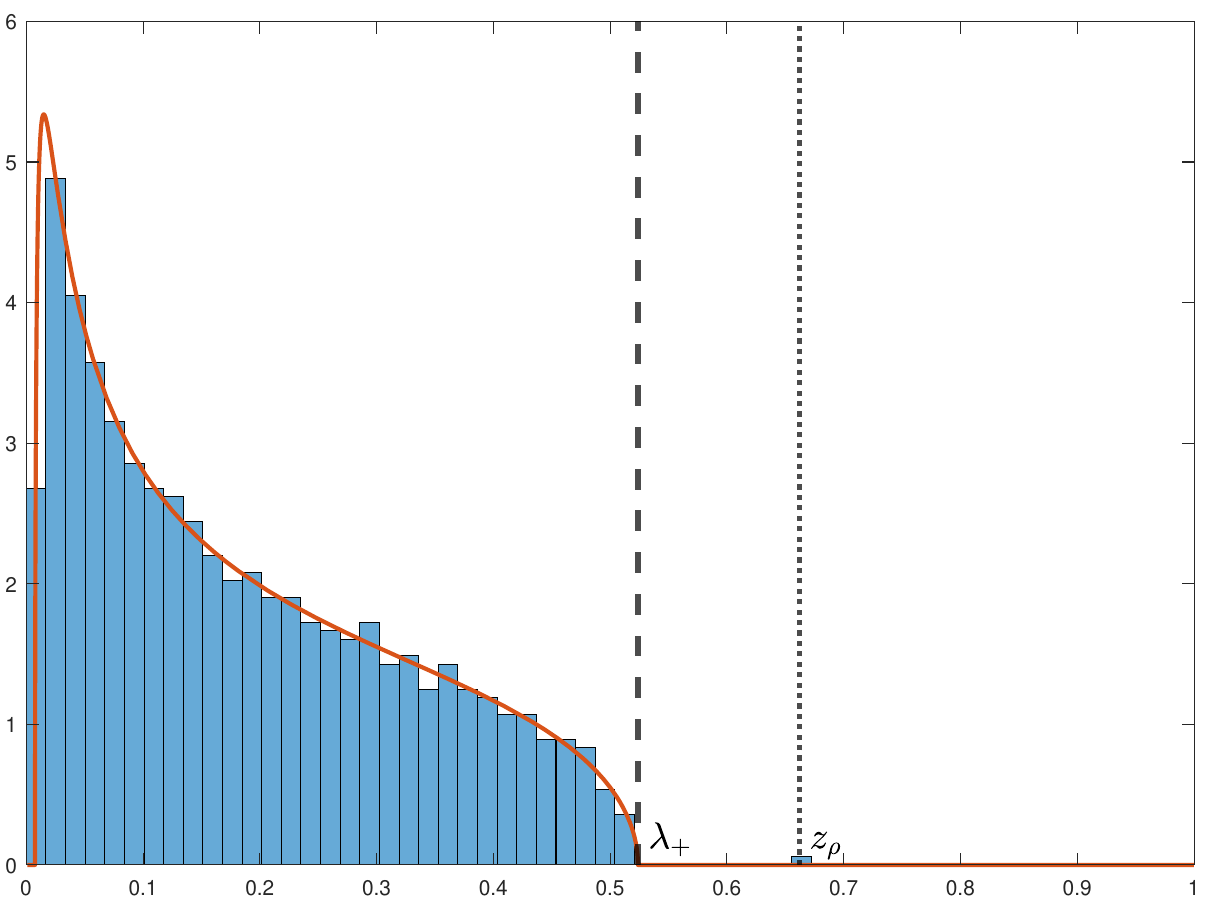}
    \caption{Illustration of Theorem \ref{Theorem_basic_setting}: Histogram of the squared sample canonical correlations from one simulation with $K=1000$, $M=1500$, $S=8000$, $r^2=0.49$. A single spike in the correlations is approximately at $z_\rho$. The density of the Wachter distribution with corresponding parameters is shown in orange.} \label{Fig_spike}
\end{figure}

\begin{theorem} \label{Theorem_basic_setting}
 Suppose that Assumption \ref{ass_basic} holds, that the columns of the data matrix $\W$ are i.i.d., and that the squared sample canonical correlations and variables are constructed as in Section \ref{Section_basic_sample_setting}. Let $S$ tend to infinity and $K\le M$ depend on it in such a way that the ratios $S/K$ and $S/M$ converge\footnote{To match the notations of \cite{bao2019canonical}, we should set $c_1=\tau_M^{-1}$ and $c_2=\tau_K^{-1}$. \citet{bao2019canonical} contains an earlier proof of \eqref{eq_canonical_limit} by another method; the part in \eqref{eq_vector_limit1} and \eqref{eq_vector_limit2} is new.} to $\tau_K>1$ and $\tau_M>1$, respectively, and $\tau_M^{-1}+\tau_K^{-1}<1$. Simultaneously, suppose that $\lim_{S\to\infty} r^2=\rho^2$.
 \begin{enumerate}
  \item[\bf 1.] If $\rho^2> \rho_c^2$, then $z_\rho>\lambda_+$, and for the two largest squared sample canonical correlations ${\lambda_1\ge \lambda_2}$, we have
  \vskip-1cm
  \begin{equation}
  \label{eq_canonical_limit}
   \lim_{S\to\infty} \lambda_1=z_\rho, \qquad \lim_{S\to\infty} \lambda_2=\lambda_+  \qquad \text{in probability}.
  \end{equation}
  The squared sine of the angle $\theta_x$ between $\x$ and $\widehat \x$ defined as
  $
   \sin^2\theta_x = 1- \frac{(\x^\T \widehat \x)^2}{(\x^\T \x)(\widehat \x^\T \widehat \x)}$ satisfies
  \vskip-1cm
  \begin{equation}
  \label{eq_vector_limit1}
   \lim_{S\to\infty} \sin^2\theta_x=\s_x \qquad \text{in probability}.
  \end{equation}
  The squared sine of the angle $\theta_y$ between $\y$ and $\widehat \y$ satisfies
  \vskip-1cm
  \begin{equation}
  \label{eq_vector_limit2}
   \lim_{S\to\infty} \sin^2\theta_y=\s_y \qquad \text{in probability}.
  \end{equation}
  \item[\bf 2.] If instead $\rho^2\le \rho_c^2$, then
  \vskip-1cm
  \begin{equation}
  \label{eq_no_spike}
   \lim_{S\to\infty} \lambda_1=\lambda_+ \qquad \text{in probability}.
  \end{equation}
 \end{enumerate}
\end{theorem}
\begin{remark}
While we do not prove it in this text, we expect (based on the related results in \citet{benaych2011eigenvalues, benaych2012singular} and \citet{bloemendal2016principal}) that in the case $\bf 2$ with strict inequality the angles $\theta_x$ and $\theta_y$ tend to $90$ degrees as $S\to\infty$, with distance to $90$ degrees of order $S^{-1/2}$ (see Fig.~\ref{fig_angles_theor_simul} for an illustration). This is in line with $\s_x$ and $\s_y$ tending to $1$ as $\rho\to \rho_c$. Intuitively, this suggests that the data contains very limited information about $\ba$ or $\bb$ in this case; see additional discussion of what can be recovered for the cases analogous to $\bf 2$ in similar contexts in \cite{onatski2013asymptotic} and \cite{johnstone2020testing}.

In the critical case when $\rho^2$ is close to $\rho_c^2$, we expect that the angles $\theta_x$ and $\theta_y$ still tend to $90$ degrees, but the convergence might be very slow: a comparison with \cite{bloemendal2016principal} and \cite{bao2022eigenvector} predicts the distance to $90$ degrees to be of order $S^{-1/6}$.
\end{remark}

Fig.~\ref{fig_rho2_depend} illustrates the dependence of $\s_x,\s_y,z_{\rho}$ on the squared correlation coefficient $\rho^2$. Notice that $z_\rho$ as a function of $\rho^2$ has a minimum at the cutoff $\rho^2=\rho_c^2$ and is monotone increasing above the cutoff. The values of $\s_x$ and $\s_y$ are between $0$ and $1$ for $\rho^2$ above the same cutoff. These properties continue to hold for general values of $\tau_K,\tau_M>1$.

Fig.~\ref{fig_angles_theor_simul} shows results from a single simulation (for each fixed value $\rho^2$, we run one simulation) vs.~theoretical predictions. As we can see, the simulated path is very close to the theoretical one, with largest discrepancy around the cutoff $\rho_c^2$.
We also report in Appendix \ref{Section_appendix_MC} analogous results obtained from many instead of a single Monte Carlo simulations, as well as the finite sample effects on the convergence of angles.

\begin{figure}[t]
\begin{subfigure}{.49\textwidth}
  \centering
  \includegraphics[width=1.0\linewidth]{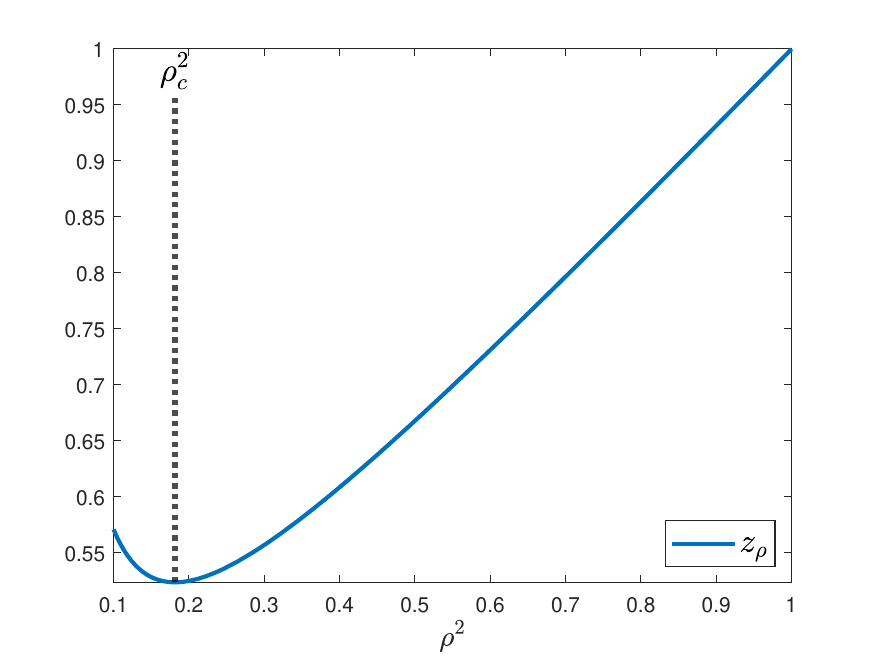}
  \caption{$z_{\rho}$ as a function of $\rho^2$.}
  \label{rho2_zrho}
\end{subfigure}%
\begin{subfigure}{.49\textwidth}
  \centering
  \includegraphics[width=1.0\linewidth]{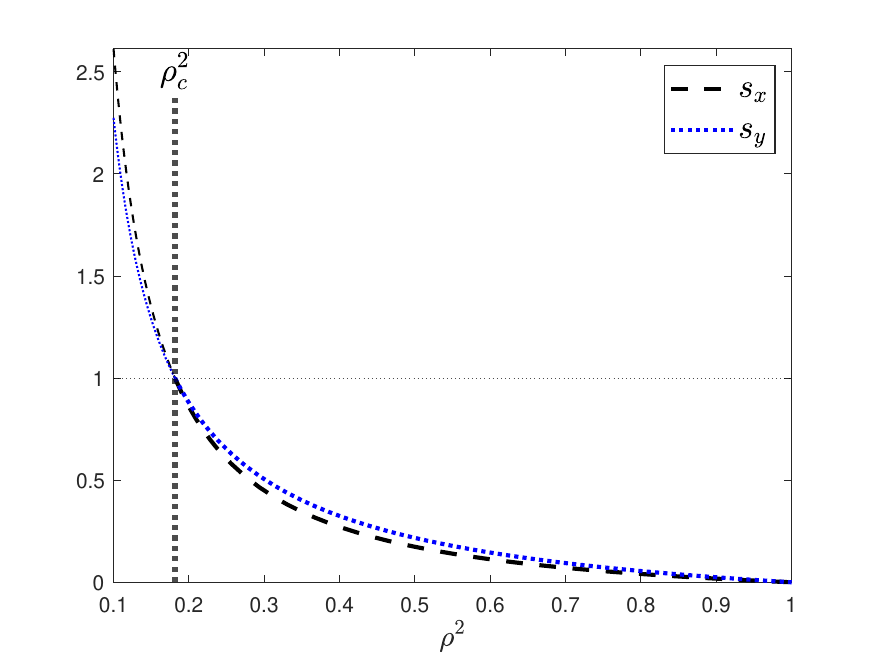}
  \caption{$\s_x$ and $\s_y$ as functions of $\rho^2$.}
  \label{rho2_sxsy}
\end{subfigure}
\caption{Illustration of Eq.~\eqref{eq_zrho}, \eqref{eq_sx}, \eqref{eq_sy} for $K=1000,\,M=1500,\,S=8000$.}
\label{fig_rho2_depend}
\end{figure}

\begin{figure}[t]
\begin{subfigure}{.49\textwidth}
  \centering
  \includegraphics[width=1.0\linewidth]{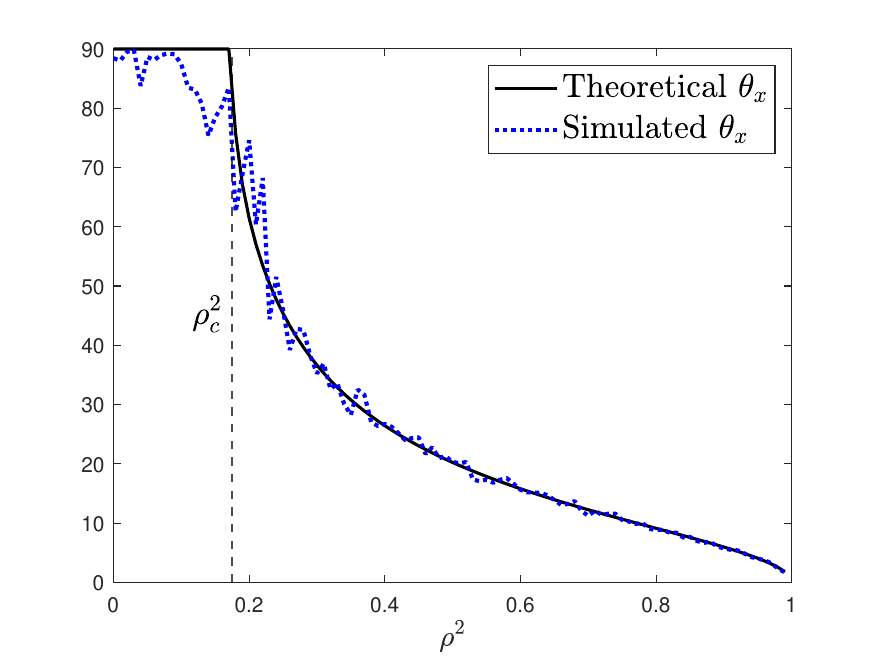}
  \caption{Values of $\theta_x$ (in degrees): \\theoretical and simulated angles.}
  \label{thetax_pic}
\end{subfigure}%
\begin{subfigure}{.49\textwidth}
  \centering
  \includegraphics[width=1.0\linewidth]{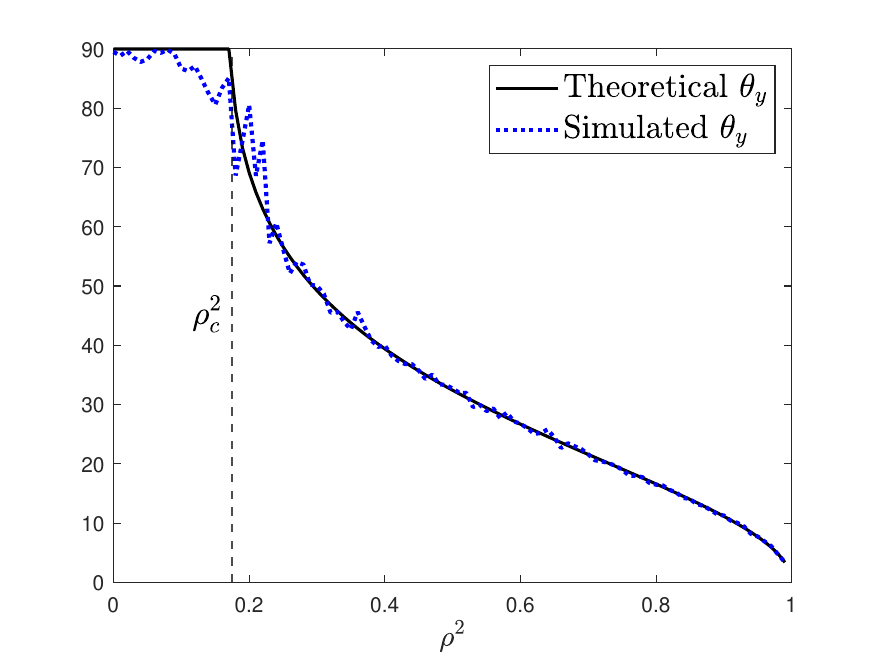}
  \caption{Values of $\theta_y$ (in degrees): \\theoretical and simulated angles.}
  \label{thetay_pic}
\end{subfigure}
\caption{Theoretical vs.~simulated results for ${K=500}$, ${M=2500}$, $S=8000$. Simulated curves are based on one simulation from a fixed value $\rho^2$.}
\label{fig_angles_theor_simul}
\end{figure}

The conditions $\tau_K>1$, $\tau_M>1$, and $\tau_M^{-1}+\tau_K^{-1}<1$ are conceptual restrictions of CCA. Suppose that $\tau_M\leq1$, so that $M\geq S$. Then, almost surely, the rows of $M\times S$ matrix $\V$ span the entire $S$--dimensional space. Therefore, the sample canonical correlations are all equal to $1$ and do not convey information about the population setting (the matrices in Definition \ref{Definition_sample_setting} are not invertible in this case, and one should use the equivalent definition of Lemma \ref{Lemma_canonical_bases}).

Similarly, if $\tau_K>1$, $\tau_M>1$, but $\tau_M^{-1}+\tau_K^{-1}>1$, which means $K<S$, $M<S$, but $K+M>S$, then any two subspaces of dimensions $K$ and $M$ must intersect in an $S$--dimensional vector space. The dimension of the intersection is at least $K+M-S$. Hence, in this situation $K+M-S$ largest canonical correlations are equal to $1$ (again, this can be seen from Lemma \ref{Lemma_canonical_bases} and subsequent discussion), no matter what the parameters in the population setting are. This makes the largest canonical correlation useless. Nevertheless, other sample canonical correlations remain bounded away from $1$ and potentially might be used to extract information about $\rho^2$. Yet, this is an open question. The condition $\tau_M^{-1}+\tau_K^{-1}<1$ is common in CCA literature, e.g., it is imposed in \citet{bao2019canonical,yang2022limiting}.


Finally, in the boundary case $\tau_K^{-1}+\tau_M^{-1}=1$, we have $\lambda_+=1$ and it becomes impossible to have sample canonical correlations larger than $\lambda_+$.

\medskip

Let us now describe the second order Gaussian corrections to the limits in \eqref{eq_canonical_limit}, \eqref{eq_vector_limit1}, and \eqref{eq_vector_limit2}. We  assume a fast enough speed of convergence for the parameters of the model, as otherwise the second order corrections in the parameters contribute to the answers:
\begin{equation}
\label{eq_params_exact_almost_main}
 \frac{S}{K}=\tau_K+o\left(S^{-1/2}\right), \qquad \frac{S}{M}=\tau_M+o\left(S^{-1/2}\right), \qquad r^2=\rho^2+o\left(S^{-1/2}\right), \qquad S\to\infty.
\end{equation}

\begin{theorem}\label{Theorem_joint_CLT_main}
Under assumptions of Theorem \ref{Theorem_basic_setting} and additionally \eqref{eq_params_exact_almost_main},  define
\begin{equation}\label{eq_Z_def_main}
        Z_x=\sqrt S\bigl(\sin^2\theta_x-\s_{x}\bigr),
        \qquad
        Z_y=\sqrt S\bigl(\sin^2\theta_y-\s_{y}\bigr),
        \qquad
        Z_z=\sqrt S\bigl(\lambda_1-z_{\rho}\bigr).
\end{equation}
For $\rho_c^2<\rho^2<1$ the random vector $(Z_x,Z_y,Z_z)$ is asymptotically Gaussian: in  distribution
\begin{equation}\label{eq_joint_CLT_main}
       \lim_{S\to\infty} \bigl(
        Z_x, Z_y, Z_z\bigr)^{\T}
        \stackrel{d}{=}
        \mathcal N(0,\Sigma),
\end{equation}
where the asymptotic covariance $\Sigma$ of $(Z_x,Z_y,Z_z)$ is given\footnote{The matrix elements of $\Sigma$ are explicit, but complicated rational expressions in $\tau_K$, $\tau_M$, and $\rho^2$. The asympotic variance of $Z_z$ was previously found in \cite{bao2019canonical}.} in Appendix \ref{Section_appendix_angle_CLT}.
\end{theorem}
\subsection{Discussion and implications of Theorem \ref{Theorem_basic_setting}}\label{sec_implic_basic}
In practice, when one is working with real data, the true value of $\rho$ is unknown, and the results of Theorem \ref{Theorem_basic_setting} give rise to the following algorithm.

Given that the model matches the data, by \cite{wachter1980limiting} and \cite{Johnstone_Jacobi}, the histogram of the squared canonical correlations should resemble the Wachter distribution supported on the $[\lambda_-,\lambda_+]$ interval. If there is a gap between the largest canonical correlation $\lambda_1$ and $\lambda_+$, as in Fig.~\ref{Fig_spike}, then in line with Eq.~\eqref{eq_canonical_limit}, one can take $\lambda_1$ as an approximation of $z_\rho$. Treating $z_\rho$ as known and approximating $\tau_K$, $\tau_M$ with $S/K$, $S/M$, Eq.~\eqref{eq_zrho} becomes a quadratic equation in $\rho^2$. Solving it (using $ \rho_c^2\le \rho^2\le 1$ to choose the correct root, cf.\ Fig.~\ref{rho2_zrho}), we obtain an estimate for $\rho^2$. We plug it into \eqref{eq_sx} and \eqref{eq_sy} and by Theorem \ref{Theorem_basic_setting} we obtain the angle between $\x$ and $\widehat \x$ and between $\y$ and $\widehat \y$. If several canonical correlations larger than $\lambda_+$ are observed, then one can use an extension of Theorem \ref{Theorem_basic_setting} provided in Theorem \ref{Theorem_basic_multi}: for each of these canonical correlations one applies exactly the same procedure as just outlined.

\begin{figure}[t]
    \centering
        \includegraphics[width=0.4\textwidth]{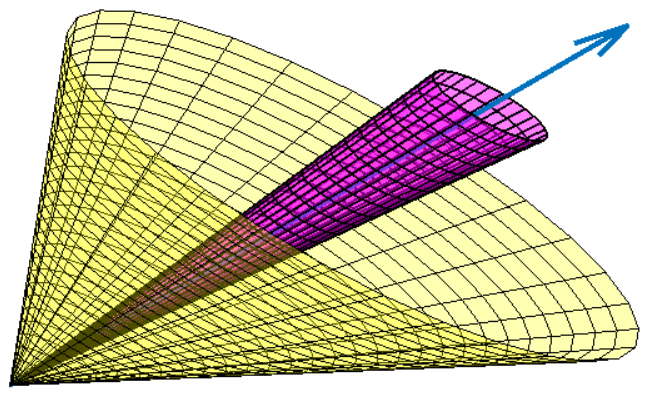}
    \caption{The estimated canonical variable belongs to a cone whose axis is the true direction shown by the blue arrow. If $\sin^2\theta$ is small, then the cone is narrow, as shown in purple; if $\sin^2\theta$ is large, then the cone is wide, as shown in yellow. \label{Fig_cone}}
\end{figure}

Next, the angles $\theta_x$, $\theta_y$ and the limiting value $z_\rho$ for the largest squared canonical correlation depend on $\rho$, $\tau_K$, and $\tau_M$ in a nontrivial fashion, leading to important features of \eqref{eq_zrho}--\eqref{eq_sy}:
\begin{itemize}
 \item For any values of $\tau_K>1$, $\tau_M>1$ with $\tau_K^{-1}+\tau_M^{-1}<1$ and $\rho_c^2<\rho^2<1$, we always have $z_\rho>\rho^2$ and $\s_x>0$, $\s_y>0$ (see Fig.~\ref{fig_rho2_depend}). In other words, the largest sample canonical correlation overestimates the true largest correlation $\rho^2$, while the estimates for the canonical variables are never consistent but rather inclined by nonvanishing angles toward the desired true directions (see Fig.~\ref{Fig_cone}).
 \item If $\tau_K>1$, $\tau_M>1$ with $\tau_K^{-1}+\tau_M^{-1}<1$ and $\rho^2< \rho_c^2$, then the asymptotic value of the largest canonical correlation, $\lambda_+$, is again larger than $\rho^2$. While we do not prove it in this text, we expect that the asymptotic values of $ \sin^2\theta_x$ and $ \sin^2\theta_y$ are close to $1$ in this situation, i.e.,\ that the estimates for the canonical variables are almost orthogonal to the desired true directions (cf.~the simulated curves in Fig.~\ref{fig_angles_theor_simul}).
 \item If both $\tau_K$ and $\tau_M$ become large while $\rho^2$ is fixed, then the condition $\rho^2>\rho_c^2$ becomes trivially true, $z_\rho\to \rho^2$ and $\s_x, \s_y\to 0$. Therefore, $z_\rho$, $\widehat \x$, and $\widehat \y$ become consistent estimates of $\rho^2$, $\x$ and $\y$, respectively.
 \item If $\tau_K$ becomes large while $\rho^2$ and $\tau_M$ remain fixed (this means that $M$ and $S$ are of the same magnitude but $K$ is much smaller), then $z_\rho$ remains larger than $\rho^2$, and $\s_y$ remains positive. Yet $\s_x$ tends to $0$, which means that $\widehat \x$ becomes a consistent estimate of $\x$.
\end{itemize}

\smallskip

The final essential aspect of Theorem \ref{Theorem_basic_setting} is the choice of the angles $\theta_x$ (between $\widehat \x$ and $\x$) and $\theta_y$ (between $\widehat \y$ and $\y$) as the measure of the quality of approximations of $\widehat\ba$ and $\widehat\bb$ for $\ba$ and $\bb$, respectively. One could have concentrated on the angle between $\widehat \ba$ and $\ba$ (or $\widehat \bb$ and $\bb$) instead. However, our choice offers significant advantages.

The angle between $\widehat \ba$ and $\ba$ (or $\widehat \bb$ and $\bb$) depends on the choice of units of measurement: the angle changes if we multiply one of the coordinates of $\ba$ by a constant (equivalently, divide a component of $\u$ by the same constant). Hence, asymptotic theory for these angles requires normalization conditions on the coordinates of $\u$ (correlations between different components of $\u$ are also important), but such natural normalizations are rare in real data, because different coordinates of $\u$ might be coming from very different sources. By concentrating on $\theta_x$ and $\theta_y$, we avoid this problem entirely, and in particular, the result in Theorem \ref{Theorem_basic_setting} does not depend on the covariance matrix of the vector $\u$ or covariance matrix of the vector $\v$, which can be arbitrary, as long as it satisfies Assumption \ref{ass_basic}.
See Appendix \ref{Section_appendix_x_vs_alpha} for a more detailed discussion.

\section{General framework}\label{Section_general}


In this section we state two generalizations of Theorem \ref{Theorem_basic_setting} by relaxing Gaussianity Assumption \ref{ass_basic_gauss} and allowing for multiple signals. In addition, we can allow for the signal \emph{or} for the noise to be correlated along the dimension $S$, as discussed in Appendix \ref{Section_cor_noise_and_signal}.

\subsection{Non-Gaussian data}
\label{Section_non_Gaussian}

\begin{definition} \label{Def_4moments}
We say that a mean zero random vector $(X_1,\dots,X_n)$ is fourth-moment Gaussian if there exists a jointly Gaussian mean zero vector $(Y_1,\dots,Y_n)$ such that the covariance matrix and the joint fourth moments of $(X_i)_{i=1}^n$ match those of $(Y_i)_{i=1}^n$.
\end{definition}
Definition \ref{Def_4moments} is equivalent to requiring the fourth joint moments of $(X_1,\dots,X_n)$ to satisfy the Wick rule: for any $1\le i,j,k,l\le n$, we should have
\vskip-0.8cm
\begin{equation}
\label{eq_Wick}
  \E X_i X_j X_k X_l = \E X_i X_j\cdot \E X_k X_l+ \E X_i X_k \cdot \E X_j X_l + \E X_i X_l \cdot \E X_j X_k.
\end{equation}
\vskip-0.3cm
For $n=1$, \eqref{eq_Wick} reduces to a single condition $\E X_1^4= 3 [\E X_1^2]^2$. Hence, many one-parametric families of mean $0$, variance $1$ distributions have fourth-moment Gaussian members. For $n=2$, there are five conditions: $\E X_1^4=3 [\E X_1^2]^2$, $\E X_1^3 X_2=3 \E X_1^2 \E X_1 X_2$, $\E X_1^2 X_2^2 = \E X_1^2 \E X_2^2 + 2 [\E X_1 X_2]^2$, $\E X_1 X_2^3= 3 \E X_1 X_2 \E X_2^2$, and $\E X_2^4=3 [\E X_2^2]^2$. For general $n$, many examples can be constructed by combining two observations. First, if mean $0$ random variables $X_1,\dots, X_n$ are independent and for each $1\le i \le n$, we have $\E X_i^4= 3 [\E X_i^2]^2$, then $(X_1,\dots,X_n)$ is fourth-moment Gaussian. Second, if $A$ is a deterministic $n\times n$ matrix and a random $n$--dimensional vector $\mathbf X$ is fourth moment Gaussian, then $A\mathbf X$ is also fourth moment Gaussian.

\smallskip

We now present a relaxation of Assumption \ref{ass_basic} on the random vectors $\u\in\mathbb R^K$ and $\v\in\mathbb R^M$.

\begin{assumption}\label{ass_4moments} There exists a deterministic vector $\ba\in \mathbb R^K$ and a deterministic $(K-1)\times K$ matrix $A$ of rank $K-1$; and a deterministic vector $\bb\in\mathbb R^M$ and a deterministic $(M-1)\times M$ matrix $B$ of rank $M-1$ such that
\begin{enumerate}
\item  The random variables $\u^\T\ba$ and $\v^\T\bb$ are jointly fourth-moment Gaussian with mean zero, as in Definition \ref{Def_4moments} with $n=2$.
\item  Denote $\tilde \u=A \u $ and $\tilde \v=B \v $, where $\tilde \u$ has coordinates $(\tilde u_i)_{i=1}^{K-1}$ and $\tilde \v$ has coordinates $(\tilde v_j)_{j=1}^{M-1}$. Then
    \begin{itemize}
      \item
      All components $\tilde u_1,\dots,\tilde u_{K-1}, \tilde v_1,\dots, \tilde v_{M-1}$ are jointly independent from each other and from both $\u^\T\ba$ and $\v^\T\bb$.
      \item $\E \tilde u_i=0$, $\E\tilde u_i^2=1$, $1\le i<K$, and $\E\tilde v_j=0$, $\E \tilde v_j^2=1$, $1\le j<M$.
      \item For constants $\kappa>0$ and $C>0$, $\sup_{i} \E |\tilde u_i|^{4+\kappa}<C$ and $\sup_{j} \E |\tilde v_j|^{4+\kappa}<C$.

     \end{itemize}
\end{enumerate}
\end{assumption}
Assumption \ref{ass_4moments} splits the vectors $\u$ and $\v$ into two components: the correlated signal part $(\u^\T\ba, \v^\T\bb)$ and the remaining noise part $(A\u, B\v)$. The coordinates of the latter are independent but not necessarily identically distributed.



\begin{theorem} \label{Theorem_4moments} If we replace Assumption \ref{ass_basic} with Assumption \ref{ass_4moments}, then Theorem \ref{Theorem_basic_setting} continues to hold with exactly the same conclusion.
\end{theorem}

\begin{figure}[t]
\begin{subfigure}{.55\textwidth}
  \centering
  \includegraphics[width=1.0\linewidth]{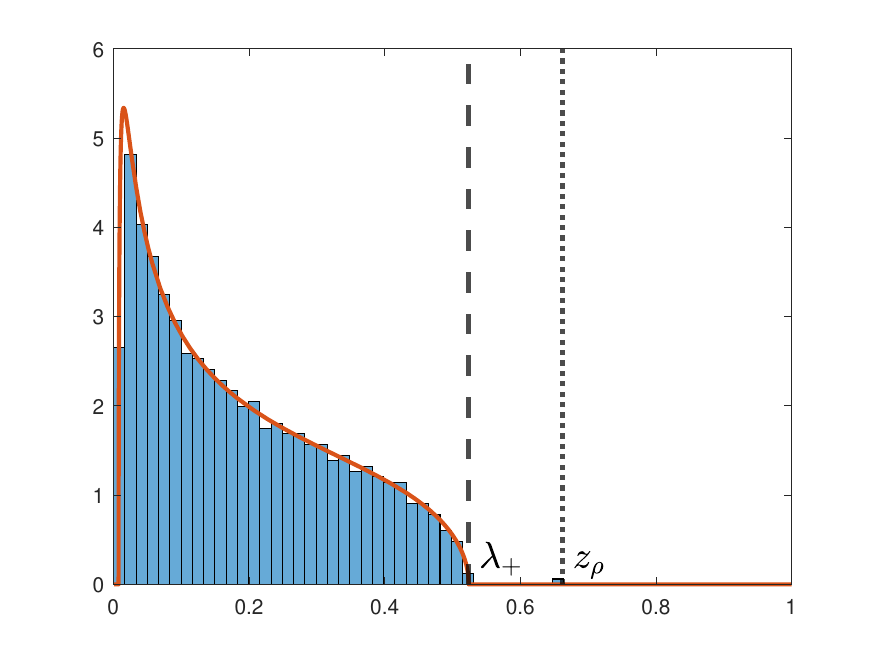}
  \caption{CCA eigenvalues and Wachter distribution.}
  \label{fig_uniform_CCA}
\end{subfigure}
\begin{subfigure}{.44\textwidth}
\centering
  \begin{tabular}{c|c|c}
  \multicolumn{3}{c}{}\\
  \multicolumn{3}{c}{}\\
  \multicolumn{3}{c}{Angles (in degrees)}\\
  \cline{2-3}
  \multicolumn{1}{c}{} & \small{Theoretical} & \small{Simulation}\\
  \hline
  \hline
  $\theta_x$ & $25.22$ & $25.40$\\
  $\theta_y$ & $28.39$ & $29.60$\\
  \hline
  \hline
  \multicolumn{3}{c}{}\\
  \multicolumn{3}{c}{}
\end{tabular}
\caption{Theoretical angles and corresponding angles from one simulation with one signal  and $U[-1,1]$ errors.}
\label{table_uniform_errors}
\end{subfigure}
\caption{One signal and $U[-1,1]$ errors, ${K=1000}$, $M=1500$, $S=8000$, $r^2=0.49$.}
\label{Fig_uniform_distr}
\end{figure}

\begin{figure}[t]
\begin{subfigure}{.55\textwidth}
  \centering
  \includegraphics[width=1.0\linewidth]{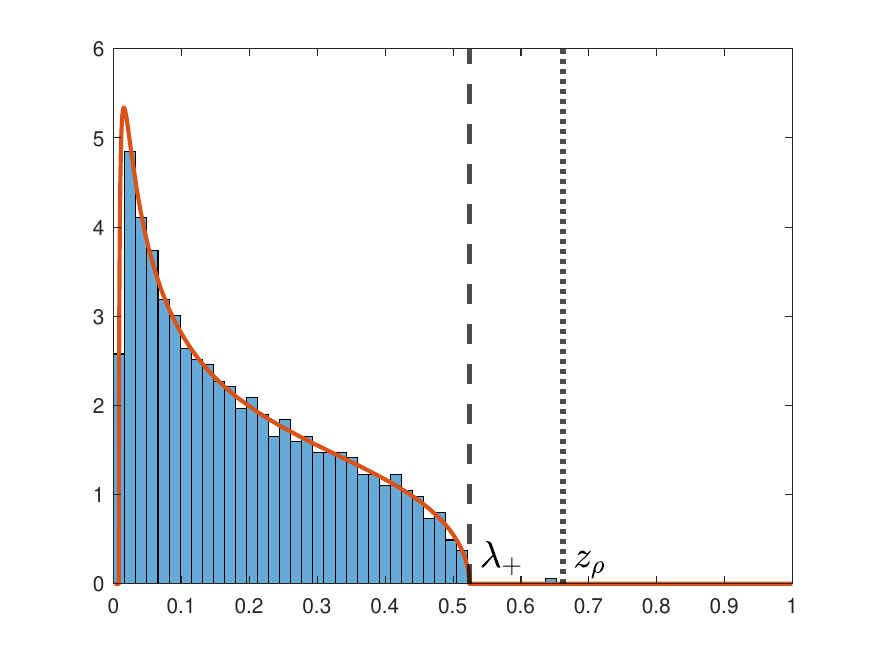}
  \caption{CCA eigenvalues and Wachter distribution.}
  \label{fig_tdistr_CCA}
\end{subfigure}
\begin{subfigure}{.44\textwidth}
\centering
  \begin{tabular}{c|c|c}
  \multicolumn{3}{c}{}\\
  \multicolumn{3}{c}{}\\
  \multicolumn{3}{c}{Angles (in degrees)}\\
  \cline{2-3}
  \multicolumn{1}{c}{} & \small{Theoretical} & \small{Simulation}\\
  \hline
  \hline
  $\theta_x$ & $25.22$ & $25.89$\\
  $\theta_y$ & $28.39$ & $29.84$\\
  \hline
  \hline
  \multicolumn{3}{c}{}\\
  \multicolumn{3}{c}{}
\end{tabular}
\caption{Theoretical angles and corresponding angles from one simulation with one signal and $t(3)$ errors.}
\label{table_student_t_errors}
\end{subfigure}
\caption{Illustration for the results with one signal and $t(3)$ errors, ${K=1000}$, $M=1500$, $S=8000$, $r^2=0.49$.}
\label{Fig_student_t_distr}
\end{figure}

The moment conditions in Assumption \ref{ass_4moments} are used in the proof of Lemma \ref{Lemma_asymptotic_approx} in the appendix: they simplify the derivations by insuring that some terms are uncorrelated and, therefore, satisfy the most basic form of the Law of Large Numbers. One can possibly relax this condition and use LLN for correlated data or come up with alternatives to our moment conditions (see \citet{yang2022limiting} for one possible approach). As shown in Fig.~\ref{Fig_uniform_distr}, uniform errors $U[-1,1]$, whose fourth moment does not coincide with the fourth moment of a Gaussian distribution, still lead to the same results as in Theorem \ref{Theorem_4moments}. Moreover, Fig.~\ref{Fig_student_t_distr} illustrates the same alignment between the theoretical results and a single simulation using  $t$ distribution with $3$ degrees of freedom for errors, where the third moment does not exist. 

We anticipate that the limits \eqref{eq_canonical_limit}, \eqref{eq_vector_limit1}, \eqref{eq_vector_limit2} remain the same under conditions weaker than Assumption \ref{ass_4moments}. However, we expect the size of the fluctuations around the limit to be sensitive to the fourth moments. For related results, see \citet{yang2022limiting} and \citet{muirhead1980asymptotic}, with the latter addressing the regime where $S\to\infty$ while $K$ and $M$ remain fixed. Appendix \ref{Section_appendix_MC} contains simulations exploring non-Gaussian extensions of Theorem \ref{Theorem_joint_CLT_main}.

\subsection{Multiple signals}


Theorems \ref{Theorem_basic_setting} and \ref{Theorem_4moments} assume that there is a unique signal in the $\U$ part of the data and a unique signal  in the $\V$ part of the data. They have extensions to the situation of several signals. In the extended statements exactly the same procedures are used for each signal. We start from the setup of Section \ref{Section_basic_sample_setting}. The population setting of Assumption \ref{ass_basic} says that there is one nonzero canonical correlation\footnote{In the population setting, squared canonical correlations can be computed as eigenvalues of the $K\times K$ matrix $(\E \u \u^\T)^{-1} (\E \u \v^\T) (\E \v \v^\T)^{-1} (\E \v \u^\T)$.} between $\u$ and $\v$, and it corresponds to the canonical variables $\u^\T\ba$ and $\v^\T\bb$. Instead, assume that there are $\mathbbm q$ nonzero canonical correlations:

\begin{assumption}\label{ass_basic_multi} The vectors $\u$ and $\v$ satisfy the following:
\begin{enumerate}
\item\label{ass_basic_gauss_multu} The vectors $\u$ and $\v$ are jointly Gaussian with mean zero.
\item\label{ass_basic_indep_set_multi} There exist $\mathbbm q$ nonzero deterministic vectors $\ba^1,\dots,\ba^\mathbbm q\in\mathbb{R}^K$ and $\mathbbm q$ nonzero deterministic vectors $\bb^1,\dots,\bb^\mathbbm q\in\mathbb{R}^M$ such that
\begin{enumerate}
  \item $\E (\u^\T\ba^q)(\u^\T\ba^{q'})=\E (\v^\T\bb^q)(\v^\T\bb^{q'})=\E (\u^\T\ba^q)(\v^\T\bb^{q'})=0$ for each $q\ne q'$.\footnote{Any $2\mathbbm q$ vectors $\ba^1,\dots,\ba^\mathbbm q\in\mathbb{R}^K$,  $\bb^1,\dots,\bb^\mathbbm q\in\mathbb{R}^M$ can be linearly transformed to satisfy this condition, cf.~Lemma \ref{Lemma_canonical_bases}.}
  \item Let $r^2[q]$ denote the squared correlation coefficient between $\u^\T\ba^q$ and $\v^\T \bb^q$, as in Definition  \ref{Def_squared_correlation_coef}. We assume that these numbers are all distinct.
  \item For any $\bg\in\mathbb{R}^K$, if $\u^\T\bg$ is uncorrelated with all  $\u^\T\ba^q$, $1\le q\le \mathbbm q$, then $\v$ and $\u^\T\bg$ are also uncorrelated;
  \item For any $\bg\in\mathbb{R}^M$, if $\v^\T\bg$ is uncorrelated with all $\v^\T\bb^q$, $1\le q\le \mathbbm q$, then $\u$ and $\v^\T\bg$ are also uncorrelated.
\end{enumerate}
\end{enumerate}
\end{assumption}
\begin{example}
\emph{Assumption \ref{ass_basic_multi} is satisfied with $\ba^q$ being the $q$th coordinate vector in $\mathbb R^K$ and $\bb^q$ being the $q$th coordinate vector in $\mathbb R^M$, $1\le q \le \mathbbm q$, if $(u_1,\dots,u_K,v_1,\dots,v_M)^\T$ is a mean zero Gaussian vector such that the only possible nonzero correlations are between $u_q$ and $v_q$ for $1\le q \le \mathbbm q$, between $u_{k}$ and $u_{k'}$ for $\mathbbm q+1\le k,k'\le K$, and between $v_m$ and $v_{m'}$ for $\mathbbm q+1\le m,m'\le M$. Additionally, assume that the squared correlation coefficients between $u_q$ and $v_q$, $1\le q\le \mathbbm q$, are distinct.} \emph{Any other example can be obtained by a change of basis.}
\end{example}

As in Section \ref{Section_basic_sample_setting}, we let $\W=\bigl(\U^\T, \V^\T\bigr)^\T$ be a $(K+M)\times S$ matrix composed of $S$ independent samples of $\bigl(\u^\T, \v^\T\bigr)^\T$. The squared sample canonical correlations $\lambda_1\ge \lambda_2\ge \dots\ge \lambda_K$ are the eigenvalues of the $K\times K$ matrix $(\U \U^\T)^{-1} \U \V^\T (\V \V^\T)^{-1} \V \U^\T$  or $K$ largest eigenvalues of the $M\times M$ matrix $(\V \V^\T)^{-1} \V \U^\T (\U \U^\T)^{-1} \U \V^\T$. Choosing one eigenvalue $\lambda_i$, we set $\widehat \ba^i$ to be the corresponding eigenvector of the former matrix, and we set $\widehat \bb^i$ to be the eigenvector of the latter matrix, corresponding to the same eigenvalue. The sample canonical variables are defined as $\widehat \x^i=\U^\T \widehat \ba^i$ and $\widehat \y^i =\V^\T \widehat \bb^i$. We also set $\x^q=\U^\T \ba^q$ and $\y^q =\V^\T \bb^q$, where $q$ is chosen so that $\rho^2[q]$ is the $q$-th largest elements of $\rho^2[1],\dots,\rho^2[\mathbbm q]$. 

\begin{theorem} \label{Theorem_basic_multi}
 Suppose that Assumption \ref{ass_basic_multi} holds and the columns of the data matrix $\W$ are i.i.d. Let $S$ and $K\le M$ tend to infinity in such a way that the ratios $S/K$ and $S/M$ converge to $\tau_K>1$ and $\tau_M>1$, respectively, and $\tau_M^{-1}+\tau_K^{-1}<1$. Suppose also $\lim_{S\to\infty} r^2[q]=\rho^2[q]$, $1\le q \le \mathbbm q$, with the numbers $\rho^2[1],\dots,\rho^2[\mathbbm q]$ being distinct. Then, for each $q=1,\dots,\mathbbm q$, 
 \begin{enumerate}
  \item[\bf I.] If $\rho^2$ is the $q$-th largest number from $\rho^2[1],\dots, \rho^2[\mathbbm q]$ and $\rho^2> \rho_c^2$, then $z_{\rho}>\lambda_+$ and
  \vskip-1cm
  \begin{equation}
  \label{eq_canonical_limit_multi}
   \lim_{S\to\infty} \lambda_q=z_{\rho}, \qquad \text{in probability}.
  \end{equation}
  The squared sine of the angle $\theta_x$ between the corresponding $\x^q$ and $\widehat \x^q$ satisfies
  \vskip-1cm
  \begin{equation}
  \label{eq_vector_limit1_multi}
   \lim_{S\to\infty} \sin^2\theta_x=\s_x, \qquad \text{in probability}.
  \end{equation}
  The squared sine of the angle $\theta_y$ between the corresponding $\y^q$ and $\widehat \y^q$ satisfies
  \vskip-1cm
  \begin{equation}
  \label{eq_vector_limit2_multi}
   \lim_{S\to\infty} \sin^2\theta_y=\s_y, \qquad \text{in probability}.
  \end{equation}
  The values of $\lambda_+$, $z_{\rho}$, $\s_x$, and $\s_y$ are obtained through \eqref{eq_lambda_plus}, \eqref{eq_zrho}, \eqref{eq_sx}, and \eqref{eq_sy}.

  If we additionally assume \eqref{eq_params_exact_almost_main}, then the corrections to the limits \eqref{eq_canonical_limit_multi},
  \eqref{eq_vector_limit1_multi}, \eqref{eq_vector_limit2_multi} are Gaussian, with asymptotic covariance as in Theorem \ref{Theorem_joint_CLT_main} and limits independent over $q$.

  \item[\bf II.] If the $q$-th largest number from $\rho^2[1],\dots, \rho^2[\mathbbm q]$ is at most $\rho_c^2$, then
  \vskip-1cm
  $$\lim_{S\to\infty} \lambda_q=\lambda_+, \qquad \text{in probability}.$$
 \end{enumerate}
 Finally, $\lim_{S\to\infty} \lambda_{\mathbbm q+1}=\lambda_+.$
\end{theorem}

\begin{figure}[t]
\begin{subfigure}{.55\textwidth}
  \centering
  \includegraphics[width=1.0\linewidth]{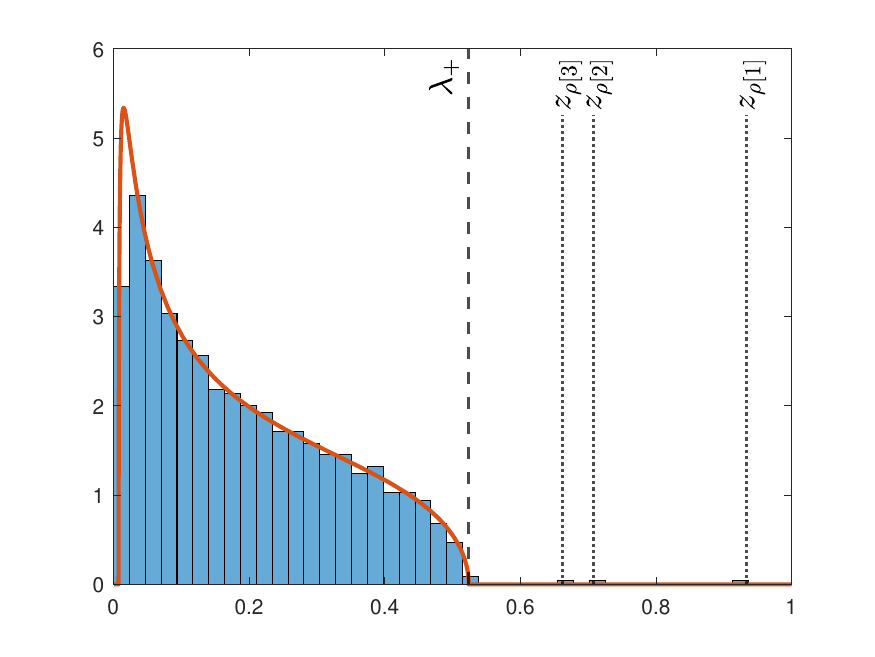}
  \caption{CCA eigenvalues and Wachter distribution.}
  \label{fig_multi_spikes_hist}
\end{subfigure}
\begin{subfigure}{.44\textwidth}
\centering
  \begin{tabular}{c|c|c}
  \multicolumn{3}{c}{}\\
  \multicolumn{3}{c}{}\\
  \multicolumn{3}{c}{Angles (in degrees)}\\
  \cline{2-3}
  \multicolumn{1}{c}{} & \small{Theoretical} & \small{Simulation}\\
  \hline
  \hline
  $\theta_x^1$ & $7.57$ & $7.82$\\
  $\theta_x^2$ & $21.37$ & $21.22$\\
  $\theta_x^3$ & $25.22$ & $25.57$\\
  \hline
  $\theta_y^1$ & $8.89$ & $9.02$\\
  $\theta_y^2$ & $24.31$ & $23.93$\\
  $\theta_y^3$ & $28.39$ & $27.79$\\
  \hline
  \hline
  \multicolumn{3}{c}{}\\
  \multicolumn{3}{c}{}
\end{tabular}
\caption{Theoretical angles from Theorem \ref{Theorem_basic_multi} and corresponding angles from one simulation  with three signals and Gaussian errors.}
\label{table_multi_spikes}
\end{subfigure}
\caption{ Theorem \ref{Theorem_basic_multi} for $K=1000,\,M=1500,\,S=8000$, $r[1,2,3]=0.95,\,0.75,\,0.7$.}
\label{fig_multi_spikes}
\end{figure}

Fig.~\ref{fig_multi_spikes} illustrates Theorem \ref{Theorem_basic_multi}: there are three signals of different strength indicated by the three separate right-most eigenvalues. The corresponding angles increase when the strength of the signal  goes down (smaller $\rho^2$), as can be seen from Table \ref{table_multi_spikes}. Note also that the theoretical formulas are close to what one obtains in the simulations.

Theorem \ref{Theorem_basic_multi} helps to find the number of signals when it is unknown. The eigenvalues to the right of $\lambda_+$ represent signals, and one can deduce\footnote{One should be able to construct a formal test of the number of signals, e.g., by using CLT-type results of \citet{bao2019canonical}, \citet{yang2022limiting}. In the setting of factors a parallel question generated considerable interest, see, e.g., \cite{bai2002determining, bai2007determining,hallin2007determining,onatski2009testing,ahn2013eigenvalue}. } how many there are by looking at a histogram (e.g, three signals in Fig.~\ref{fig_multi_spikes}).
In practice we do not know the values of $\rho^2[q]$ a priori. Instead, we should look at the squared sample canonical correlations $\lambda_1\ge \lambda_2\ge\dots$: if several of the largest ones are well separated from $\lambda_+$ and from each other, then we can use Theorem \ref{Theorem_basic_multi}, reconstruct the corresponding $\rho^2[q]$ and deduce the values for the angles $\theta^q_x$ and $\theta^q_y$.

The requirement that all $\rho^2[q]$ are distinct is not a technical artifact: if two canonical correlations coincide, then the corresponding canonical variables are no longer well defined because any linear combination of two eigenvectors with the same eigenvalue is again an eigenvector with the same eigenvalue. On the quantitative side, let $\delta$ denote the minimal spacing between $r^2[q]$, $1\le q \le \mathbbm q$. Theorem \ref{Theorem_basic_multi} holds as long as $\delta$ stays bounded away from $0$ as $S\to\infty$. Analyzing the proofs, one sees that small $\delta$ leads to an additional error term of order $\frac{1}{S \delta^2}$ (the largest error comes from $1/(z-c_i^2)^2$ in \eqref{eq_CCA_alpha}, \eqref{eq_CCA_beta} in Appendix \ref{Section_CCA_master}). Hence, our proof of \eqref{eq_canonical_limit_multi}, \eqref{eq_vector_limit1_multi}, \eqref{eq_vector_limit2_multi} relies on $\frac{1}{S \delta^2}=o(1)$, while Gaussian corrections require $\frac{1}{S \delta^2}=o(S^{-1/2})$.

\smallskip

Theorem \ref{Theorem_4moments} has a similar extension to the case of $\mathbbm q$ nontrivial canonical correlations: the data are allowed to be fourth-moment Gaussian rather than Gaussian, and the final statement is as in Theorem \ref{Theorem_basic_multi}, with the relations for each $q=1,2,\dots,\mathbbm q$ are exactly the same as for the $\mathbbm q=1$ situation. Note, however, that our proof of Gaussian corrections in Theorem \ref{Theorem_joint_CLT_main}, and therefore, also of its multiple signals extension, requires Gaussianity of the data.

\section{Empirical illustrations}\label{Section_empirical}

In this section we provide an empirical implementation of our approach in two data sets.

\subsection{Cyclical vs.~non-cyclical stocks} \label{Section_stocks}

How to model and explain stock returns has always been an important question in finance and financial econometrics. There are many different approaches and techniques, including asset pricing models, volatility models, and dimension-reduction factor models. This section adds CCA to the list and shows how this method can be used to analyze stock returns. Our illustration is in parallel with the literature on maximally predictable portfolio, which originated with \citet{lo1997maximizing} and aims at predicting some combination of stocks and bonds from a combination of observable economic variables. Canonical variables in CCA can serve as such portfolio. Along these lines, recently \citet{firoozye2023canonical} applied CCA to current and past returns to construct a predictable portfolio.

Compared to the preceding, we do not distinguish between outcome variables and predictors and treat them symmetrically. We use two corpuses of data and concentrate on consumer cyclical and non-cyclical (consumer defensive) stocks. The former are known to follow the state of the economy, while the latter are not related to business cycles and are useful during economic slowdowns. Knowing their correlations is helpful for portfolio allocation. We apply CCA and search for maximally correlated combinations of cyclical and non-cyclical stocks.



\begin{figure}[t]
  \centering
\begin{subfigure}{0.7\textwidth}
  \includegraphics[width=0.75\linewidth]{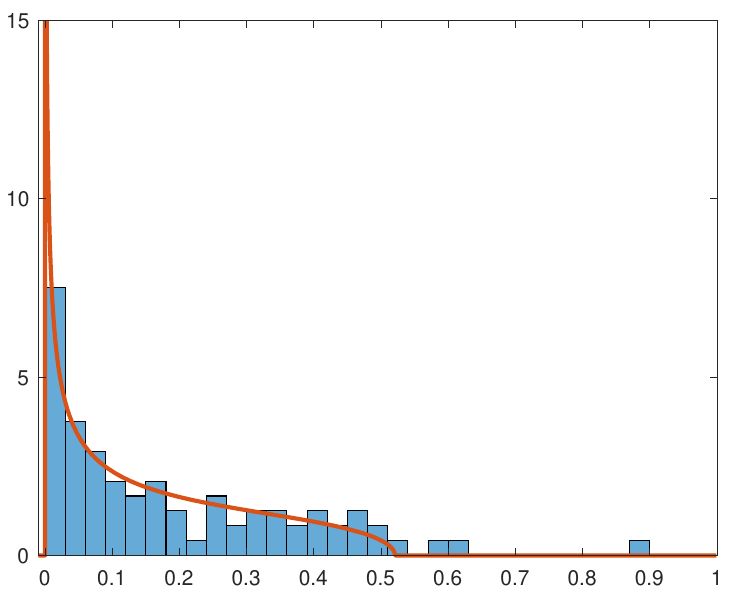}
  \caption{Squared sample canonical correlations between cyclical and non-cyclical stocks (blue columns) and Wachter distribution (orange curve).}
  \label{stocks_hist}
\end{subfigure}
\begin{subfigure}{1.0\textwidth}
\centering
  \begin{tabular}{c|c|c|c|c|c}
  \multicolumn{6}{c}{}\\
  \multicolumn{6}{c}{}\\
   & $z_{\rho}$ & $\rho^2$ & $|\rho|$ & Angle (in degrees) & Sine squared\\
  \hline
  \hline
  $1$st signal & $0.89$ & $0.84$ & $0.92$ & $10.8$ & $0.03$\\
  $2$nd signal & $0.62$ & $0.43$ & $0.65$ & $31.0$ & $0.27$\\
  $3$rd signal & $0.58$ & $0.34$ & $0.59$ & $38.2$ & $0.38$\\
  \hline
  \hline
\end{tabular}
\caption{Strength of the signals and respective vector estimation precision.}
\label{table_stocks}
\end{subfigure}
\caption{Results of CCA performed on cyclical vs.~non-cyclical stocks.}
\label{stocks_pic}
\end{figure}

We use weekly returns (which are widely believed to be uncorrelated across time) for the $80$ largest cyclical and $80$ largest non-cyclical stocks over ten years ($01.01.2010-01.01.2020$), which gives us $521$ observations across time. The time window is chosen to exclude the $2007-2008$ financial crisis and COVID-19 periods. The exact list of stocks is reported in Appendix \ref{Section_appendix_data}. Fig.~ \ref{stocks_hist} shows the squared sample canonical correlations between the returns of the two groups of stocks. We can see that the overall shape of the histogram of the correlations matches the Wachter distribution (as discussed in Section \ref{sec_implic_basic}, this indicates that the data matches our modelling assumptions and we can use the formulas of Theorem \ref{Theorem_basic_multi}) and, at the same time, that the three largest correlations are clearly separated from the rest. The largest squared sample canonical correlation, which is close to $0.9$, represents a very strong signal, which can be interpreted as an overall market movement. In favor of this interpretation is also the fact that the majority of the coordinates in the corresponding pair of eigenvectors are positive, though the eigenvectors are not reported here for the sake of space. For the cyclicals, the sum of squares of the positive coordinates is $0.85$ out of $1$, and for the noncyclicals, it is $0.78$ out of $1$. Two other correlations represent medium-strength signals. The presence of exactly three signals is interesting for several reasons. First, it is resonant of the three factors of \citet{fama1992cross}. Second, it is in noticeable contrast to the results of a principal component analysis (PCA) applied to the same stock returns. PCA often shows few or even only one clearly distinct factor for a large group of stocks; see, e.g., \citet[20.4.2 A One-Factor Model]{potters2020first}. For our two groups of stocks, PCA indicates the presence of at most two factors---fewer than what we detect with CCA. Fig.~\ref{fig_PCA_stocks} in Appendix \ref{Section_appendix_data} shows the PCA eigenvalues, i.e., eigenvalues of $(\V \V^\T)/S$ and $(\U \U^\T)/S$, where $\V$ and $\U$ are the de-meaned cyclical and non-cyclical returns, respectively. We can see that the cyclicals have only one strong factor while the non-cyclicals have two. As with CCA, the largest eigenvalue in PCA is believed to be market related. In our stock returns data set, the largest factors are highly correlated across the two groups (the correlation is $0.83$) and with their CCA counterparts (the first pair of canonical variables). The two other CCA pairs of canonical variables are capturing additional information (cf.~the size and value factors in \citet{fama1992cross}) not reflected by PCA. Based on the above, we observe that CCA picks up more information than PCA by extracting additional relevant factors.

\begin{figure}[t]
\begin{subfigure}{.32\textwidth}
  \centering
  \includegraphics[width=1.0\linewidth]{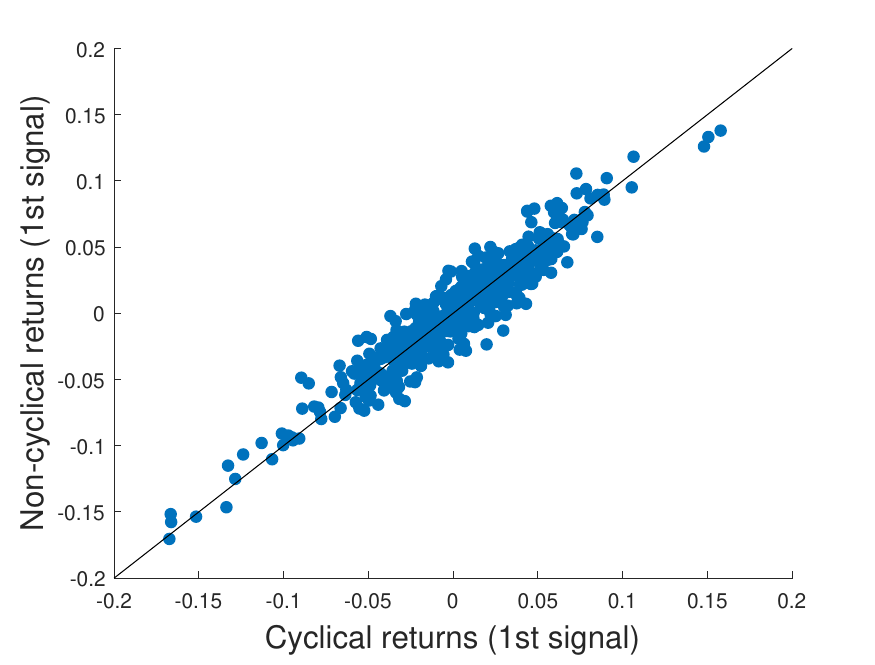}
  \caption{First signal.}
  \label{fig_scatter_stocks_s1}
\end{subfigure}%
\begin{subfigure}{.32\textwidth}
  \centering
  \includegraphics[width=1.0\linewidth]{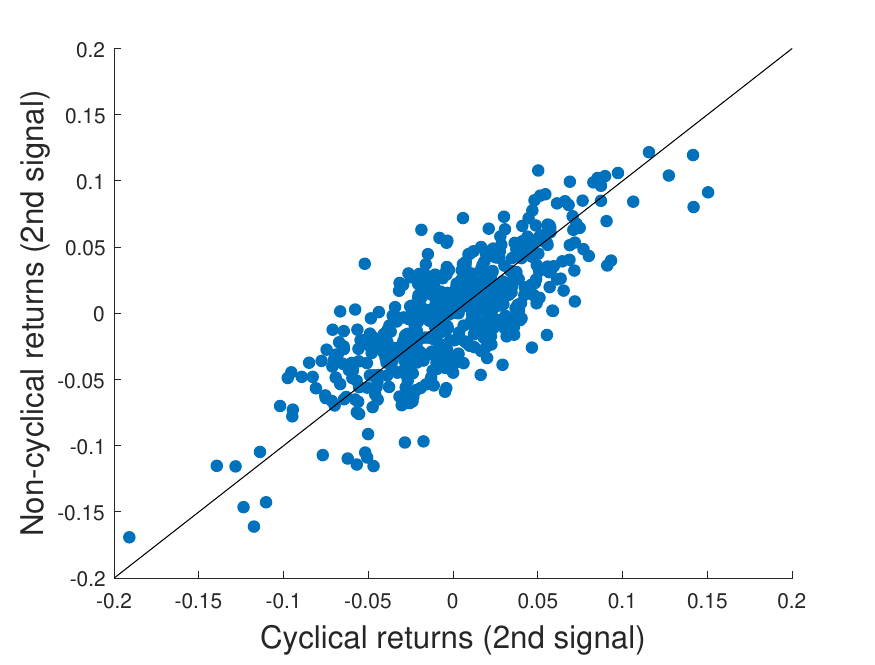}
  \caption{Second signal.}
  \label{fig_scatter_stocks_s2}
\end{subfigure}
\begin{subfigure}{.32\textwidth}
  \centering
  \includegraphics[width=1.0\linewidth]{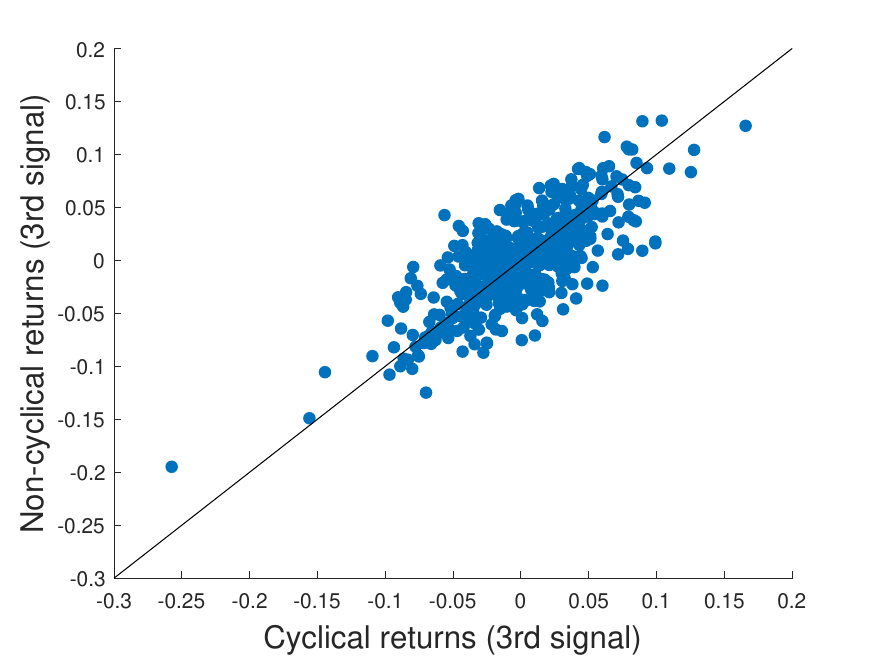}
  \caption{Third signal.}
  \label{fig_scatter_stocks_s3}
\end{subfigure}
\caption{Scatter plot of $(\V^\T \widehat \bb^i,\,\U^\T \widehat \ba^i)$, where $\V$ is composed of cyclical stock returns and $\U$ of non-cyclical stock returns and $i=1,2,3$ is the number of the signal. The length of $\V^\T \widehat \bb^i$ and $\U^\T \widehat \ba^i$ is normalized to $1$.}
\label{fig_scatter_stocks}
\end{figure}

The exact values of the CCA signals are reported in Table \ref{table_stocks}. The correlations and angles in Table \ref{table_stocks} are calculated based on Theorem \ref{Theorem_basic_multi}. Since we have $K=M=80$, the angles are the same for the cyclicals and non-cyclicals, and we report only one angle per signal. In line with the strength of the signals, the angle between the first estimated signal and the unobserved truth is the smallest, showing precise estimation. The two other signals are weaker, leading to larger angles, i.e., large estimation error. This observation is reinforced by Fig.~\ref{fig_scatter_stocks}, which shows a scatter plot of the cyclical vs.~non-cyclical canonical variables.  The points in Fig.~\ref{fig_scatter_stocks_s1} show the tightest fit to the $45$-degree line, demonstrating a strong correlation, while those in Fig.~\ref{fig_scatter_stocks_s2} and \ref{fig_scatter_stocks_s3} are more spread out. The tightness in Fig.~\ref{fig_scatter_stocks_s2} and \ref{fig_scatter_stocks_s3} is approximately the same, in line with the second- and third-largest eigenvalues being very close.

\subsection{The limestone grassland community data}\label{Section_empirical_limestone}

\begin{table}[t]
	\begin{tabular}{c|c|c|c}
		\hline
		 & Eigenvector  & Angle & Sine\\
        & & (in degrees) & squared\\
		\hline
		\hline
		$\widehat \ba,\,\theta_x,\,\s_x$ & $(-0.30,-0.49,-0.80,-0.08,0.15,0.09)$ & $14.21$ & $0.06$\\
        $\widehat \bb,\,\theta_y,\,\s_y$ & $(0.77,0.51,0.16,0.10,-0.10,0.15,0.12,-0.25)$ & $15.77$ & $0.07$\\
		\hline
        \hline
	\end{tabular}
	\caption{Limestone grassland community: CCA. Eigenvectors are obtained from the data, while angles are calculated based on Theorem \ref{Theorem_basic_setting}.}	
  \label{table_limestone}
\end{table}

\begin{figure}[t]
    \centering
        \includegraphics[width=0.5\textwidth]{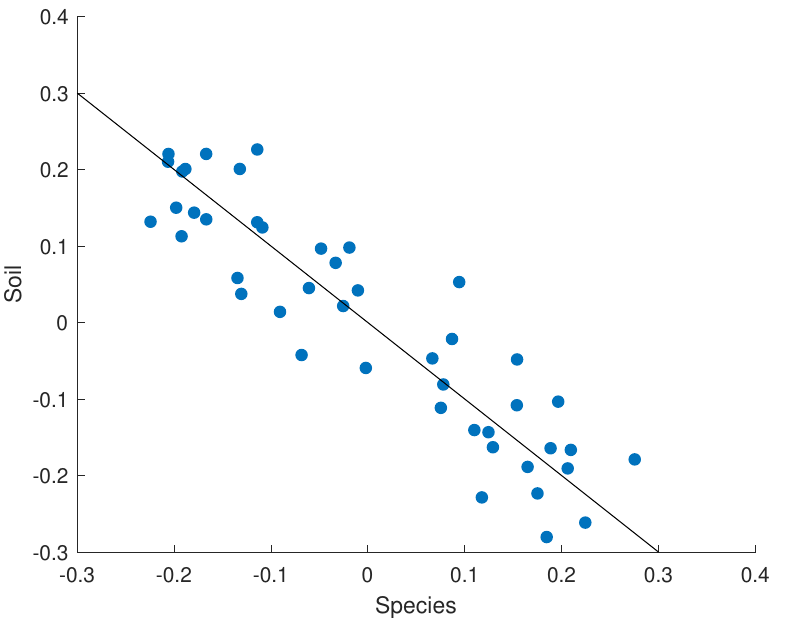}
    \caption{Scatter plot of $(\V^\T \widehat \bb,\,\U^\T \widehat \ba)$, where $\V$ is composed of plant species and $\U$ of soil properties. The length of $\V^\T \widehat \bb^i$ and $\U^\T \widehat \ba^i$ is normalized to $1$.} \label{fig_scatter_limestone}
\end{figure}

A classic data sets examined by CCA is the limestone grassland community data from Anglesey, North Wales (\citet{gittins1985canonical}), discussed recently in \cite{bao2019canonical}. The goal of the CCA in this case is to identify the relationship between soil properties and the plant species. There are $M=8$ species considered---\textit{Helictotrichon pubescens}, \textit{Trifolium pratense}, \textit{Poterium sanguisorba}, \textit{Phleum bertolonii}, \textit{Rhytidiadelphus squarrosus}, \textit{Hieracium pilosella}, \textit{Briza media}, and \textit{Thymus drucei}---and $K=6$ soil characteristics---depth ($d$), extractable phosphate ($P$), exchangeable potassium ($K$), and cross-product terms between all three soil variables ($d\times P,\, d\times K,\, P\times K$). The number of samples is $S=45$. Although the dimensions seem small, Monte Carlo simulations in Appendix \ref{Section_appendix_MC} indicate that Theorem \ref{Theorem_basic_setting} gives a reasonable approximation for these values: the theoretical and simulated angles are close to each other.

To analyze the data, we first de-mean all the variables (the means are calculated across $S$-space). After that we find that the six nonzero squared canonical correlations are $0.83,0.52,0.36,0.11,0.09,0.04$, and we have $\lambda_{+}=0.53$, which indicates the presence of a rank $1$ signal. This signal has estimated strength $|\rho|=0.86$ or $\rho^2=0.75$. As can be seen from Table \ref{table_limestone}, the estimation precision is quite high; i.e., angles $\theta_x$ and $\theta_y$ are small (the angles are obtained via Theorem \ref{Theorem_basic_setting}). The estimated vectors $\widehat \ba$ and $\widehat \bb$ suggest that most of the correlation comes from the first coordinate of $\widehat \bb$ (\textit{Helictotrichon pubescens}) and the third coordinate of $\widehat \ba$ (potassium). That is, more potassium in the soil is strongly correlated with the presence of \textit{Helictotrichon pubescens}. Fig.~\ref{fig_scatter_limestone} shows a scatter plot of $(\U^\T \widehat \ba,\,\V^\T \widehat \bb)$. As we can see, the data points lie close to the $-45$-degree line, which indicates a correlation close to $-1$.

\section{Conclusion}
\label{Section_conclusion}

High-dimensional econometrics and statistics have been playing a prominent role in the analysis of data in social sciences (see, e.g., \citet{fan2011sparse}). A popular method for dealing with two large data sets is canonical correlation analysis. Yet, little is known about the performance of CCA in the high-dimensional setting where all dimensions of data are large and comparable. This paper provides the results on the precision of CCA in high dimensions: we show that estimation of canonical vectors is inconsistent and compute the exact degree of inconsistency in terms of angles between the true and estimated canonical variables. Consistency is restored as one of the dimensions becomes much larger than the others.

There are various directions for future research on high-dimensional CCA. The first important generalization is to allow for for a setting where both the signal and the noise are simultaneously dependent across the $S$ dimension. This extension will make the results useful for a typical time series setting, where the coordinate $S$ represents time.
An interesting application would be to then use the procedure on the FRED-MD macro time series data set (\citet{mccracken2016fred}) and look for canonical correlations between real and nominal variables. The traditional approach to searching for factors in the FRED-MD data is principal component analysis. However, this approach does not use any additional information on the type of each variable; real and nominal variables are mixed together. CCA applied to a subset of real variables vs.~a subset of nominal variables can help provide another ``common'' factor.

An important direction for future research is to extend the second-order theory for the angles $\theta_x$ and $\theta_y$ beyond the
Gaussian setting.  Finally, a more refined analysis of what is happening near the identification boundary $\rho^2> \rho_c^2$ in the spirit of local-to-unit root limit theory is interesting, cf.\ \citet{bykhovskaya2025weak}.

\appendix

\section{Correlated signal and noise}

\label{Section_cor_noise_and_signal}

 This section is an addition to Section \ref{Section_general} of the main text. Here we relax the assumption that the observed data matrices $\U$ and $\V$ have i.i.d.\ columns. In Section \ref{Section_cor_signal} we allow for the signals $\u^\T\ba$ and $\v^\T\bb$ that are correlated along the dimension $S$. In Section \ref{Section_Correlated_noise} we  allow for the noise  (orthogonal to $\u^\T\ba$ and $\v^\T\bb$) that is correlated along the dimension $S$. While it is plausible that additional extensions might be feasible, we currently do not see a way to allow for correlated signals and noise simultaneously. This limitation is inherent to our approach, as the asymptotic concentration in \eqref{eq_canonical_limit}, \eqref{eq_vector_limit1}, and \eqref{eq_vector_limit2} relies on a form of the law of large numbers. To achieve this we require at least some degree of independence --- or, at a minimum, a decay of correlations along the $S$ dimension.

\subsection{Correlated signal} \label{Section_cor_signal} We adjust the procedure of Section \ref{Section_basic_sample_setting} and no longer assume that the data are obtained from i.i.d.\ samples of $\u$ and $\v$ vectors.

We still use Definition \ref{Definition_sample_setting}: we are given a $K\times S$ matrix $\U$ and $M\times S$ matrix $\V$, and we compute their sample canonical correlations and corresponding variables. What changes is the probabilistic mechanism creating these matrices: while in Section \ref{Section_setting} we started from vectors $\u$ and $\v$ and then considered i.i.d.\ samples of them, now we weaken the i.i.d.\ assumption and therefore have to introduce assumptions directly on the distributions of $\U$ and $\V$ rather than on $\u$ and $\v$. Despite the more complicated probabilistic setting, the final interpretation remains the same: the signal part of the data comes from two deterministic vectors $\ba\in\mathbb R^K$ and $\bb\in\mathbb R^M$, and we are trying to reconstruct them using $\widehat \ba$ and $\widehat \bb$ of Definition \ref{Definition_sample_setting}.

Our assumptions are now phrased in terms of $K\times S$ and $M\times S$ matrices $\U$ and $\V$.

\begin{assumption}\label{ass_cor_signal}There exists a deterministic vector $\ba\in \mathbb R^K$ and a deterministic  ${(K-1)\times K}$ matrix $A$ of rank $K-1$, and  a deterministic vector $\bb\in\mathbb R^M$ and a deterministic $(M-1)\times M$ matrix $B$ of rank $M-1$,  such that the $(M+K-2)S$ matrix elements of the matrices $A \U$ and $B\V$ are i.i.d.\  $\mathcal N(0,1)$ random variables independent from $\x=\U^\T \ba$ and $\y =\V^\T \bb$.

\end{assumption}

Note that there are no restrictions on $\x=\U^\T \ba$ and $\y =\V^\T \bb$ in Assumption \ref{ass_cor_signal}. E.g., they are allowed to have correlated or even deterministic coordinates. In the setting of Assumption  \ref{ass_cor_signal}, we replace the squared correlation coefficient $r^2$ of Definition \ref{Definition_sample_setting}  with its sample version:
\vskip-0.5cm
\begin{equation}
\label{eq_sample_r}
     \hat r^2:=\frac{ (\x^\T \y)^2}{ (\x^\T \x) (\y^\T\y)}.
\end{equation}
\vskip-0.2cm

As before, we treat vectors $\x$ and $\y$ as the signal part of the data and the rest as the noise. Note that if $\x$ and $\y$ have i.i.d.\ Gaussian components, then Assumption \ref{ass_cor_signal} coincides with Assumption \ref{ass_basic} (up to a linear transformation of the noise). In this situation $\hat r^2$ differs from $r^2$ from Definition \ref{Definition_sample_setting}, but their difference tends to $0$ as $S\to\infty$ by the law of large numbers.

\begin{theorem} \label{Theorem_non_iid_signal}
 Suppose that the squared sample canonical correlations and variables are constructed as in Definition \ref{Definition_sample_setting} with data matrices $\U$ and $\V$ satisfying Assumption \ref{ass_cor_signal}. Let $S$ tend to infinity and $K\le M$ depend on it in such a way that the ratios $S/K$ and $S/M$ converge to $\tau_K>1$ and $\tau_M>1$, respectively, and $\tau_M^{-1}+\tau_K^{-1}<1$. Simultaneously, suppose that $\lim_{S\to\infty} \hat r^2=\rho^2$. Then, the conclusions of Theorem \ref{Theorem_basic_setting} continue to hold.
\end{theorem}

The practical implementation of the procedure discussed at the beginning of Section \ref{sec_implic_basic} continues to hold in the setting of Theorem \ref{Theorem_non_iid_signal}. Theorem \ref{Theorem_non_iid_signal} uses the Gaussianity of the noise part of the data, $A\U$ and $B\V$. While it is plausible that this restriction can be relaxed, we do not address such a question in this paper.

Extension of Theorem \ref{Theorem_non_iid_signal} to multiple signals works in the same way as in Theorem \ref{Theorem_basic_multi}. We have $2\mathbbm q$ vectors $\x^q= \U^\T \ba^q$, $\y^q=\V^\T \bb^{q}$, $1\le q\le \mathbbm q$, which represent $\mathbbm q$ ``true'' canonical variables (and therefore $\x^q$, $\y^{q'}$ are pairwise orthogonal except for the allowed correlation when $q=q'$, which should result in distinct correlation coefficients as we vary $q$). The final statement is as in Theorem \ref{Theorem_basic_multi}: the relations for each $q=1,2,\dots,\mathbbm q$ are exactly the same as for the $\mathbbm q=1$ situation.

\subsection{Correlated noise}
\label{Section_Correlated_noise}

In this section we complement the correlated signal setting of Section \ref{Section_cor_signal}. The results for the correlated noise depend on the knowledge of the canonical correlations of the noise itself, and the formulas become much more complicated than in Theorems \ref{Theorem_basic_setting}, \ref{Theorem_4moments}, and \ref{Theorem_non_iid_signal}. Nevertheless, one can still efficiently use them when working with the data. Our assumptions are again phrased in terms of $K\times S$ and $M\times S$ matrices $\U$ and $\V$, respectively.

\begin{assumption}\label{ass_cor_noise} There exists  a deterministic vector $\ba\in \mathbb R^K$ and a deterministic  $(K-1)\times K$ matrix $A$ of rank $K-1$, and a deterministic vector $\bb\in \mathbb R^M$ and a deterministic $(M-1)\times M$ matrix $B$ of rank $M-1$,  such that:
\begin{enumerate}[label=(\arabic*), ref=\ref{ass_cor_noise}.(\arabic*)]
\item We set $\x=\U^\T \ba$ and $\y =\V^\T \bb$ and assume that $(\x,\y)$ is $S\times 2$ matrix with i.i.d.\ rows.\footnote{As the proof of Theorem \ref{Theorem_master} reveals, it is also sufficient to require only that the rows are uncorrelated if we additionally assume that the joint distribution of all $2S$ matrix elements is fourth-moment Gaussian.}  Each row is a mean zero fourth-moment Gaussian (as in Definition \ref{Def_4moments}) two-dimensional vector with covariance matrix $\begin{pmatrix} C_{uu} & C_{uv}\\ C_{vu} & C_{vv}\end{pmatrix}$.
\item \label{ass_corr_noise_independence} The matrices $A \U$ and $B\V$ are assumed to be independent from $\x$ and $\y$.
\item \label{ass_corr_noise_non_degen} The matrices $A\U$ and $B\V$ almost surely have ranks $K-1$ and $M-1$, respectively.
\end{enumerate}
\end{assumption}

\noindent In the setting of Assumption \ref{ass_cor_noise}, we  set $\displaystyle r^2:=\frac{C_{uv}^2}{C_{uu} C_{vv}}$. The matrices $A \U$ and $B\V$ in the assumption might depend on each other and might be deterministic. These are $(K-1)\times S$ and $(M-1)\times S$ matrices, respectively, and we denote through $1\ge c_1^2 \ge c_2^2\ge \dots \ge c_{K-1}^2\ge 0$ their sample squared canonical correlations, which are eigenvalues of the $(K-1)\times(K-1)$ matrix\footnote{The non-degeneracy of Assumption \ref{ass_corr_noise_non_degen} was introduced to guarantee the existence of inverses.}  $(A \U \U^\T A^\T)^{-1} A \U \V^\T B^\T (B \V \V^\T B^\T)^{-1} B \V \U^\T A^\T$. We start by stating two theoretical results and then present a practical statistical implication as a corollary.

\begin{theorem} \label{Theorem_master} Suppose that the squared sample canonical correlations are constructed as in Definition \ref{Definition_sample_setting} with data matrices $\U$ and $\V$ satisfying Assumption \ref{ass_cor_noise}. Let $\lambda_i$ be one of the squared sample canonical correlations (not necessarily the largest one), and let $\widehat \x$, $\widehat \y$ be the corresponding canonical variables. Let $S$ tend to infinity and $K\le M$ depend on it in such a way that $S/K$ and $S/M$ converge to $\tau_K>1$ and $\tau_M>1$, respectively, and $\tau_M^{-1}+\tau_K^{-1}<1$.
In addition, fix $\eps>0$, and let
\begin{equation}
\label{eq_G_def}
 G(z):= \frac{1}{S} \sum_{k=1}^{K-1} \frac{1}{z- c_k^2}, \qquad z\in \mathbb C.
\end{equation}
Then, as $S\to\infty$, any canonical correlation $\lambda_i$ that is at a distance of at least $\eps$ from all $\{c_k^2\}_{k=1}^{K-1}$ satisfies the following relation \eqref{eq_CCA_master_asymptotic}, where the remainder term $o(1)$ tends to $0$ as $S \to \infty$ with fixed $\eps$, $\tau_K$, and $\tau_M$. Moreover, any $\lambda_i$ solving the relation \eqref{eq_CCA_master_asymptotic} (and a distance of at least $\eps$ away from $\{c_k^2\}_{k=1}^{K-1}$) is a canonical correlation.
\begin{equation} \label{eq_CCA_master_asymptotic}
  r^2 + o(1)=\frac{\displaystyle \lambda_i \left[1-2\frac{K}{S}-\frac{1}{\lambda_i} \cdot \frac{M-K}{S}-(1-\lambda_i)G(\lambda_i) \right]  \left[1-\frac{K}{S}-\frac{M}{S}-(1 - \lambda_i) G(\lambda_i) \right]}{\displaystyle  \left[1 - \frac{M}{S}-\lambda_i\frac{K}{S}  -  \lambda_i(1-\lambda_i)  G(\lambda_i)  \right]^2}.
\end{equation}
Let us further denote
\begin{equation}
\label{eq_Q_def}
 Q_x(z)=  -  \frac{1-2\frac{K}{S}-\frac{1}{z} \cdot \frac{M-K}{S}-(1-z)G(z)}{1 - \frac{M}{S}-z\frac{K}{S}  -  z(1-z)  G(z) }, \qquad Q_y(z)=- \frac{1-\frac{K}{S}-\frac{M}{S}-(1 - z) G(z)}{1 - \frac{M}{S}-z\frac{K}{S}  -  z(1-z)  G(z) },
\end{equation}
let $\cos^2\theta_x$ be the squared cosine of the angle between $\x$ and $\widehat \x$ defined as
  $
   \cos^2\theta_x = \frac{(\x^\T \widehat \x)^2}{(\x^\T \x)(\widehat \x^\T \widehat \x)}$, and let $\cos^2\theta_y$ be the squared cosine of the angle between $\y$ and $\widehat \y$. Then we have
  \begin{multline} \label{eq_alpha_cos_final}
 \cos^2\theta_x +o(1)= \Biggl(1-\frac{K}{S} -\lambda_i Q_x(\lambda_i) \left(\frac{K}{S}+ (1-\lambda_i) G(\lambda_i)\right) \Biggr)^2\\ \times \Biggl( 1-2\frac{K}{S} - 2\frac{K}{S}\lambda_i Q_x(\lambda_i)
 +G(\lambda_i)\left[2\lambda_i-1+2\lambda_i(2\lambda_i-1)  Q_x(\lambda_i) + \frac{\lambda_i^2}{r^2}  Q_x^2(\lambda_i)\right] \\ +(\lambda_i^2-\lambda_i)G'(\lambda_i) \Bigl[1  +2 \lambda_i  Q_x(\lambda_i) + \frac{\lambda_i}{r^2} Q_x^2(\lambda_i) \Bigr]\Biggr)^{-1},
\end{multline}
\begin{multline}\label{eq_beta_cos_final}
\cos^2\theta_y + o(1) =  \Biggl(1-\frac{M}{S} -\lambda_i Q_y(\lambda_i) \left(\frac{M}{S}+(1-\lambda_i)G(\lambda_i)+\frac{1-\lambda_i}{\lambda_i}\cdot \frac{M-K}{S}\right) \Biggr)^2
\\ \times \Biggl(
1-2\frac{M}{S} - 2\frac{K}{S}\lambda_i  Q_y(\lambda_i) +\frac{M-K}{S} \left[1+\frac{ Q_y^2 (\lambda_i)}{r^2}  \right]
 +G(\lambda_i)\Bigl[2\lambda_i-1\\ +2\lambda_i(2\lambda_i-1) Q_y(\lambda_i) + \frac{\lambda_i^2}{r^2}  Q_y^2(\lambda_i)\Bigr]  +(\lambda_i^2-\lambda_i)G'(\lambda_i) \Bigl[1  +2 \lambda_i  Q_y(\lambda_i) + \frac{\lambda_i}{r^2}  Q_y^2(\lambda_i) \Bigr]\Biggr)^{-1}.
\end{multline}
\end{theorem}
\begin{remark}
 In parallel to Theorem \ref{Theorem_basic_multi}, there is a straightforward extension of Theorem \ref{Theorem_master} to the case of $\mathbbm q>1$ signals. In this extension we have $2\mathbbm q$ random vectors $\x^q$, $\y^q$, $1\le q\le \mathbbm q$, with i.i.d.\ components, which represent the canonical variables in the population, and the corresponding canonical correlations should be all distinct. The final statement is as in Theorem \ref{Theorem_basic_multi}: the relations for each $q=1,2,\dots,\mathbbm q$ are the same as for the $\mathbbm q=1$, with the only difference that $r$ in the left-hand side of \eqref{eq_CCA_master_asymptotic} and in the right-hand side of \eqref{eq_alpha_cos_final}, \eqref{eq_beta_cos_final} depends on $q$.
\end{remark}
\begin{remark}
 In order to give $o(1)$ terms in \eqref{eq_CCA_master_asymptotic}, \eqref{eq_alpha_cos_final}, \eqref{eq_beta_cos_final} a proper meaning, one should think of $G(z)$ as changing with $S$ and converging to a limiting function.
\end{remark}

The formulas in Theorem \ref{Theorem_master} depend on the unknown function $G(z)$. There are two ways to avoid this difficulty. First, by imposing additional assumptions on the noise, one can deduce exact asymptotic formulas for $G(z)$---this is how we deduce in Section \ref{Section_appendix_asymptotics} Theorems \ref{Theorem_basic_setting} and \ref{Theorem_4moments} from Theorem \ref{Theorem_master}. In those theorems $G(z)$ approximates the Stieltjes transform of the Wachter distribution, with the density shown in orange in Fig.~\ref{Fig_spike}, \ref{fig_uniform_CCA}, \ref{fig_tdistr_CCA}, \ref{fig_multi_spikes_hist}, and \ref{stocks_hist}.

Alternatively, one can reuse the observed squared canonical correlations $\lambda_1\ge \lambda_2\ge \dots\ge \lambda_K$ through the following approximation statement, which is a direct corollary of Lemma \ref{Lemma_interlacing}.

\begin{lemma} \label{Lemma_G_approximation}
 Take $K\le M$, and let $\U$ and $\V$ be $K\times S$ and $M\times S$ matrices, respectively. In addition, fix $(K-1)\times K$ matrix $A$ and $(M-1)\times M$ matrix $B$. Let $\lambda_1\ge \dots\ge \lambda_K$ be the squared sample canonical correlations between $\U$ and $\V$; let $c_1^2\ge c_2^2\ge \dots\ge c_{K-1}^2$ be the squared sample canonical correlations between $A\U$ and $B\V$. For each $\ell\ge 1$, we have
\begin{equation}
 \lim_{S\to\infty} \left| \frac{1}{S} \sum_{k=1}^{K-1} \frac{1}{z-c_k^2} - \frac{1}{S} \sum_{k=\ell}^{K} \frac{1}{z-\lambda_k}\right|=
 \lim_{S\to\infty} \left| \frac{\partial}{\partial z}\left[\frac{1}{S} \sum_{k=1}^{K-1} \frac{1}{z-c_k^2}\right] - \frac{\partial}{\partial z}\left[\frac{1}{S} \sum_{k=\ell }^{K} \frac{1}{z-\lambda_k}\right]\right|=0,
\end{equation}
where the convergence is uniform over the choices of $\U$, $\V$, $A$, $B$, and over complex $z$ bounded away from the segment  $[\min(c^2_{K-1}, \lambda_K), \max(c_1^2,\lambda_{\ell})]$, and over $K,M$, such that $K/S$, $M/S$ are bounded away from $0$.
\end{lemma}

The primary application of Theorem \ref{Theorem_master} arises when $i=1$, focusing on the largest canonical correlation $\lambda_1$. In this case Theorem \ref{Theorem_master}, combined with Lemma \ref{Lemma_G_approximation} implies that one can use the asymptotic formulas \eqref{eq_CCA_master_asymptotic}--\eqref{eq_beta_cos_final} with $i=1$ and with  $G(z)$ replaced with \begin{equation}
 \hat G(z)= \frac{1}{S} \sum_{k=2}^{K} \frac{1}{z- \lambda_{k}},
\end{equation}
and therefore readily computable from the data. The assumption that $\lambda_1$ must be a distance of at least $\eps$ away from $\{c_k^2\}_{k=1}^{K-1}$ should be still preserved in the theorem. This is a conceptual, rather than a technical condition: without any additional assumptions, in the setting of this section the noise and signal components play symmetric roles and can not be distinguished from each other.

\section{The master equation}
\label{Section_appendix_master}

The proofs of all our theorems are based on \emph{the master equation} --- an exact equation satisfied by the canonical correlations and corresponding variables. For illustration purpose, we start by presenting the main ideas for developing such an equation in a simpler setting of principal component analysis (PCA) in Appendix \ref{Section_PCA}. We then proceed to the canonical correlation analysis (CCA) situation of our main interest in Appendix \ref{Section_CCA_master}.  Appendix \ref{Section_appendix_asymptotics} analyzes these equations as the dimensions go to infinity jointly and proportionally.

\subsection{PCA master equation}
\label{Section_PCA}

Suppose that we are given $N\times S$ matrix $\U$, in which the rows are indexed by $i=0,1,\dots,N-1$ and $S\ge N$. We treat the $0$--th row as ``signal'' and the remaining $N-1$ rows as ``noise''. Let $\widetilde \U$ denote the latter, i.e., it is the $(N-1)\times S$ matrix formed by rows $i=1,2,\dots,N-1$ of $\U$. As for the former, we let $\lambda_*$ denote the length of the zeroth row of $\U$ (treated as a $S$--dimensional vector) and let $\u^*$ be the unit vector in the direction of this row.

We would like to connect the singular values and singular vectors of $\U$ to the triplet: singular values and vectors of $\widetilde \U$;  $\lambda_*$; and $\u^*$. The goal is to view this connection through the lenses of reconstruction of $\lambda_*$ and $\u^*$ by the singular values and vectors of $\U$.

Because the singular values are invariant under orthogonal transformations in $N$--dimensional space, the fact that ``signal'' is the $0$th row of $\U$ is not important: it could have been any other row or their linear combination. Rephrasing, we would like to understand how the singular value decomposition of $\U$ distinguishes the signal vector from the background noise.\footnote{A related, but slightly different setup is to take a $N\times S$ noise matrix $\mathbf{X}$ and rank $1$ signal matrix $\mathbf{P}$ of the same size, set $\U=\mathbf{X}+\mathbf{P}$ and aim to identify $P$ from singular values (and vectors) of $\U$. Applying orthogonal transformation, one can again assume that $\mathbf{P}$ has only first nonzero row, however, this time the first row of $\U$ is not pure signal, but is also contaminated by the noise coming from the first row of $\mathbf{X}$. Therefore, in this setup the noise has two effects: it contaminates $\mathbf{P}$ by simple addition and then plays a role in computation of singular values. In contrast, in our setup we distinguish the two effects and only look into the second one (because the role of the first one is quite straightforward).}

We let $(\v_i,\u_i, \lambda_i)$, $1\le i \le N-1$, be the left singular vector (of $(N-1)\times 1$ dimensions), right singular vector (of $S\times 1$ dimensions), singular value triplets for $\widetilde \U$, which means that
$$
 \widetilde \U=\begin{pmatrix} \v_1; \v_2; \dots; \v_{N-1} \end{pmatrix} \begin{pmatrix} \lambda_1& 0 & \dots \\ 0 & \lambda_2 & 0 & \\ & 0 & \ddots\\ & & 0 &\lambda_{N-1}\end{pmatrix} \begin{pmatrix} \u_1^\T\\ \u_2^\T \\ \vdots \\ \u_{N-1}^\T\end{pmatrix}=\sum_{i=1}^{N-1} \lambda_i \v_i \u_i^\T
$$
and $\langle \u_i, \u_j\rangle=\delta_{i=j}$, $\langle \v_i, \v_j\rangle=\delta_{i=j}$. We order the singular values so that $ \lambda_1 \ge \lambda_{2}\ge \dots \ge \lambda_{N-1}\ge 0$.

Next, let $\widehat \ba$ be a left singular vector of $\U$. We represent it through coordinates $(\alpha_0,\alpha_1,\dots,\alpha_N)$ in the orthonormal basis
\begin{equation} \label{eq_ortho_basis}
\begin{pmatrix} 1\\ 0^{N-1}\end{pmatrix}, \quad \begin{pmatrix} 0\\ \v_i\end{pmatrix},\, 1\le i \le N-1,
\end{equation}
of $N$--dimensional space and normalize by $\sum_{i=0}^{N-1} \alpha_i^2=1$.
\begin{proposition} \label{Proposition_PCA_master}Suppose that $\widehat \ba=(\alpha_0,\alpha_1,\dots,\alpha_{N-1})$ is a left singular vector of $\U$ in basis \eqref{eq_ortho_basis} and with a squared singular value $a\ge 0$. Then
\begin{equation}
\label{eq_PCA_master_2}
 \lambda_*^2  =a \left[1+ \sum_{i=1}^{N-1} \frac{\lambda_i^2 \langle \u^*, \u_i\rangle^2}{a-\lambda_i^2}\right]^{-1}, \qquad  \alpha_0^2= \left[ 1+ a \dfrac{\sum_{i=1}^{N-1} \frac{\lambda_i^2 \langle \u^*, \u_i\rangle^2}{(a-\lambda_i^2)^2}}{1+\sum_{i=1}^{N-1} \frac{\lambda_i^2 \langle \u^*, \u_i\rangle^2}{a-\lambda_i^2}}\right]^{-1}.
\end{equation}
\end{proposition}
\begin{proof}
Each left singular vector $\widehat \ba$ of $\U$ is an $N\times 1$--dimensional column-vector of unit length, which is an eigenvector of $\U \U^\T$; equivalently, this is one of the critical points (on the space of unit vectors) for the quadratic form
$$
 \widehat \ba \mapsto f(\widehat \ba)=\langle \widehat \ba, \U \U^\T \widehat \ba \rangle= \| \U^\T \widehat \ba\|^2,
$$
with the eigenvector of largest eigenvalue corresponding to the maximum of this quadratic form.  Then we have
$$
 f(\widehat \ba)=\left\|\alpha_0 \lambda_* \u^* + \sum_{i=1}^{N-1} \alpha_i \lambda_i \u_i\right\|^2= \alpha_0^2 \lambda_*^2 + 2\alpha_0 \lambda_* \sum_{i=1}^{N-1} \alpha_i \lambda_i \langle \u^*, \u_i\rangle +\sum_{i=1}^{N-1} \alpha_i^2 \lambda_i^2.
$$
We aim to find critical points of $f(\widehat \ba)$ subject to the constraint $\|\widehat \ba\|^2=\sum_{i=0}^{N-1} \alpha_i^2=1$. Thus, we introduce the Lagrange multiplier $a$ and find critical points of the function
$$
 g(\widehat \ba,a)=f(\widehat \ba)-a \left( \sum_{i=0}^{N-1} \alpha_i^2 -1\right).
$$
Taking derivatives with respect to $\alpha$, we get:
\begin{equation}
\label{eq_PCA_equations}
\begin{dcases}
 0=\frac{\partial g}{\partial \alpha_i}= 2\alpha_0 \lambda_* \lambda_i \langle \u^*,\u_i\rangle+2\alpha_i \lambda_i^2 - 2a \alpha_i, & 1\le i \le N-1,\\ 0 =\frac{\partial g}{\partial \alpha_0}=2\alpha_0 \lambda_*^2 + 2\sum_{i=1}^{N-1} \alpha_i \lambda_* \lambda_i \langle \u^*, \u_i\rangle - 2 a \alpha_0.
\end{dcases}
\end{equation}
The equations \eqref{eq_PCA_equations} are supplemented by the normalization condition $\sum_{i=0}^{N-1} \alpha_i^2=1$.
Note that equations \eqref{eq_PCA_equations} equivalently express the fact that $\widehat \ba$ is an eigenvector with eigenvalue $a$ for the matrix $\U \U^\T$ written in the orthonormal basis \eqref{eq_ortho_basis}. From this interpretation we see that the values of $a$, for which \eqref{eq_PCA_equations} has a solution, are precisely eigenvalues of $\U \U^\T$, i.e., squared singular values of $\U$. On the other hand, \eqref{eq_PCA_equations} can be also solved directly. Indeed, the first set of $N-1$ equations leads to
\begin{equation}
\label{eq_x1}
 \alpha_i= \alpha_0 \frac{\lambda_* \lambda_i \langle \u^*, \u_i\rangle}{a-\lambda_i^2}, \qquad 1\le i \le N-1.
\end{equation}
Plugging the expressions for $\alpha_i$ into the last equation of \eqref{eq_PCA_equations}, we get
$$
 2\alpha_0 \left( \lambda_*^2 + \lambda_*^2 \sum_{i=1}^{N-1} \frac{\lambda_i^2 \langle \u^*, \u_i\rangle^2}{a-\lambda_i^2} -a \right)=0.
$$
If $\alpha_0=0$, then so are all $\alpha_i$ through \eqref{eq_x1}, which contradicts the $\sum_{i=0}^{N-1} \alpha_i^2=1$ normalization. Hence, we can divide by $2\alpha_0$, arriving at the final equation for $a$:
\begin{equation}
\label{eq_PCA_master}
 \lambda_*^2  \left(1+ \sum_{i=1}^{N-1} \frac{\lambda_i^2 \langle \u^*, \u_i\rangle^2}{a-\lambda_i^2}\right)=a,
\end{equation}
which is equivalent to the first equation in \eqref{eq_PCA_master_2}.

In order to compute the corresponding value of $\alpha_0$, which is the cosine of the angle between $\widehat \ba$ and the signal direction $(1,0^{N-1})$, plug \eqref{eq_x1} into the normalization condition $\sum_{i=0}^{N-1} \alpha_i^2=1$  to get
\begin{equation}
\label{eq_vector_1}
 \alpha_0^2 \left( 1+ \lambda_*^2 \sum_{i=1}^{N-1} \frac{\lambda_i^2 \langle \u^*, \u_i\rangle^2}{(a-\lambda_i^2)^2} \right)=1
\end{equation}
Plugging $\lambda_*^2$ from the first equation in \eqref{eq_PCA_master_2}, we rewrite \eqref{eq_vector_1} as the second equation of  \eqref{eq_PCA_master_2}.
\end{proof}

Here is a brief qualitative analysis of the formulas \eqref{eq_PCA_master_2}.
First, let us treat the first equation in \eqref{eq_PCA_master_2} in the form \eqref{eq_PCA_master} as an equation on $a$, assuming that  the values of $\lambda_*$, $\lambda_1,\dots, \lambda_{N-1}$ and $\langle \u^*, \u_1\rangle$, \dots, $\langle \u^*, \u_{N-1}\rangle$ are known. Let us also assume, for simplicity, that all $\lambda_i$ are distinct\footnote{If $\lambda_i=\lambda_{i+1}$, then an additional singular value $a=\lambda_i^2$ appears.}. After multiplying by $\prod_{i=1}^{N-1} (a-\lambda_i^2)$, the  equation  \eqref{eq_PCA_master} becomes a degree $N$ polynomial equation on $a$; hence, it has $N$ complex roots. For each $i=1,2,\dots,N-2$, the difference between right-hand side and left-hand side of \eqref{eq_PCA_master} continuously varies  from $+\infty$ to $-\infty$ on $(\lambda_{i+1}^2, \lambda_{i}^2)$ segment; hence there is one root of the equation on each such segment. In addition, by the same sign change argument there is one more root on $[0,\lambda_{N-1})$ and one more root on $(\lambda_1,+\infty)$. We are mostly interested in the latter two values of $a$, because one can easily distinguish a separated largest/smallest singular value from others, while doing the same for singular values in the bulk of spectrum is challenging.

Next, let us treat the formulas \eqref{eq_PCA_master_2} as a parameterization of $\lambda_*^2$ and $\alpha_0^2$ by the value of $a$. When we work with data, we know $a$: this is one of the squared singular values of the matrix $\U$. Therefore, it is reasonable to ask what information on (unknown to us) $\lambda_*$ and $\alpha_0$ it provides. There are two important regimes here:
\begin{enumerate}
 \item If $a$ is large, then $\lambda_*^2=a+O(1)$ is also large, while $\alpha_0^2=1+O(\tfrac{1}{a})$ approaches $1$.
 \item If $a$ is close to $0$, then $\lambda_*^2=O(a)$ is also close to $0$, while $\alpha_0^2=1+O(a)$ approaches $1$.
\end{enumerate}
On the other hand, for the intermediate values of $a$ and $\lambda_*^2$, the value of $\alpha_0^2$ is typically quite far away from $1$. The conclusion is that we can effectively distinguish the signal vector (the zeroth row of $\U$), if its length is either very small or very large, as compared to the values of $\{\lambda_i\}_{i=1}^{N}$, which are singular values of the matrix $\widetilde \U$ formed by the remaining $N$ rows of $\U$. Recall that $\widetilde \U$ represents the noise part; it is, perhaps, strange to ask the noise to be large enough, and much more natural to assume that the noise is small enough. Hence, in practice, the situation of interest is when the length $\lambda_*$ of the signal vector is large (rather than small) compared to the magnitude of the noise.

\subsection{CCA master equation}
\label{Section_CCA_master}
We proceed to the canonical correlations analysis setting. Suppose that we are given two subspaces, $\widetilde \U$ and $\widetilde \V$ in $S$--dimensional space and let $\dim(\widetilde \V)={M-1}\ge {K-1}=\dim(\widetilde \U)$. In addition, we have two vectors $\u^*$ and $\v^*$, which we add to spaces $\widetilde \U$ and $\widetilde \V$. Define
$$
 \U=\mathrm{span}(\u^*, \widetilde \U), \qquad \V=\mathrm{span}(\v^*, \widetilde \V).
$$
Our task is to reconstruct the vectors $\u^*$ and $\v^*$ inside the spaces $\U$ and $\V$, respectively, by analyzing the canonical correlations between $\U$ and $\V$.

\smallskip

We start by bringing the pair of subspaces $\widetilde \U$ and $\widetilde \V$ to the canonical form. The following definitions of canonical correlations and variables can already be found in \citet{harold1936relations}, see bottom of page 330 there. We also refer to the textbook \citet[Chapter 12]{anderson1958introduction}.

\begin{lemma}\label{Lemma_canonical_bases}
 Suppose that $S> M \ge K$ and let $\widetilde \U$ and $\widetilde \V$ be ${K-1}$-dimensional and ${M-1}$-dimensional subspaces of $S$ dimensional space, respectively. Then there exist two orthonormal bases: vectors $\u_1, \u_2, \dots, \u_{K-1}$ span $\widetilde \U$ and vectors $\v_1,\dots, \v_{M-1}$ span $\widetilde \V$ -- such that for all meaningful indices $i$ and $j$:
\begin{equation}
\label{eq_scalar_products_table}
 \langle \u_i, \u_j\rangle = \delta_{i=j}, \qquad \langle \v_i, \v_j\rangle = \delta_{i=j},\qquad \langle \u_i, \v_j\rangle = c_i \delta_{i=j},
\end{equation}
where  $1\ge c_1\ge c_2 \ge \dots \ge c_{K-1}\ge 0$.
\end{lemma}
The numbers $c_1\ge c_2 \ge \dots \ge c_{K-1}$ are called canonical correlation coefficients between subspaces $\widetilde \U$ and $\widetilde \V$; they are also cosines of the canonical angles between the subspaces. The vectors $\u_i$, $1\le i \le {K-1}$, and $\v_j$, $1\le j\le {M-1}$, are called canonical variables; they split into ${K-1}$ pairs $(\u_i,\v_i)$ and $M-K$ singletons $v_j$, $j\ge K$. In some of the following formulas we also use $c_j$ with $j\ge K$ under the convention:
\begin{equation}
\label{eq_can_cor_convention}
 c_{K}=c_{K+1}=\dots=c_{M-1}=0.
\end{equation}

There are two equivalent ways to find the canonical correlations and variables:
\begin{itemize}
 \item If we consider a function $f(\u,\v)=\langle \u, \v\rangle$, in which $\u$ varies over all unit vectors in $\widetilde \U$ and $\v$ varies over all unit vectors in $\widetilde \V$, then $(\u_i,\v_i,c_i)$, $1\le i \le {K-1}$, are critical points of $f$ and corresponding values of $f$. In particular, $c_1$ is the maximum of $f$, which is achieved at $(\u_1,\v_1)$. The remaining vectors $\v_j$, $j>{K-1}$ are an arbitrary orthogonal basis in the part of $\widetilde \V$ orthogonal to all vectors of $\widetilde \U$ and $\v_1,\dots,\v_{K-1}$.
 \item For a subspace $W$, let us denote through $P_W$ the orthogonal projector on $W$. Then nonzero $c_i^2$, $1\le i \le {K-1}$ are nonzero eigenvalues of $P_{\widetilde \U}P_{\widetilde \V} P_{\widetilde \U}$ and $\u_i$ are corresponding eigenvectors. Simultaneously, nonzero $c_i^2$ are eigenvalues of $P_{\widetilde \V}P_{\widetilde \U} P_{\widetilde \V}$. If we identify the spaces $\widetilde \U$ and $\widetilde \V$ with spans of rows of $(K-1)\times S$ and $(M-1)\times S$ matrices, then this is equivalent\footnote{One can move $\widetilde \U^\T$ from the right to the left in $(\widetilde \U \widetilde \U^\T)^{-1} \widetilde \U \widetilde \V^\T (\widetilde \V \widetilde \V^\T)^{-1} \widetilde \V \widetilde \U^\T$ without changing the eigenvalues. We get the matrix  $P_{\widetilde \U} P_{\widetilde \V}$, which is the same as $P_{\widetilde \U} P_{\widetilde \U} P_{\widetilde \V}$ and has the same eigenvalues as $P_{\widetilde \U} P_{\widetilde \V} P_{\widetilde \U}$. }  to Definition \ref{Definition_sample_setting}.
\end{itemize}

Throughout the rest of this section we use the bases of Lemma \ref{Lemma_canonical_bases}.

In exactly the same fashion we can also define the canonical correlation coefficients and corresponding variables for the spaces $\U$ and $\V$. The next theorem connects them to $\u^*$, $\v^*$ and the data of Lemma \ref{Lemma_canonical_bases}

Take two vectors $\widehat \ba=(\alpha_0,\alpha_1,\dots,\alpha_{K-1})$ and $\widehat \bb=(\beta_0,\beta_1,\dots,\beta_{M-1})$, such that $ \alpha_0\u^*+\sum_{i=1}^{K-1} \alpha_i \u_i$ and $\beta_0 \v^*+\sum_{j=1}^{M-1} \beta_j \v_j$ is a pair of canonical variables for $\U$ and $\V$. We normalize the vectors so that
\begin{equation}
\label{eq_normalization}
 \left\|\alpha_0 \u^*+ \sum_{i=1}^{K-1} \alpha_i \u_i\right\|^2=\left\|\beta_0 \v^*+ \sum_{j=1}^{M-1} \beta_j \v_j\right\|^2=1.
\end{equation}

The next theorem is a direct analogue of Proposition \ref{Proposition_PCA_master} in the CCA setting.

\begin{theorem} \label{Theorem_master_equation}
Take two vectors $\widehat \ba$, $\widehat \bb$ normalized as in \eqref{eq_normalization} and such that $\alpha_0 \u^*+  \sum_{i=1}^{K-1} \alpha_i \u_i$ and $\beta_0 \v^*+\sum_{j=1}^{M-1} \beta_j \v_j$ is a pair of canonical variables for $\U$ and $\V$ with squared canonical correlation coefficient $z=\left\langle \alpha_0 \u^*+ \sum_{i=1}^{K-1} \alpha_i \u_i,\, \beta_0 \v^* +\sum_{j=1}^{M-1} \beta_j \v_j\right\rangle^2$. Then we have
\begin{multline}
\label{eq_CCA_master}
 \left[\langle \u^*, \v^*\rangle + \sum_{j=1}^{M-1}  \frac{\langle\u^*, \v_j\rangle(c_j \langle \v^*, \u_j\rangle -z \langle \v^*, \v_j\rangle)}{z-c_j^2} -z  \sum_{i=1}^{K-1}    \frac{\langle\u^*, \u_i\rangle(\langle \v^*, \u_i\rangle - c_i \langle \v^*, \v_i\rangle)}{z - c_i^2}   \right]^2\\= z \left[-\langle\u^*,\u^*\rangle+\sum_{j=1}^{M-1}\frac{  \langle \u^*, \v_j \rangle^2 -2 c_j  \langle\u^*, \v_j\rangle  \langle \u^*, \u_j\rangle }{z-c_j^2} +z\sum_{i=1}^{K-1}   \frac{\langle \u^*, \u_i\rangle^2}{z - c_i^2} \right] \\ \times \left[-\langle\v^*,\v^*\rangle+\sum_{i=1}^{K-1}  \frac{\langle \v^*, \u_i\rangle^2 - 2  c_i \langle \v^*, \u_i\rangle \langle \v^*, \v_i\rangle}{z -c_i^2} +z\sum_{j=1}^{M-1} \frac{\langle\v^*,\v_j\rangle^2}{z - c_j^2}\right],
\end{multline}
where we used the convention \eqref{eq_can_cor_convention}, and
\begin{align}
\notag \frac{1}{\alpha_0^2}= \langle\u^*,\u^*\rangle &+ 2 \sum_{i=1}^{K-1}  \frac{\langle \u^*, \u_i\rangle}{z - c_i^2} \Bigl[ c_i \langle \u^*, \v_i\rangle -z \langle \u^*, \u_i\rangle -z \mathfrak Q_\alpha(z) \bigl(\langle \v^*, \u_i\rangle - c_i \langle \v^*, \v_i\rangle\bigr)\Bigr]\\ & +\sum_{i=1}^{K-1} \frac{1}{(z - c_i^2)^2} \Bigl[ c_i \langle \u^*, \v_i\rangle -z \langle \u^*, \u_i\rangle -z \mathfrak Q_\alpha(z) \bigl(\langle \v^*, \u_i\rangle - c_i \langle \v^*, \v_i\rangle\bigr)\Bigr]^2, \label{eq_CCA_alpha}
\end{align}
\begin{align}
\notag \frac{1}{\beta_0^2}= \langle\v^*,\v^*\rangle &+ 2 \sum_{j=1}^{M-1} \frac{\langle \v^*, \v_j\rangle }{z-c_j^2} \Bigl[-z \mathfrak Q_\beta(z) \bigl( \langle \u^*, \v_j \rangle -c_j \langle \u^*, \u_j\rangle \bigr)+c_j \langle \v^*, \u_j\rangle -z \langle \v^*, \v_j\rangle\Bigr] \\  &+\sum_{j=1}^{M-1}  \frac{1}{(z-c_j^2)^2} \Bigl[-z \mathfrak Q_\beta(z) \bigl( \langle \u^*, \v_j \rangle -c_j \langle \u^*, \u_j\rangle \bigr)+c_j \langle \v^*, \u_j\rangle -z \langle \v^*, \v_j\rangle\Bigr]^2, \label{eq_CCA_beta}
\end{align}
and
\begin{multline}
\label{eq_CCA_cos_a}
  \left\langle \u^*,\,  \alpha_0 \u^*+ \sum_{i=1}^{K-1} \alpha_i \u_i \right\rangle\\=\alpha_0 \Biggl(\langle \u^*,\u^*\rangle+ \sum_{i=1}^{K-1} \frac{\langle \u^*,\u_i\rangle}{z - c_i^2} \Bigl[ c_i \langle \u^*, \v_i\rangle -z \langle \u^*, \u_i\rangle -z \mathfrak Q_\alpha(z) \bigl(\langle \v^*, \u_i\rangle - c_i \langle \v^*, \v_i\rangle\bigr)\Bigr]\Biggr),
\end{multline}
\begin{multline}
\label{eq_CCA_cos_b}
   \left\langle \v^*,\,  \beta_0\v^*+\sum_{j=1}^{M-1} \beta_i \v_i \right\rangle\\=\beta_0 \Biggl(\langle \v^*,\v^*\rangle+ \sum_{j=1}^{M-1}\frac{\langle \v^*,\v_j\rangle}{z-c_j^2} \Bigl[-z \mathfrak Q_\beta(z) \bigl( \langle \u^*, \v_j \rangle -c_j \langle \u^*, \u_j\rangle \bigr)+c_j \langle \v^*, \u_j\rangle -z \langle \v^*, \v_j\rangle\Bigr]\Biggr),
\end{multline}
where
\begin{equation}
\label{eq_CCA_alpha_Q} \mathfrak Q_\alpha(z):= \frac{-\langle\u^*,\u^*\rangle+\sum_{j=1}^{M-1}\frac{  \langle \u^*, \v_j \rangle^2 -2 c_j  \langle\u^*, \v_j\rangle  \langle \u^*, \u_j\rangle }{z-c_j^2} +z\sum_{i=1}^{K-1}   \frac{\langle \u^*, \u_i\rangle^2}{z - c_i^2} }{\langle \u^*, \v^*\rangle + \sum_{j=1}^{M-1}  \frac{\langle\u^*, \v_j\rangle(c_j \langle \v^*, \u_j\rangle -z \langle \v^*, \v_j\rangle)}{z-c_j^2} -z  \sum_{i=1}^{K-1}    \frac{\langle\u^*, \u_i\rangle(\langle \v^*, \u_i\rangle - c_i \langle \v^*, \v_i\rangle)}{z - c_i^2}},
\end{equation}
\begin{equation}
\label{eq_CCA_beta_Q} \mathfrak Q_\beta(z):=
  \frac{-\langle\v^*,\v^*\rangle+\sum_{i=1}^{K-1}  \frac{\langle \v^*, \u_i\rangle^2 - 2  c_i \langle \v^*, \u_i\rangle \langle \v^*, \v_i\rangle}{z -c_i^2} +z\sum_{j=1}^{M-1} \frac{\langle\v^*,\v_j\rangle^2}{z - c_j^2}}{\langle\u^*, \v^*\rangle+\sum_{i=1}^{K-1} \frac{\langle \v^*, \u_i\rangle (c_i \langle \u^*, \v_i\rangle -z \langle \u^*, \u_i\rangle)}{z-c_i^2}-z\sum_{j=1}^{M-1} \frac{\langle \v^*, \v_j\rangle( \langle \u^*, \v_j \rangle -c_j \langle \u^*, \u_j\rangle )}{z-c_j^2}}.
\end{equation}
\end{theorem}
\begin{remark}
 Theorem \ref{Theorem_master_equation} details the generic situation when all $c_i$, $1\le i \le K-1$, have multiplicity one and the condition
 $\langle \u^*,\u_i\rangle=\langle \u^*,\v_i\rangle=\langle \v^*,\u_i\rangle=\langle \v^*,\v_i\rangle=0$ does not hold. The two remaining complementary cases are:
 \begin{itemize}
  \item $z=c_i^2$ for $1\le i \le {K-1}$, where $c_i$ has multiplicity one, $\langle \u^*,\u_i\rangle=\langle \u^*,\v_i\rangle=\langle \v^*,\u_i\rangle=\langle \v^*,\v_i\rangle=0$, and the equation \eqref{eq_CCA_master} holds with the $i$-th terms excluded;
  \item $z=c_i^2$, where $c_i$ has multiplicity larger than one.
 \end{itemize}
  These two degenerate cases are not relevant for our asymptotic analysis.
\end{remark}
\begin{remark}
 Comparing to the PCA setting, \eqref{eq_CCA_master} is an analogue of the first equation in \eqref{eq_PCA_master_2} and \eqref{eq_CCA_alpha}, \eqref{eq_CCA_beta} are analogues of the second equation in \eqref{eq_PCA_master_2}. A version of \eqref{eq_CCA_cos_a}, \eqref{eq_CCA_cos_b} for the PCA would be the computation of the scalar product between $\u^*$ and the \emph{right} singular vector corresponding to the singular value $a$, which we omitted there.
\end{remark}


\begin{proof}[Proof of Theorem \ref{Theorem_master_equation}]
We seek for a pair of vectors $\widehat \ba=(\alpha_0,\alpha_1,\dots,\alpha_{K-1})$ and $\widehat \bb=(\beta_0,\beta_1,\dots,\beta_{M-1})$, which represent critical points of the function
$$
 f(\widehat \ba,\widehat \bb)=\left\langle \alpha_0 \u^*+ \sum_{i=1}^{K-1} \alpha_i \u_i,\, \beta_0 \v^*+ \sum_{j=1}^{M-1} \beta_j \v_j\right\rangle,
$$
subject to the normalization constraints
\begin{equation}
\label{eq_normalization_2}
 \left\|\alpha_0 \u^*+  \sum_{i=1}^{K-1} \alpha_i \u_i\right\|^2=\left\|\beta_0 \v^*+  \sum_{j=1}^{M-1} \beta_j \v_j\right\|^2=1.
\end{equation}
In this way, $\alpha_0 \u^*+  \sum_{i=1}^{K-1} \alpha_i \u_i$,  $\beta_0 \v^*+ \sum_{j=1}^{M-1} \beta_j \v_j$ is a pair of canonical variables for $\U$ and $\V$ and the corresponding value of $f(\widehat \ba,\widehat \bb)$ is a canonical correlation coefficient.

Introducing the Lagrange multipliers $a$ and $b$ corresponding to two normalizations \eqref{eq_normalization_2}, we are led to the Lagrangian function
\begin{multline}
 g(\widehat \ba,\widehat \bb,a,b)= \alpha_0 \beta_0 \langle \u^*, \v^*\rangle + \alpha_0 \sum_{j=1}^{M-1} \beta_j \langle \u^*,\v_j\rangle +\beta_0 \sum_{i=1}^{K-1} \alpha_i \langle \v^*, \u_i\rangle + \sum_{i=1}^{K-1} \alpha_i \beta_i c_i
 \\- a\left( \alpha_0^2 \langle\u^*,\u^*\rangle + 2\alpha_0 \sum_{i=1}^{K-1} \alpha_i \langle \u^*, \u_i\rangle +\sum_{i=1}^{K-1} \alpha_i^2 - 1\right)
   \\- b\left( \beta_0^2 \langle\v^*,\v^*\rangle + 2 \beta_0 \sum_{j=1}^{M-1} \beta_j \langle \v^*, \v_j\rangle +\sum_{j=1}^{M-1} \beta_j^2 -1\right).
\end{multline}
We need to find critical points of this function.
Differentiating with respect to $\alpha_i$ and $\beta_j$, we get a system of $K+M$ homogeneous linear equations on coordinates of $\widehat \ba$ and $\widehat \bb$:

\begin{equation} \label{eq_CCA_equations}
\begin{dcases}
 0=\frac{\partial g}{\partial \alpha_0}= \beta_0 \langle \u^*, \v^*\rangle + \sum_{j=1}^{M-1} \beta_j \langle \u^*, \v_j\rangle - 2 a\alpha_0 \langle\u^*,\u^*\rangle -2 a \sum_{i=1}^{K-1} \alpha_i \langle \u^*, \u_i\rangle,\\
 0=\frac{\partial g}{\partial \alpha_i}= \beta_0 \langle \v^*, \u_i\rangle + \beta_i c_i -2a \alpha_0 \langle \u^*, \u_i \rangle - 2 a \alpha_i, & 1\le i < K,\\
 0=\frac{\partial g}{\partial \beta_0} = \alpha_0 \langle \u^*, \v^*\rangle +\sum_{i=1}^{K-1} \alpha_i \langle \v^*, \u_i\rangle  - 2b \beta_0 \langle\v^*,\v^*\rangle - 2b \sum_{j=1}^{M-1} \beta_j \langle \v^*, \v_j\rangle,\\
  0=\frac{\partial g}{\partial \beta_j}= \alpha_0 \langle \u^*, \v_j\rangle + \delta_{j< K} \cdot \alpha_j c_j - 2 b \beta_0 \langle \v^*, \v_j\rangle - 2b \beta_j, & 1 \le j <M.\\
\end{dcases}
\end{equation}
In the matrix form, the equations \eqref{eq_CCA_equations} can be rewritten as
\begin{equation} \label{eq_CCA_equations_matrix}
\begin{cases}
 S_{uv} \widehat \bb=2a S_{uu} \widehat \ba,\\
 S_{vu} \widehat \ba= 2b S_{vv} \widehat \bb,
\end{cases}
\end{equation}
where $\widehat \ba$ and $\widehat \bb$ are treated as column-vectors of sizes $K\times 1$ and $M\times 1$, respectively, and $S_{\cdot\cdot}$  are matrices of scalar products:
\begin{equation}
 [S_{uv}]_{i,j}=[S_{vu}]_{j,i}=\langle \u_i, \v_j\rangle, \qquad [S_{uu}]_{i,i'}=\langle \u_i, \u_{i'}\rangle, \qquad [S_{vv}]_{j,j'}=\langle \v_j, \v_{j'}\rangle,
\end{equation}
with  $0\le i,i'\le {K-1}$, $0\le j,j'\le {M-1}$ and upon identification $\u_0=\u^*$, $\v_0=\v^*$. Note that \eqref{eq_CCA_equations_matrix} implies eigenrelations
\begin{equation} \label{eq_CCA_equations_matrix_2}
\begin{cases}
 (S_{vv})^{-1} S_{vu} (S_{uu})^{-1} S_{uv}\, \widehat \bb=4ab\, \widehat \bb,\\
 (S_{uu})^{-1} S_{uv} (S_{vv})^{-1} S_{vu}\, \widehat \ba= 4ab\, \widehat \ba.
\end{cases}
\end{equation}
Comparing \eqref{eq_CCA_equations_matrix_2} with standard algorithms for finding canonical correlations as in Definition \ref{Definition_sample_setting}, we conclude that $4ab$ should necessary be one of the squared sample canonical correlations between spaces $\U$ and $\V$. On the other hand, the equations \eqref{eq_CCA_equations} can be solved directly.\footnote{On the technical level, this is the main difference of our approach from the previous attempts on this problem in the literature: solving \eqref{eq_CCA_equations_matrix_2} is difficult, because the left-hand sides involve matrix inversions. On the other hand, directly solving \eqref{eq_CCA_equations} by exploiting their block structure turns out to be much simpler.}    We first combine the second and the forth equation for $i=j\le {K-1}$ and treat the result as a system of two linear equations on the variables $(\alpha_i, \beta_i)$:
\begin{equation}
\label{eq_2equations}
\begin{cases}
 - 2 a \alpha_i + c_i\beta_i  &=\,\,2a \alpha_0 \langle \u^*, \u_i \rangle -\beta_0 \langle \v^*, \u_i\rangle, \\
  \,\,\,\,\,c_i\alpha_i - 2b \beta_i &=\,\,- \alpha_0 \langle \u^*, \v_i\rangle  + 2 b \beta_0 \langle \v^*, \v_i\rangle.
\end{cases}
\end{equation}
We solve these equation by inverting the $2\times 2$ matrix in the left-hand side:
$$
 \begin{pmatrix} - 2 a &  c_i \\
                 c_i & - 2b
 \end{pmatrix}^{-1}
 =\frac{1}{4ab - c_i^2} \begin{pmatrix} -2b & -c_i \\ -c_i & -2a \end{pmatrix}.
$$
Multiplying by the right-hand side of \eqref{eq_2equations}, the solution is
\begin{equation}
\label{eq_2equations_solved}
 \begin{dcases}
  \alpha_i=\frac{1}{4ab - c_i^2} \Bigl[ \bigl(c_i \langle \u^*, \v_i\rangle -4ab \langle \u^*, \u_i\rangle\bigr)\cdot \alpha_0 &+\quad \bigl(\langle \v^*, \u_i\rangle - c_i \langle \v^*, \v_i\rangle\bigr) \cdot 2b \cdot \beta_0\Bigr],\\
  \beta_i=\frac{1}{4ab-c_i^2} \Bigl[ \bigl( \langle \u^*, \v_i \rangle -c_i \langle \u^*, \u_i\rangle \bigr)\cdot 2a \cdot \alpha_0 &+\quad \bigl(c_i \langle \v^*, \u_i\rangle -4ab \langle \v^*, \v_i\rangle\bigr)\cdot \beta_0\Bigr].
 \end{dcases}
\end{equation}
In addition, for $K\le j<M$, we rewrite the last equation of \eqref{eq_CCA_equations} in a similar form:
\begin{equation}
\label{eq_additional_solved}
 \beta_j= \frac{1}{4ab} \Bigl[  \langle \u^*, \v_j\rangle \cdot 2a \cdot \alpha_0  - 4 a b \langle \v^*, \v_j  \rangle \cdot \beta_0 \Bigr],
\end{equation}
which can be thought of as the second equation of \eqref{eq_2equations_solved} with $c_j=0$. Next, we plug \eqref{eq_2equations_solved} and \eqref{eq_additional_solved} back into the first and third equations of \eqref{eq_CCA_equations} and get:
\begin{multline}
\label{eq_x2}
  \alpha_0 \cdot 2a \cdot  \left[-\langle\u^*,\u^*\rangle+\sum_{j=1}^{M-1}\frac{  \langle \u^*, \v_j \rangle^2 -2 c_j  \langle\u^*, \v_j\rangle  \langle \u^*, \u_j\rangle }{4ab-c_j^2} +4ab\sum_{i=1}^{K-1}   \frac{\langle \u^*, \u_i\rangle^2}{4ab - c_i^2} \right]
  \\
  +\beta_0 \left[\langle \u^*, \v^*\rangle + \sum_{j=1}^{M-1}  \frac{\langle\u^*, \v_j\rangle(c_j \langle \v^*, \u_j\rangle -4ab \langle \v^*, \v_j\rangle)}{4ab-c_j^2} -4ab  \sum_{i=1}^{K-1}    \frac{\langle\u^*, \u_i\rangle(\langle \v^*, \u_i\rangle - c_i \langle \v^*, \v_i\rangle)}{4ab - c_i^2}   \right]
  \\ =0,
\end{multline}
\begin{multline}
\label{eq_x3}
 \alpha_0 \left[\langle\u^*, \v^*\rangle+\sum_{i=1}^{K-1} \frac{\langle \v^*, \u_i\rangle (c_i \langle \u^*, \v_i\rangle -4ab \langle \u^*, \u_i\rangle)}{4ab-c_i^2}-4ab\sum_{j=1}^{M-1} \frac{\langle \v^*, \v_j\rangle( \langle \u^*, \v_j \rangle -c_j \langle \u^*, \u_j\rangle )}{4ab-c_j^2}\right]
 \\+  \beta_0\cdot 2b\cdot \left[-\langle\v^*,\v^*\rangle+\sum_{i=1}^{K-1}  \frac{\langle \v^*, \u_i\rangle^2 - 2  c_i \langle \v^*, \u_i\rangle \langle \v^*, \v_i\rangle}{4ab -c_i^2} +4ab\sum_{j=1}^{M-1} \frac{\langle\v^*,\v_j\rangle^2}{4ab - c_j^2}\right]
 =0,
\end{multline}
where we adopted the agreement $c_j=0$ for $K\le j<M$. \eqref{eq_x2} and \eqref{eq_x3} is a system of two homogeneous linear equations on $(\alpha_0,\beta_0)$. If the system is non-degenerate, then the only solution is $(0,0)$, which then leads through \eqref{eq_2equations_solved} and \eqref{eq_additional_solved} to vanishing of all $\alpha_i$, $\beta_j$, which contradicts the normalization condition \eqref{eq_normalization}. Hence, the system must be degenerate, which is equivalent to vanishing of the determinant of $2\times 2$ matrix of its coefficients. Noticing that the coefficient of $\beta_0$ in \eqref{eq_x2} and the coefficient of $\alpha_0$ in \eqref{eq_x3} are two equivalent forms of the same expression, the condition becomes
\begin{multline}
 \left[\langle \u^*, \v^*\rangle + \sum_{j=1}^{M-1}  \frac{\langle\u^*, \v_j\rangle(c_j \langle \v^*, \u_j\rangle -4ab \langle \v^*, \v_j\rangle)}{4ab-c_j^2} -4ab  \sum_{i=1}^{K-1}    \frac{\langle\u^*, \u_i\rangle(\langle \v^*, \u_i\rangle - c_i \langle \v^*, \v_i\rangle)}{4ab - c_i^2}   \right]^2\\= 4ab \left[-\langle\u^*,\u^*\rangle+\sum_{j=1}^{M-1}\frac{  \langle \u^*, \v_j \rangle^2 -2 c_j  \langle\u^*, \v_j\rangle  \langle \u^*, \u_j\rangle }{4ab-c_j^2} +4ab\sum_{i=1}^{K-1}   \frac{\langle \u^*, \u_i\rangle^2}{4ab - c_i^2} \right] \\ \times \left[-\langle\v^*,\v^*\rangle+\sum_{i=1}^{K-1}  \frac{\langle \v^*, \u_i\rangle^2 - 2  c_i \langle \v^*, \u_i\rangle \langle \v^*, \v_i\rangle}{4ab -c_i^2} +4ab\sum_{j=1}^{M-1} \frac{\langle\v^*,\v_j\rangle^2}{4ab - c_j^2}\right].
\end{multline}
Denoting $z:=4ab$, we arrive at the desired \eqref{eq_CCA_master}. Recalling the interpretation of $z$ explained after \eqref{eq_CCA_equations_matrix_2}, \eqref{eq_CCA_master} is the equation for the squared canonical correlation coefficients $z$ between subspaces $\U$ and $\V$.

\medskip

Once $z$ is identified from \eqref{eq_CCA_master}, we plug it back into \eqref{eq_x3} and find the relation
\begin{equation}
\label{eq_x4}
 \alpha_0=-\beta_0 \cdot 2b \cdot \mathfrak Q_\beta(z),
\end{equation}
where $\mathfrak Q_\beta(z)$ is given by \eqref{eq_CCA_beta_Q}.
 Further, \eqref{eq_2equations_solved} transforms into
\begin{equation}
\label{eq_x9}
  \beta_j=\beta_0 \frac{1}{z-c_j^2} \Bigl[-z \mathfrak Q_\beta(z) \bigl( \langle \u^*, \v_j \rangle -c_j \langle \u^*, \u_j\rangle \bigr)+c_j \langle \v^*, \u_j\rangle -z \langle \v^*, \v_j\rangle\Bigr].
\end{equation}
The normalization equation $\left\|\beta_0 \v^*+ \sum_{j=1}^{M-1} \beta_j \v_j\right\|^2=1$ becomes the desired \eqref{eq_CCA_beta}:
\begin{align}
\notag \frac{1}{\beta_0^2}= \langle\v^*,\v^*\rangle &+ 2 \sum_{j=1}^{M-1} \frac{\langle \v^*, \v_j\rangle }{z-c_j^2} \Bigl[-z \mathfrak Q_\beta(z) \bigl( \langle \u^*, \v_j \rangle -c_j \langle \u^*, \u_j\rangle \bigr)+c_j \langle \v^*, \u_j\rangle -z \langle \v^*, \v_j\rangle\Bigr] \\ \notag &+\sum_{j=1}^{M-1}  \frac{1}{(z-c_j^2)^2} \Bigl[-z \mathfrak Q_\beta(z) \bigl( \langle \u^*, \v_j \rangle -c_j \langle \u^*, \u_j\rangle \bigr)+c_j \langle \v^*, \u_j\rangle -z \langle \v^*, \v_j\rangle\Bigr]^2
\end{align}
We further reconstruct $\alpha_0$ by using  \eqref{eq_x2} instead of \eqref{eq_x3} in the form:
$$
 \beta_0=-\alpha_0\cdot 2a \cdot \mathfrak Q_\alpha(z),
$$
where $\mathfrak Q_\alpha(z)$ is given by \eqref{eq_CCA_alpha_Q}. Then we transform \eqref{eq_2equations_solved} into
\begin{equation}
\label{eq_x8}
   \alpha_i=\alpha_0 \frac{1}{z - c_i^2} \Bigl[ c_i \langle \u^*, \v_i\rangle -z \langle \u^*, \u_i\rangle -z \mathfrak Q_\alpha(z) \bigl(\langle \v^*, \u_i\rangle - c_i \langle \v^*, \v_i\rangle\bigr)\Bigr].
\end{equation}
The normalization condition  $\left\|\alpha_0 \u^*+  \sum_{i=1}^{K-1} \alpha_i \u_i\right\|^2=1$ becomes the desired \eqref{eq_CCA_alpha}:
\begin{align}
\notag \frac{1}{\alpha_0^2}= \langle\u^*,\u^*\rangle &+ 2 \sum_{i=1}^{K-1}  \frac{\langle \u^*, \u_i\rangle}{z - c_i^2} \Bigl[ c_i \langle \u^*, \v_i\rangle -z \langle \u^*, \u_i\rangle -z \mathfrak Q_\alpha(z) \bigl(\langle \v^*, \u_i\rangle - c_i \langle \v^*, \v_i\rangle\bigr)\Bigr]\\ &\notag +\sum_{i=1}^{K-1} \frac{1}{(z - c_i^2)^2} \Bigl[ c_i \langle \u^*, \v_i\rangle -z \langle \u^*, \u_i\rangle -z \mathfrak Q_\alpha(z) \bigl(\langle \v^*, \u_i\rangle - c_i \langle \v^*, \v_i\rangle\bigr)\Bigr]^2.
\end{align}
The formulas for the scalar products \eqref{eq_CCA_cos_a} and \eqref{eq_CCA_cos_b} follow by using \eqref{eq_x8} and \eqref{eq_x9}, respectively.
\end{proof}

Lemmas \ref{Lemma_number_of_roots} and \ref{Lemma_interlacing} clarify the structure of \eqref{eq_CCA_master} in Theorem \ref{Theorem_master_equation} treated as an equation on an unknown variable $z$, assuming that all $\langle \u_i,\u_j\rangle$, $\langle \u_i,\v_j\rangle$, $\langle \v_i,\v_j\rangle$, and $c_j$ are given.

\begin{lemma} \label{Lemma_number_of_roots}
The identity \eqref{eq_CCA_master}, treated as an equation on an unknown $z$, is equivalent to a polynomial equation of degree (at most) $K$, and, therefore, has (at most) $K$ roots.
\end{lemma}
\begin{proof}
Let us study the behavior of \eqref{eq_CCA_master} near $z=c_k^2$. The left-hand side behaves as:
\begin{equation}
   \left[\frac{c_k(\langle\u^*, \v_k\rangle-c_k\langle\u^*, \u_k\rangle) (\langle \v^*, \u_k\rangle -c_k \langle \v^*, \v_k\rangle)}{z - c_k^2} +O(1)   \right]^2.
\end{equation}
The right-hand side behaves as:
\begin{equation}
 c_k^2 \left[ \frac{  (\langle \u^*, \v_k \rangle -c_k  \langle \u^*, \u_k\rangle)^2}{z-c_k^2}   +O(1) \right] \left[  \frac{(\langle \v^*, \u_k\rangle -   c_k  \langle \v^*, \v_k\rangle)^2}{z -c_k^2}  +O(1)\right].
\end{equation}
We notice the cancelation of the double pole in the difference of the left-hand side and right-hand side. Hence, the difference has only a simple pole. Therefore, if we multiply both sides of \eqref{eq_CCA_master} by $\prod_{i=1}^{K-1} (z-c_i^2)$, we get a polynomial equation\footnote{Note that there is no singularity at $0$: all denominators $z-c_j^2$ with $K\le j <M$ are matched with factors $z$ in numerators.} of degree (at most) $K$.
\end{proof}

\begin{lemma} \label{Lemma_interlacing} Let $z_1,\dots, z_K$ be squared sample canonical correlations from Theorem \ref{Theorem_master_equation}. Then
\begin{enumerate}
\item All $z_i$ are real numbers between $0$ and $1$.
\item If we arrange $z_i$ in the decreasing order, then there exists another sequence of ${K-1}$ real numbers $y_1\ge y_2\ge \dots\ge y_{{K-1}}$, such that two interlacing conditions hold:
\begin{equation}
\label{eq_interlacing_1}
 z_1\ge y_1 \ge z_2 \ge \dots \ge y_{K-1}\ge z_{K},\qquad \text{ and }
\end{equation}
\begin{equation}
\label{eq_interlacing_2}
 y_1 \ge c_1^2 \ge y_2 \ge \dots \ge y_{K-1} \ge c_{K-1}^2.
\end{equation}
\end{enumerate}
\end{lemma}
\begin{proof}
It is tricky to see this property by directly analyzing the equation \eqref{eq_CCA_master} and we proceed in another way
by using the identification of $\{z_i\}$ with squared canonical correlation coefficients between $\U$ and $\V$, as claimed in Theorem \ref{Theorem_master_equation}. It immediately follows that they are real numbers between $0$ and $1$. Next, by definition, $\{c_i^2\}$ are squared canonical correlations between $\widetilde \U$ and $\widetilde \V$. The numbers $\{y_i\}$ in \eqref{eq_interlacing_1}, \eqref{eq_interlacing_2} also have a similar interpretation: they are squared canonical correlations between $\widetilde \U$ and $\V$.

In order to see interlacement, it is helpful to use the identification of canonical correlations with eigenvalues of products of projectors.  $\{z_i\}$ are nonzero eigenvalues of $P_\U P_\V P_\U$, while $\{y_i\}$ are nonzero eigenvalues of $P_{\widetilde \U} P_\V P_{\widetilde \U}$. If we choose an orthonormal basis of $\U$, for which $\widetilde \U$ is spanned by the first ${K-1}$ basis vectors, then the former matrix (ignoring zero part) becomes a $K\times K$ matrix, while the latter matrix is its $(K-1)\times (K-1)$ principal submatrix. Hence, \eqref{eq_interlacing_1} are classical (see e.g., \citet[Corollary III.1.5]{bhatia_matrix}) interlacing inequalities between eigenvalues of a symmetrix matrix and its submatrix. On the other hand, we can also identify $\{y_i\}$ with nonzero eigenvalues of $P_\V P_{\widetilde \U} P_\V$. Then \eqref{eq_interlacing_2} is obtained by comparing this matrix (viewed as $M\times M$ matrix in $V$) with its $(M-1)\times (M-1)$ submatrix $P_{\widetilde \V} P_{\widetilde \U} P_{\widetilde \V}$, whose nonzero eigenvalues are $\{c_i^2\}$; the interlacing condition looks slightly different because of the additional $0$s among the eigenvalues: if we add the smallest coordinate $y_{K}=0$, then one can think of \eqref{eq_interlacing_2} as being of the same form as \eqref{eq_interlacing_1}.
\end{proof}

\section{Asymptotic approximations}

\label{Section_appendix_asymptotics}

The proofs of all theorems in Sections \ref{Section_setting} and \ref{Section_general} are based on the $K,M,S\to\infty$ asymptotic approximations of the formulas of Theorem \ref{Theorem_master_equation}. We first prove Theorem \ref{Theorem_master}, and then show that all other theorems in Sections \ref{Section_setting} and \ref{Section_general} are its corollaries.

\subsection{Proof of Theorem \ref{Theorem_master}}

\label{Section_proof_most_general}

We start by connecting the settings of Theorems \ref{Theorem_master} and \ref{Theorem_master_equation}.

\begin{lemma} \label{Lemma_simple_data} In Theorem \ref{Theorem_master}, one can take without loss of generality the vectors and matrices of Assumption \ref{ass_cor_noise} to be: $\ba=(1,0^{K-1})$, $\bb=(1,0^{M-1})$; $A$ is the matrix of the projection operator onto the last $K-1$ coordinate vectors, i.e., $A[i,j]=\delta_{j=i+1}$, $i=1,\dots,K-1$, $j=1,\dots,K$, and $B$ is the matrix of the projection operator onto the last $M-1$ coordinate vectors, i.e., $B[i,j]=\delta_{j=i+1}$, $i=1,\dots,M-1$, $j=1,\dots,M$.
\end{lemma}
\begin{proof}
Suppose that Assumption \ref{ass_cor_noise} holds with some vectors $\ba^\sharp$, $\bb^\sharp$, some matrices $A^\sharp$, $B^\sharp$, and some data $\U^\sharp$, $\V^\sharp$. Our task is to transform the data to the form claimed in Lemma \ref{Lemma_simple_data}.
Let $\tilde{A}^\sharp$ be $K\times K$ matrix, whose first row is $(\ba^\sharp)^\T$ and the remaining rows are given by matrix $A^\sharp$. We claim that ${\tilde A}^\sharp$ is invertible. Indeed, if $\tilde{A}^\sharp$ is degenerate, while $A^\sharp$ has rank $K-1$, then there should be a way to express $(\ba^\sharp)^\T$ as a linear combination of rows of $A^\sharp$. Hence, $\x^\T=(\ba^\sharp)^\T \U$ is a linear combination of rows of $A^\sharp \U$. But then independence of $\x$ and $A^\sharp\U$ postulated in Assumption \ref{ass_corr_noise_independence} is impossible. Similarly, we define an invertible matrix ${\tilde B}^\sharp$ to be $M\times M$ matrix, whose first row is $(\bb^\sharp)^\T$ and remaining rows are given by $B^\sharp$.

We let $\U=\tilde A^\sharp \U^\sharp$ and $\V=\tilde B^\sharp\V^\sharp$ and rephrase everything in terms of $\U$ and $\V$. Note that the canonical correlations and variables in Theorem \ref{Theorem_master} depend only on the linear subspaces (in $S$--dimensional space) spanned by the rows of $\U$ and $\V$, rather than on the matrices $\U$ and $\V$ themselves. Hence, due to invertibility of $\tilde A^\sharp$ and $\tilde B^\sharp$, they are the same for the matrices $(\U,\V)$ and $(\U^\sharp, \V^\sharp)$. Hence, $\lambda_1$, $\hat\x$, and $\hat \y$ in Theorem \ref{Theorem_master} are unchanged. Simultaneously, $\x$, $\y$, $A \U$, $B\V$ do not change since $\x=(\U^\sharp)^\T \ba^\sharp=\U^\T\ba$, $\y=(\V^\sharp)^\T\bb^\sharp=\V\bb$, $A^\sharp\U^\sharp=A \U$, $B^\sharp \V^\sharp=B\V$, where $A$, $B$, $\ba$, $\bb$ are from the statement of the lemma. Thus, all the ingredients in the formulas \eqref{eq_G_def}-\eqref{eq_beta_cos_final} remain the same.
\end{proof}

The next step is to get to rid of (unknown to us) $\langle \u^*,\u_i\rangle$, $\langle \u^*,\v_i\rangle$, $\langle \v^*,\u_i\rangle$, $\langle \u^*,\v_i\rangle$ appearing in the formulas of Theorem \ref{Theorem_master_equation}.

\begin{lemma} \label{Lemma_asymptotic_approx} Suppose that in $S$--dimensional space we are given a pair of random vectors $(\u^*,\v^*)$ and a collection of vectors $(\u_i)_{i=1}^{K-1}$, $(\v_j)_{j=1}^{M-1}$ and numbers $(c_j)_{j=1}^{M-1}$, with $c_K=\dots=c_{M-1}=0$ such that
 \begin{itemize}
  \item The $S\times 2$ matrix $(\u^*,\v^*)$ has fourth-moment Gaussian elements, in the sense of Definition \ref{Def_4moments}, such that the $S$ rows are mean zero and uncorrelated with each other. The covariance matrix of each row is the same $\begin{pmatrix} C_{uu} & C_{uv}\\ C_{vu} & C_{vv}\end{pmatrix}$.
  \item Either $(\u_i)_{i=1}^{K-1}$, $(\v_j)_{j=1}^{M-1}$, $(c_j)_{j=1}^{M-1}$ are deterministic, or independent from $(\u^*,\v^*)$.
  \item The scalar products for $i,j\ge 1$ are $\langle \u_i,\u_j\rangle =\delta_{i=j}$, $\langle \v_i,\v_j\rangle =\delta_{i=j}$, $\langle \u_i,\v_j\rangle =\delta_{i=j} c_i$.
 \end{itemize}
Then as $S\to\infty$ we have
 \begin{align}
 \label{eq_approx_1}
 \sum_{i=1}^{K-1}   \frac{\langle \u^*, \u_i\rangle^2}{z - c_i^2}&\approx C_{uu}  \sum_{i=1}^{K-1}   \frac{1}{z - c_i^2},
 &
 \sum_{j=1}^{M-1}   \frac{\langle \u^*, \v_j\rangle^2}{z - c_j^2}&\approx C_{uu}  \sum_{j=1}^{M-1}   \frac{1}{z - c_j^2},
 \\
 \label{eq_approx_2}
 \sum_{i=1}^{K-1}   \frac{\langle \v^*, \u_i\rangle^2}{z - c_i^2}&\approx C_{vv}  \sum_{i=1}^{K-1}   \frac{1}{z - c_i^2},
 &
 \sum_{j=1}^{M-1}   \frac{\langle \v^*, \v_j\rangle^2}{z - c_j^2}&\approx C_{vv}  \sum_{j=1}^{M-1}   \frac{1}{z - c_j^2},
 \\
  \label{eq_approx_3}
 \sum_{i=1}^{K-1}   \frac{\langle \u^*, \u_i\rangle \langle \v^*, \u_i\rangle}{z - c_i^2}&\approx C_{uv}  \sum_{i=1}^{K-1}   \frac{1}{z - c_i^2},
 &
 \sum_{j=1}^{M-1}   \frac{\langle \u^*, \v_j\rangle\langle \v^*, \v_j\rangle}{z - c_j^2}&\approx C_{uv}  \sum_{j=1}^{M-1}   \frac{1}{z - c_j^2},
 \\
  \label{eq_approx_4}
 \sum_{i=1}^{K-1}   \frac{c_i\langle \u^*, \u_i\rangle \langle \u^*, \v_i\rangle}{z - c_i^2}&\approx C_{uu}  \sum_{i=1}^{K-1}   \frac{c_i^2}{z - c_i^2},
 &
 \sum_{i=1}^{K-1}   \frac{c_i\langle \v^*, \u_i\rangle \langle \v^*, \v_i\rangle}{z - c_i^2}&\approx C_{vv}  \sum_{i=1}^{K-1}   \frac{c_i^2}{z - c_i^2},
  \\
   \label{eq_approx_5}
 \sum_{i=1}^{K-1}   \frac{c_i\langle \u^*, \u_i\rangle \langle \v^*, \v_i\rangle}{z - c_i^2}&\approx C_{uv}  \sum_{i=1}^{K-1}   \frac{c_i^2}{z - c_i^2},
 &
 \sum_{i=1}^{K-1}   \frac{c_i\langle \v^*, \u_i\rangle \langle \u^*, \v_i\rangle}{z - c_i^2}&\approx C_{uv}  \sum_{i=1}^{K-1}   \frac{c_i^2}{z - c_i^2},
 \end{align}
 where the $\approx$ sign means that the difference between right- and left-hand sides is $o(S)$ (tends to $0$ in probability after dividing by $S$) uniformly in complex $z$ bounded away from all zeros  of the denominators, $\{c_i^2\}$. Similar asymptotic approximations hold if we replace all $(z-c_i^2)^{-1}$  with $(z-c_i^2)^{-2}$.
\end{lemma}
\begin{proof}
 We condition on  $(\u_i)_{i=1}^{K-1}$, $(\v_j)_{j=1}^{M-1}$, $(c_j)_{j=1}^{M-1}$ throughout the proof and assume them to be deterministic.

{\bf Step 1.} Let us show that the expectation of the left-hand side matches the right-hand side in all approximations. Take any two deterministic $S$--dimensional vectors ${\boldsymbol \chi}$  and ${\boldsymbol \psi}$. Using the uncorrelated mean $0$ assumption on the components of $\u^*$ and $\v^*$, we have
$$
 \E \langle \u^*, {\boldsymbol \chi}\rangle \langle \u^*, {\boldsymbol \psi}\rangle= \E \left(\sum_{k=1}^S [\u^*]_k \chi_k\right) \left(\sum_{k=1}^S [\u^*]_k \psi_k\right)
 =\sum_{k=1}^S \chi_k \psi_k \E  \bigl([\u^*]_k [\u^*]_k\bigr) = C_{uu} \langle {\boldsymbol \chi},{\boldsymbol \psi}\rangle.
$$
Similarly,
$$
 \E \langle \v^*, {\boldsymbol \chi}\rangle \langle \v^*, {\boldsymbol \psi}\rangle= C_{vv} \langle {\boldsymbol \chi},{\boldsymbol \psi}\rangle, \qquad  \E \langle \u^*, {\boldsymbol \chi}\rangle \langle \v^*, {\boldsymbol \psi}\rangle= C_{uv} \langle {\boldsymbol \chi},{\boldsymbol \psi}\rangle.
$$
Applying these expectation identities and using the scalar products table for $\u_i$ and $\v_i$, we conclude that the expectations of the left sides of \eqref{eq_approx_1}-\eqref{eq_approx_5} are given by the right sides.

\smallskip

{\bf Step 2.} Next, we show that the terms in the sums in  \eqref{eq_approx_1}-\eqref{eq_approx_5} are uncorrelated. For that we take four $S$--dimensional deterministic vectors ${\boldsymbol \chi}$, ${\boldsymbol \psi}$, ${\boldsymbol \chi}'$, ${\boldsymbol \psi}'$ such that
$$
 \langle {\boldsymbol \chi}, {\boldsymbol \chi}'\rangle=  \langle {\boldsymbol \chi}, {\boldsymbol \psi}'\rangle =  \langle {\boldsymbol \psi}, {\boldsymbol \chi}'\rangle =  \langle {\boldsymbol \psi}, {\boldsymbol \psi}'\rangle =0.
$$
Since the coordinates of $\u^*$ are mean $0$, uncorrelated, fourth-moment Gaussian,  we have:
\begin{align}
\label{eq_x5}
\E &\langle \u^*, {\boldsymbol \chi}\rangle \langle \u^*, {\boldsymbol \psi}\rangle  \langle \u^*, {\boldsymbol \chi}'\rangle \langle \u^*, {\boldsymbol \psi}'\rangle
\\ \notag &
=
\E \left(\sum_{k=1}^S [\u^*]_k \chi_k\right) \left(\sum_{k=1}^S [\u^*]_k \psi_k\right) \left(\sum_{k=1}^S [\u^*]_k \chi'_k\right) \left(\sum_{k=1}^S [\u^*]_k  \psi'_k\right)
\\ \notag &=\sum_{k=1}^S \left(\E \bigl([\u^*]_k\bigr)^4- 3 \left(\E \bigl([\u^*]_k\bigr)^2\right)^2\right)  \chi_k \psi_k \chi_k' \psi_k'
\\ \notag & \qquad + (C_{uu})^2 \left[ \langle {\boldsymbol \chi}, {\boldsymbol \psi}\rangle \langle {\boldsymbol \chi}', {\boldsymbol \psi}'\rangle+\langle {\boldsymbol \chi}, {\boldsymbol \psi}'\rangle \langle {\boldsymbol \chi}', {\boldsymbol \psi}\rangle
+\langle {\boldsymbol \chi}, {\boldsymbol \chi}'\rangle \langle {\boldsymbol \psi}, {\boldsymbol \psi}'\rangle \right]
\\ & \notag =(C_{uu})^2 \langle {\boldsymbol \chi}, {\boldsymbol \psi}\rangle \langle {\boldsymbol \chi}', {\boldsymbol \psi}'\rangle=\E \langle \u^*, {\boldsymbol \chi}\rangle \langle \u^*, {\boldsymbol \psi}\rangle  \cdot \E \langle \u^*, {\boldsymbol \chi}'\rangle \langle \u^*, {\boldsymbol \psi}'\rangle,
\end{align}
where in transition from the second to the third line we used that all the joint moments of coordinates, in which some coordinate is repeated one time, vanish because of fourth-moment Gaussian assumption, and the sum $\sum\limits_{k=1}^S$ in the third line vanishes by the same reason. In the same way, we get
\begin{equation}
\label{eq_x6}
\E \langle \v^*, {\boldsymbol \chi}\rangle \langle \v^*, {\boldsymbol \psi}\rangle  \langle \v^*, {\boldsymbol \chi}'\rangle \langle \v^*, {\boldsymbol \psi}'\rangle=\E \langle \v^*, {\boldsymbol \chi}\rangle \langle \v^*, {\boldsymbol \psi}\rangle \cdot \E \langle \v^*, {\boldsymbol \chi}'\rangle \langle \v^*, {\boldsymbol \psi}'\rangle,
\end{equation}
And by a similar computation, we have
\begin{align}
\label{eq_x7}
\E &\langle \u^*, {\boldsymbol \chi}\rangle \langle \v^*, {\boldsymbol \psi}\rangle  \langle \u^*, {\boldsymbol \chi}'\rangle \langle \v^*, {\boldsymbol \psi}'\rangle
\\ & \notag =
\E \left(\sum_{k=1}^S [\u^*]_k \chi_k\right) \left(\sum_{k=1}^S [\v^*]_k \psi_k\right) \left(\sum_{k=1}^S [\u^*]_k \chi'_k\right) \left(\sum_{k=1}^S [\v^*]_k \psi'_k\right)
\\ \notag
 &=\sum_{k=1}^S \Bigl(\E \left([\u^*]_k^2 [\v^*]_k^2\right) -\E [\u^*]_k^2 \E [\v^*]_k^2 - 2 \left( \E [\u^*]_k [\v^*]_k \right)^2   \Bigr)  \chi_k \psi_k \chi_k' \psi_k'
 \\ \notag
 &\qquad +  C_{uv} C_{uv} \langle {\boldsymbol \chi}, {\boldsymbol \psi}\rangle \langle {\boldsymbol \chi}', {\boldsymbol \psi}'\rangle + C_{uv} C_{uv}\langle {\boldsymbol \chi}, {\boldsymbol \psi}'\rangle \langle {\boldsymbol \chi}', {\boldsymbol \psi}\rangle
+ C_{uu} C_{vv} \langle {\boldsymbol \chi}, {\boldsymbol \chi}'\rangle \langle {\boldsymbol \psi}, {\boldsymbol \psi}'\rangle
 \\
 & \notag =(C_{uv})^2 \langle {\boldsymbol \chi}, {\boldsymbol \psi}\rangle \langle {\boldsymbol \chi}', {\boldsymbol \psi}'\rangle
 =\E \langle \u^*, {\boldsymbol \chi}\rangle \langle \v^*, {\boldsymbol \psi}\rangle  \cdot \E \langle \u^*, {\boldsymbol \chi}'\rangle \langle \v^*, {\boldsymbol \psi}'\rangle.
\end{align}
Altogether, \eqref{eq_x5}, \eqref{eq_x6}, and \eqref{eq_x7} show that all the sums in \eqref{eq_approx_1}-\eqref{eq_approx_5} have uncorrelated terms.

\smallskip

{\bf Step 3.} The statement of the lemma for a fixed $z$ follows by the weak law of large numbers for uncorrelated sum: the variance of each sum is upper bounded by a constant times $S$.

In order to prove the uniformity in $z$, let $\xi_S(z)$ denote the difference of the left- and right-hand sides in one of the approximations \eqref{eq_approx_1}--\eqref{eq_approx_5}. We fix $\delta>0$ and aim to prove that $\tfrac{1}{S} \xi_S(z)$ tends to $0$ in probability as $S\to\infty$ uniformly over all $z\in\mathbb C$ at distance at least $\delta$ from $\{c_i^2\}$. For a large $R>0$, consider the event
$$
 \Omega_R=\left\{ \frac{1}{S}\left(\sum_{i=1}^{K-1}\langle \u^*, \u_i\rangle^2 + \sum_{i=1}^{K-1}\langle \v^*, \u_i\rangle^2 + \sum_{j=1}^{M-1}\langle \u^*, \v_j\rangle^2 + \sum_{j=1}^{M-1}\langle \v^*, \v_j\rangle^2\right)< R\right\}.
$$
By the law of large numbers, probability of $\Omega_R$ tends to $1$, as $R\to\infty$, and, therefore, it is sufficient to prove the convergence of
$\tfrac{1}{S} \xi_S(z)$ to $0$ on the event $\Omega_R$ for each $R>0$.

Choose $\eps>0$ and note that the magnitude of each term in the sums in the right-hand sides of \eqref{eq_approx_1}--\eqref{eq_approx_5} is upper bounded by $\frac{1}{|z|-1}$. Hence, there exists a constant $d>0$, such that
\begin{equation}
\label{eq_x26}
 \left|\frac{1}{S} \xi_S(z)\right|< \eps, \qquad \text{ on the event } \Omega_R  \text{ for all }z\text{ such that }|z|>d.
\end{equation}
Introduce a compact set $\mathfrak C\subset \mathbb C$ given by
$$
 \mathfrak C=\{z\in\mathbb C:\, |z|\le d,\quad |z-c_i^2|\ge \delta \text{ for all }i\}.
$$
There exists a constant $\ell$, depending only on $\delta$ and $R$, such that $\frac{1}{S} \xi_S(z)$ is $\ell$--Lipshitz on $\mathfrak C$:
$$
 \left|\frac{1}{S} \xi_S(z_1)-\frac{1}{S} \xi_S(z_2)\right|\le \ell |z_1-z_2|, \qquad \text{ for all } z_1,z_2\in\mathfrak C,\qquad \text{on the event }\Omega_R.
$$
Therefore, we can choose a finite collection of points $z_1,\dots,z_n\in \mathfrak C$, such that $n$ does not grow with $S$ and
\begin{equation}
 \mathrm{Prob}\left[\sup_{z\in\mathfrak C} \left|\frac{1}{S} \xi_S(z)\right| > \eps; \,\, \Omega_R\right]\le \sum_{i=1}^n  \mathrm{Prob}\left[\left|\frac{1}{S} \xi_S(z_i)\right| > \eps/2; \,\, \Omega_R\right],
\end{equation}
where the events are restricted onto $\Omega_R$. The right-hand side of the last formula tends to $0$ as $S\to\infty$ by the fixed $z$ convergence result. Hence, combining with \eqref{eq_x26} we deduce the desired uniformity in $z$.
\end{proof}

\begin{proof}[Proof of Theorem \ref{Theorem_master}] Using Lemma \ref{Lemma_simple_data}, we assume without loss of generality that $\ba=(1,0^{K-1})$ and $\bb=(1,0^{M-1})$, which means that the signal vectors are the first rows of $\U$ and $\V$, respectively. Further, by the same lemma we assume $A[i,j]=\delta_{j=i+1}$ and $B[i,j]=\delta_{j=i+1}$, which means that the noise part is given by the remaining $K-1$ rows of $\U$, denoted as $(K-1)\times S$ matrix $\widetilde \U$ and remaining $M-1$ rows of $\V$ denoted as $(M-1)\times S$ matrix $\widetilde \V$. This is the setting of Theorem \ref{Theorem_master_equation} and we apply it.

We divide \eqref{eq_CCA_master} by $S^2$ and note that by the law of large numbers
\begin{equation}
\label{eq_x10}
 \frac{1}{S}\langle \u^*, \u^* \rangle = C_{uu}+o(1),\qquad  \frac{1}{S}\langle \v^*, \v^* \rangle = C_{vv}+o(1).
\end{equation}
Hence, using Lemma \ref{Lemma_asymptotic_approx}, we get an asymptotic approximation of \eqref{eq_CCA_master}:
\begin{multline*}
 C^2_{uv}\left[1 - \frac{M-1}{S} -z  \frac{1}{S} \sum_{i=1}^{K-1}    \frac{1-c_i^2}{z - c_i^2}   \right]^2+o(1)\\= C_{uu} C_{vv} z \left[-1+\frac{1}{S} \sum_{j=1}^{M-1}\frac{  1 -2 c_j^2 }{z-c_j^2} +z\frac{1}{S} \sum_{i=1}^{K-1}   \frac{1}{z - c_i^2} \right] \left[-1+\frac{1}{S} \sum_{i=1}^{K-1}  \frac{1 - 2  c_i^2}{z -c_i^2} +z\frac{1}{S} \sum_{j=1}^{M-1} \frac{1}{z - c_j^2}\right].
\end{multline*}
The desired \eqref{eq_CCA_master_asymptotic} is an equivalent form of the same equation with renamed variable $z=\lambda_i$. By Rouche's theorem, locally uniform convergence of equations implies convergence of their solutions. Hence, we arrive at the statement of Theorem \ref{Theorem_master} regarding the canonical correlation and \eqref{eq_CCA_master_asymptotic}.

Further, recalling that $\widehat \x$ in Theorem \ref{Theorem_master} and Definition \ref{Definition_sample_setting} becomes  $\alpha_0 \u^*+ \sum_{i=1}^{K-1} \alpha_i \u_i$ in Theorem \ref{Theorem_master_equation}, while $\x$ in Theorem \ref{Theorem_master} becomes $\u^*$, and noticing that $\widehat \x$ was normalized in our procedure, while $\x$ was not, we rewrite $\cos^2(\theta_x)$ as
\begin{equation}
\label{eq_x11}
\cos^2(\theta_x)=\frac{\left(\langle \u^*, \alpha_0 \u^*+ \sum_{i=1}^{K-1} \alpha_i \u_i \rangle \right)^2 }{\langle \u^*,\u^*\rangle}
\end{equation}
 For the denominator in \eqref{eq_x11} we use \eqref{eq_x10}. For the numerator we first use \eqref{eq_CCA_cos_a} to express it through $\alpha_0$. Using Lemma \ref{Lemma_asymptotic_approx} and notation \eqref{eq_G_def}, we get
 \begin{align}
\notag \cos^2(\theta_x)+o(1)&=S \frac{\alpha_0^2}{C_{uu}}
 \Biggl(C_{uu}+ \frac{1}{S}\sum_{i=1}^{K-1} \frac{1}{z - c_i^2} \Bigl[ c_i^2 C_{uu} -z C_{uu} -z \mathfrak Q_\alpha(z) \bigl(C_{uv} - c_i^2 C_{uv}\bigr)\Bigr]\Biggr)^2
 \\ &= C_{uu} S\alpha_0^2
 \Biggl(1-\frac{K-1}{S} -z \mathfrak Q_\alpha(z) \frac{C_{uv}}{C_{uu}}\left(\frac{K-1}{S}+ (1-z) G(z)\right) \Biggr)^2 \label{eq_x25}
\end{align}
In the last formula $K-1$ can be replaced with $K$ leading to another $o(1)$ error. For $\alpha_0^2$, we use \eqref{eq_CCA_alpha},  Lemma \ref{Lemma_asymptotic_approx} and \eqref{eq_x10} getting:
\begin{align}
\notag  \frac{1}{C_{uu} S \alpha_0^2}&=o(1)+1 + \frac{2}{S} \sum_{i=1}^{K-1}  \frac{1}{z - c_i^2} \Bigl[ c_i^2  -z  -z \mathfrak Q_\alpha(z) \tfrac{C_{uv}}{C_{uu}} \bigl(1 - c_i^2 \bigr)\Bigr]
\\
\notag & \qquad +\frac{1}{S}\sum_{i=1}^{K-1} \frac{1}{(z - c_i^2)^2} \Bigl[ c_i^2 -2 z c_i^2  +z^2  +2 z^2 \tfrac{C_{uv}}{C_{uu}} \mathfrak Q_\alpha(z)(1-c_i^2) + z^2 \mathfrak Q_\alpha(z)^2 \tfrac{C_{vv}}{C_{uu}}(1-c_i^2)\Bigr]
\\
&=\notag o(1)+1-2\tfrac{K}{S} - 2\tfrac{K}{S }z \tfrac{C_{uv}}{C_{uu}} \mathfrak Q_\alpha(z)  \\ &\notag \qquad +G(z)\left[2z-1+2z(2z-1)\tfrac{C_{uv}}{C_{uu}}  \mathfrak Q_\alpha(z) + \tfrac{C_{vv}}{C_{uu}}z^2 \mathfrak Q_\alpha^2(z)\right]
\\  &\qquad +(z^2-z)G'(z) \Bigl[1  +2\tfrac{C_{uv}}{C_{uu}} z \mathfrak Q_\alpha(z) + \tfrac{C_{vv}}{C_{uu}}z \mathfrak Q_\alpha^2(z) \Bigr]\notag .
\end{align}
Next, we analyze the asymptotics of $\mathfrak Q_\alpha(z)$ from \eqref{eq_CCA_alpha_Q} using again Lemma \ref{Lemma_asymptotic_approx} and \eqref{eq_x10}:
\begin{align}
 \mathfrak Q_\alpha(z)\frac{C_{uv}}{C_{uu}} &=  \frac{-1+ \frac{1}{S} \sum_{j=1}^{M-1}\frac{  1 -2 c_j^2 }{z-c_j^2} +\frac{z}{S}\sum_{i=1}^{K-1}   \frac{1}{z - c_i^2} }
 {1+\frac{1}{S} \sum_{j=1}^{M-1}  \frac{c_j^2-z}{z-c_j^2} -\frac{z}{S} \sum_{i=1}^{K-1}    \frac{1-c_i^2}{z - c_i^2}}+o(1)
 \\
 \notag &= -  \frac{1-\frac{2K}{S} -\frac{1}{z} \frac{M-K}{S}- (1-z) G(z)}
 {1-\frac{M}{S} -z\frac{K}{S} -z(1-z) G(z)}+o(1)
\end{align}
This is precisely the function $Q_x(z)$ of \eqref{eq_Q_def}. Hence, plugging $\alpha_0^2$ and $\mathfrak Q_\alpha(z)$  into \eqref{eq_x25} and recalling that $r^2=\frac{C_{uv}^2}{C_{uu} C_{vv}}$, we get \eqref{eq_alpha_cos_final}.

For $\cos^2(\theta_y)$ we argue similarly and obtain:
\begin{align*}
\cos^2(\theta_y)&=\frac{\langle \v^*, \beta_0 \v^*+  \sum_{i=1}^{M-1} \beta_i \v_i \rangle^2 }{\langle \v^*,\v^*\rangle}\\&= S\frac{\beta_0^2}{C_{vv}} \Biggl(C_{vv}+ \frac{1}{S}\sum_{j=1}^{M-1}\frac{1}{z-c_j^2} \Bigl[-z \mathfrak Q_\beta(z) \bigl(  C_{uv} -c_j^2 C_{uv}  \bigr)+C_{vv}c_j^2  -z C_{vv}\Bigr]\Biggr)^2 +o(1)
\\&= S \beta_0^2 C_{vv}\left(1-\tfrac{M}{S}- z \mathfrak Q_\beta(z) \tfrac{C_{uv}}{C_{vv}} \left(\tfrac{M}{S}+(1-z) G(z) + \tfrac{1-z}{z}  \cdot \tfrac{M-K}{S}\right)\right)^2+o(1),
\end{align*}
\begin{align*}
 \frac{1}{C_{vv} S \beta_0^2}&= o(1)+1-2\tfrac{M}{S} - 2\tfrac{K}{S }z \tfrac{C_{uv}}{C_{vv}} \mathfrak Q_\beta(z)+\tfrac{M-K}{S}\left[ 1+  \tfrac{C_{uu}}{C_{vv}} \mathfrak Q_\beta^2(z)\right]  \\ &\notag \qquad +G(z)\left[2z-1+2z(2z-1)\tfrac{C_{uv}}{C_{vv}}  \mathfrak Q_\beta(z) + \tfrac{C_{uu}}{C_{vv}}z^2 \mathfrak Q_\beta^2(z)\right]
\\  &\qquad +(z^2-z)G'(z) \Bigl[1  +2\tfrac{C_{uv}}{C_{vv}} z \mathfrak Q_\beta(z) + \tfrac{C_{uu}}{C_{vv}}z \mathfrak Q_\beta^2(z) \Bigr],
\end{align*}
\begin{align*}
\mathfrak Q_\beta(z)\frac{C_{uv}}{C_{vv}} &=
   -  \frac{1-\frac{K}{S} -\frac{M}{S}- (1-z) G(z)}
 {1-\frac{M}{S}-z\frac{K}{S} -z(1-z) G(z)}+o(1).
\end{align*}
Combining all ingredients, we get \eqref{eq_beta_cos_final}. \end{proof}

\subsection{Proofs of Theorems \ref{Theorem_basic_setting} and \ref{Theorem_4moments}}

\label{Section_Wachter_specialization}

Theorem \ref{Theorem_4moments} is a particular case of Theorem \ref{Theorem_master}, and the proof of Theorem \ref{Theorem_4moments} consists of simplifications of the formulas \eqref{eq_CCA_master_asymptotic}--\eqref{eq_beta_cos_final} in the i.i.d.\ noise situation, which we do in this section.

Take two real parameters\footnote{In order to match the notations of \cite{bao2019canonical}, we should set there $c_1=\tau_M^{-1}$ and $c_2=\tau_K^{-1}$. In order to match the notations of \cite{BG1, BG2}, we should set there $\q=\tau_K-\tau_K/\tau_M$,
$\p=\tau_K/\tau_M$.}   $\tau_K\ge \tau_M>1$ with $\tau_K^{-1}+\tau_M^{-1}<1$ and define the \emph{Wachter distribution} $\omega_{\tau_K,\tau_M}$ through its density
\begin{equation}
\label{eq_Wachter}
 \omega_{\tau_K,\tau_M}(x)\, \dd x=\frac{\tau_K}{2\pi } \frac{\sqrt{(x-\lambda_-)(\lambda_+-x)}}{x (1-x)} \mathbf 1_{[\lambda_-,\lambda_+]}\, \dd x,
\end{equation}
where the support $[\lambda_-,\lambda_+]$ of the measure is defined via
\begin{equation}
\label{eq_lambda_pm_def}
\lambda_\pm=\left(\sqrt{\tau_M^{-1}(1-\tau_K^{-1})}\pm \sqrt{\tau_K^{-1}(1-\tau_M^{-1})}  \right)^2.
\end{equation}
One can check that $0\le \lambda_-<\lambda_+<1$ for every  $\tau_K>\tau_M>1$ with $\tau_K^{-1}+\tau_M^{-1}<1$ and that \eqref{eq_Wachter} is a probability measure. A direct computation also shows that:

\begin{lemma} \label{Lemma_Wachter_modified_Stieltjes}
 The \emph{modified} Stieltjes transform of $\omega_{\tau_K,\tau_M}$ is:
\begin{equation}
\label{eq_Wachter_Stieltjes}
 \mathcal G_{\tau_K,\tau_M}(z):=\frac{1}{\tau_K} \int_{\lambda_-}^{\lambda_+} \frac{1}{z-x} \omega_{\tau_K,\tau_M}(x)\, \dd x = \frac{\tau_M^{-1}+\tau_K^{-1}-z+ \sqrt{ (z-\lambda_-)(z-\lambda_+)}}{2 z(z-1)} + \frac{1}{z \tau_K},
\end{equation}
where the branch of the square root is chosen so that for large positive $z$ the value of the square root is positive and for negative $z$ it is negative.
\end{lemma}

\begin{theorem} \label{Theorem_Wachter}
 For $K\le M <S$, let $\tilde \U$ and $\tilde \V$ be $(K-1)\times S$ and $(M-1)\times S$ random matrices, respectively, so that all their matrix elements are independent mean $0$ variance $1$ random variables with uniformly bounded $(4+\kappa)$th moments for some $\kappa>0$.  Let
 $c_1^2\ge \dots \ge c_{K-1}^2$ be squared sample canonical correlation coefficients between $\tilde \U$ and $\tilde \V$, as in Definition \ref{Definition_sample_setting}, and let $\mu^S$ be their empirical distribution:
 $$
 \mu^S=\frac{1}{K-1}\sum_{i=1}^{K-1} \delta_{c_i^2}.
 $$
 Then as $K,M,S\to\infty$, so that $S/M\to \tau_M$ and $S/K\to \tau_K$ with $\tau_K\ge \tau_M>1$, $\tau_K^{-1}+\tau_M^{-1}< 1$, we have
 $$
  \lim_{S\to\infty} \mu^S=\omega_{\tau_K,\tau_M}, \text{ weakly, in probability; and for each fixed } i\ge 1
 $$
 $$
  \lim_{S\to\infty} c_i^2=\lambda_+, \quad \lim_{S\to\infty} c_{K-i}^2=\lambda_-, \text{ in probability}.
 $$
\end{theorem}

\begin{proof} The first statement of this kind is due to \cite{wachter1980limiting}, and the exact form we use is from \cite[Corollary 2.6]{FanYang}, see also references in the latter article.
\end{proof}

\begin{corollary} \label{Corollary_Wachter}
  Let $\W=\bigl(\U^\T, \V^\T\bigr)^\T$  be $(K+M)\times S$ matrix composed of $S$ independent samples of $\bigl( \u^\T,\v^\T\bigr)^\T$, with the latter satisfying Assumption \ref{ass_4moments}. Then as $K,M,S\to\infty$, so that $S/M\to \tau_M$ and $S/K\to \tau_K$ with $\tau_K\ge \tau_M>1$, $\tau_K^{-1}+\tau_M^{-1}< 1$, the function $G(z)$ of \eqref{eq_G_def} in Theorem \ref{Theorem_master}, converges towards  $\mathcal G_{\tau_K,\tau_M}(z)$ of \eqref{eq_Wachter_Stieltjes} in probability, uniformly over $z$ in compact subsets of $\mathbb C\setminus [\lambda_-,\lambda_+]$.
\end{corollary}
\begin{proof}
 $G(z)=\frac{1}{S} \sum_{k=1}^{K-1} \frac{1}{z- c_k^2}$ in  Theorem \ref{Theorem_master} is constructed by the canonical correlations between the matrices $A \U$ and $B\V$. The latter two matrices under Assumption \ref{ass_4moments} are of the form $\tilde \U$ and $\tilde \V$ of Theorem \ref{Theorem_Wachter}. Applying this theorem and Lemma \ref{Lemma_Wachter_modified_Stieltjes}, we are done.
\end{proof}

Our next step is to analyze the behavior of the relation \eqref{eq_CCA_master_asymptotic} in the i.i.d.\ noise setting and derive the formula \eqref{eq_zrho}.
\begin{lemma} \label{Lemma_Wachter_spike_answer} Consider the relationship between $z$ and $\rho$, written using the function \eqref{eq_Wachter_Stieltjes}:
\begin{equation} \label{eq_x13}
  \frac{\displaystyle z \left[1-2\tau_K^{-1}-\frac{1}{z} \cdot (\tau_M^{-1}-\tau_K^{-1})-(1-z) \mathcal G_{\tau_K,\tau_M}(z) \right]  \left[1-\tau_K^{-1}-\tau_M^{-1}-(1 - z)  \mathcal G_{\tau_K,\tau_M}(z) \right]}{\displaystyle  \left[1 - \tau_M^{-1}-z\tau_K^{-1}  -  z(1-z)  \mathcal G_{\tau_K,\tau_M}(z)  \right]^2}
  =\rho^2.
\end{equation}
If $0\le \rho^2\le \tfrac{1}{\sqrt{(\tau_M-1)(\tau_K-1)}}$, then there is no $z>\lambda_+$ satisfying \eqref{eq_x13}. But if
\begin{equation}
 \label{eq_Wachter_cutoff}
  \rho_c^2<\rho^2 \le 1,
\end{equation}
then there is a unique $z>\lambda_+$ satisfying \eqref{eq_x13}, denoted $z_\rho$. In this situation the relationship between $\rho$ and $z_\rho$ is:
\begin{align}
\label{eq_z_Wachter} z_\rho&=\frac{\bigl(  (\tau_K-1)\rho^2  + 1 \bigr) \bigl(  (\tau_M-1) \rho^2 + 1\bigr)}{\rho^2 \tau_K \tau_M } \qquad \text{ or, equivalently, }\\
\label{eq_z_inverse_Wachter} \rho^2&= \frac{z_\rho-\tau_M^{-1}-\tau_K^{-1}+2\tau_M^{-1}\tau_K^{-1} +\sqrt{(z_\rho-\lambda_-)(z_\rho-\lambda_+)}}{2(1-\tau_M^{-1})(1-\tau_K^{-1})}.
\end{align}
\end{lemma}
\begin{proof} First, we simplify \eqref{eq_x13}. The second factor in the numerator of \eqref{eq_x13} is
\begin{multline}
\label{eq_x17}
1-2\tau_K^{-1}-\frac{1}{z} \cdot (\tau_M^{-1}-\tau_K^{-1})+ \frac{\tau_M^{-1}+\tau_K^{-1}-z+ \sqrt{ (z-\lambda_-)(z-\lambda_+)}}{2  z}  -\frac{(1-z)\tau_K^{-1}}{z}
 \\= \frac{\tau_K^{-1}-\tau_M^{-1} +z(1-2\tau_K^{-1}) +\sqrt{ (z-\lambda_-)(z-\lambda_+)}}{2  z},
\end{multline}
the third factor is
\begin{multline}
\label{eq_x18}
1-\tau_M^{-1}-\tau_K^{-1}+  \frac{\tau_M^{-1}+\tau_K^{-1}-z+ \sqrt{ (z-\lambda_-)(z-\lambda_+)}}{2 z} - \frac{(1 - z)\tau_K^{-1}}{z}
\\= \frac{\tau_M^{-1}-\tau_K^{-1}+z(1-2\tau_M^{-1})+ \sqrt{ (z-\lambda_-)(z-\lambda_+)}}{2 z},
\end{multline}
and the expression in denominator (which is being squared) is
\begin{multline}
\label{eq_x19}
1 -\tau_M^{-1}- \tau_K^{-1} z + \frac{\tau_M^{-1}+\tau_K^{-1}-z+ \sqrt{ (z-\lambda_-)(z-\lambda_+)}}{2 } - (1-z)\tau_K^{-1}
\\=  \frac{2-\tau_M^{-1}-\tau_K^{-1}-z+ \sqrt{ (z-\lambda_-)(z-\lambda_+)}}{2 }.
\end{multline}
We plug \eqref{eq_lambda_pm_def} into  \eqref{eq_x17}, \eqref{eq_x18}, \eqref{eq_x19}, and then substitute into the left-hand side of \eqref{eq_x13}. After a straightforward computation, a simplification occurs and the left-hand side of \eqref{eq_x13} transforms into
\begin{equation}
\label{eq_x12}
  \frac{z-\tau_M^{-1}-\tau_K^{-1}+2\tau_M^{-1}\tau_K^{-1} +\sqrt{(z-\lambda_-)(z-\lambda_+)}}{2(1-\tau_M^{-1})(1-\tau_K^{-1})},
\end{equation}
which matches the right-hand side of \eqref{eq_z_inverse_Wachter} if we set $z=z_\rho$.

Note that \eqref{eq_x12} is a monotonously increasing function of $z\in [\lambda_+, 1]$. Hence, for $z$ in this interval, it takes exactly once all values between the value at $z=\lambda_+$, which is
$$
  \frac{\lambda_+-\tau_M^{-1}-\tau_K^{-1}+2\tau_M^{-1}\tau_K^{-1} }{2(1-\tau_M^{-1})(1-\tau_K^{-1})}
  =
   \frac{1}{\sqrt{(\tau_M-1)(\tau_K-1)}},
$$
 and the value at $z=1$, which is
$$
  \frac{1-\tau_M^{-1}-\tau_K^{-1}+2\tau_M^{-1}\tau_K^{-1} +\sqrt{(1-\lambda_-)(1-\lambda_+)}}{2(1-\tau_M^{-1})(1-\tau_K^{-1})}
  %
   =1.
$$
Therefore, for $\rho^2$ satisfying \eqref{eq_Wachter_cutoff}, there exists a unique $z\in [\lambda_+, 1]$, solving \eqref{eq_x13} and this $z=z_\rho$ is given as a solution to \eqref{eq_z_inverse_Wachter}. Simultaneously, we have shown that for $\rho^2\le  \frac{1}{\sqrt{(\tau_M-1)(\tau_K-1)}}$, \eqref{eq_x13} does not have solutions $z\in [\lambda_+, 1]$.

In order to get \eqref{eq_z_Wachter}, we need to compute the inverse function to \eqref{eq_x12} on $z\in [\lambda_+,1]$. In other words, we treat $\rho^2$ as given, and solve \eqref{eq_z_inverse_Wachter} as an equation on $z_\rho$. After another short computation, this results in \eqref{eq_z_Wachter}.
\end{proof}
\begin{remark} \label{Remark_no_bottom_outlier}
 The $z$--derivative of \eqref{eq_x12} on $[0,\lambda_-]$ is
 $$
  \frac{1}{2(1-\tau_M^{-1})(1-\tau_K^{-1})}\left(1 -\frac{\lambda_--z+\lambda_+-z}{2 \sqrt{(\lambda_--z)(\lambda_+-z)}}\right)\le 0,
 $$
 and therefore, \eqref{eq_x12} is a decreasing function of $z$. The value of \eqref{eq_x12} at $z=0$ is
 $$
  \frac{-\tau_M^{-1}-\tau_K^{-1}+2\tau_M^{-1}\tau_K^{-1} +\sqrt{\lambda_-\lambda_+}}{2(1-\tau_M^{-1})(1-\tau_K^{-1})}=\frac{-\tau_K^{-1}}{1-\tau_K^{-1}}<0.
 $$
 Hence, all the values on $z\in [0, \lambda_-]$ are negative, and there is no $z\in[0,\lambda_-]$ satisfying \eqref{eq_x13}.
\end{remark}

We further simplify the functions $Q_x(z)$ and $Q_y(z)$ of Theorem \ref{Theorem_master} in the i.i.d.\ noise case.

\begin{lemma} Consider the limit versions of $Q_x(z)$ and $Q_y(z)$ explicitly given by:
\begin{equation}
\label{eq_x14}
  - \frac{1-2\tau_K^{-1}-\frac{1}{z} \cdot (\tau_M^{-1}-\tau_K^{-1})-(1-z)\mathcal G_{\tau_K,\tau_M}(z)}{1 - \tau_M^{-1}-z\tau_K^{-1}  -  z(1-z)  \mathcal G_{\tau_K,\tau_M}(z) }, \qquad - \frac{1-\tau_K^{-1}-\tau_M^{-1}-(1 - z) \mathcal G_{\tau_K,\tau_M}(z)}{1 - \tau_M^{-1}-z\tau_K^{-1}  -  z(1-z) \mathcal G_{\tau_K,\tau_M}(z)},
\end{equation}
If we plug $z=z_\rho$ given by \eqref{eq_z_Wachter}, then these functions become the following functions of $\rho^2> \frac{1}{\sqrt{(\tau_M-1)(\tau_K-1)}}$, respectively:
\begin{equation}
 -\frac{\rho^2 \tau_M}{\rho^2(\tau_M-1) + 1} \qquad\text{ and }\qquad -\frac{\rho^2 \tau_K}{\rho^2(\tau_K-1) + 1}.
\end{equation}
\end{lemma}
\begin{proof} We start by computing $\mathcal G_{\tau_K,\tau_M}(z_\rho)$. We have
\begin{equation}
\label{eq_x15}
 \sqrt{(z_\rho-\lambda_-)(z_\rho-\lambda_+)}=  \frac{\rho^4 (\tau_K-1)(\tau_M-1) - 1}{\rho^2 \tau_M \tau_K}
\end{equation}
by plugging \eqref{eq_z_Wachter} into the first appearance of $z_\rho$ in \eqref{eq_z_inverse_Wachter}. Thus,
\begin{multline}
\label{eq_G_z_rho}
 \mathcal G_{\tau_K,\tau_M}(z_\rho) = \frac{\tau_M^{-1}+\tau_K^{-1}-z_\rho+ \sqrt{ (z_\rho-\lambda_-)(z_\rho-\lambda_+)}}{2 z_\rho(z_\rho-1)} + \frac{1}{z_\rho \tau_K}
 \\=\frac{\rho^2 \tau_M(\tau_K-1)}{(\rho^2( \tau_K-1)(\tau_M-1) - 1)\, (\rho^2 (\tau_K-1) + 1)}.
\end{multline}
Plugging into \eqref{eq_x14}, we find that the denominators in both formulas are
$$
 1 - \tau_M^{-1}-z\tau_K^{-1}  -  z(1-z) \mathcal G_{\tau_K,\tau_M}(z)= \frac{\rho^2(\tau_K-1)(\tau_M-1) - 1}{\rho^2 \tau_M \tau_K},
$$
while the numerators are
$$
 \frac{\rho^2 (\tau_K-1)(\tau_M-1) - 1}{(\rho^2(\tau_M-1) + 1)\tau_K} \quad\text{ and }\quad \frac{\rho^2(\tau_K-1)(\tau_M-1) - 1}{(\rho^2(\tau_K-1) + 1)\tau_M}.\qedhere
$$
\end{proof}

Finally, we simplify the formulas \eqref{eq_alpha_cos_final} and \eqref{eq_beta_cos_final} in the i.i.d.\ noise case.

\begin{lemma} \label{Lemma_cos_limit_iid} In the limit $K,M,S\to\infty$, so that $S/M\to \tau_M$ and $S/K\to \tau_K$ with $\tau_K\ge \tau_M>1$, $\tau_K^{-1}+\tau_M^{-1}< 1$, in the independent noise case under Assumption \ref{ass_4moments}, we have
\begin{equation}
\label{eq_x16}
 \cos^2 \theta_x\to \frac{\tau_K(\rho^4(\tau_K-1)(\tau_M-1) - 1)}{(\rho^2(\tau_K-1) + 1)(\rho^2(\tau_K-1)(\tau_M-1) - 1)},
\end{equation}
\begin{equation}
\label{eq_x20}
\cos^2 \theta_y \to \frac{\tau_M (\rho^4(\tau_K-1)(\tau_M-1) - 1)}{(\rho^2(\tau_M-1) + 1)(\rho^2(\tau_K-1)(\tau_M-1) - 1)}.
\end{equation}
\end{lemma}
\begin{proof}
In the previous two lemmas we investigated the asymptotic behavior of all ingredients of the formulas \eqref{eq_alpha_cos_final} and \eqref{eq_beta_cos_final} expect for $G'(z)$. Differentiating \eqref{eq_Wachter_Stieltjes}, we have
\begin{align*}
 \frac{\partial}{\partial z} \mathcal G_{\tau_K,\tau_M}(z)&= \frac{-1+ \frac{2z-\lambda_--\lambda_+}{2\sqrt{ (z-\lambda_-)(z-\lambda_+)}}}{2 z(z-1)}- \frac{\tau_M^{-1}+\tau_K^{-1}-z+ \sqrt{ (z-\lambda_-)(z-\lambda_+)}}{2 z^2(z-1)}\\ &- \frac{\tau_M^{-1}+\tau_K^{-1}-z+ \sqrt{ (z-\lambda_-)(z-\lambda_+)}}{2 z(z-1)^2} - \frac{1}{z^2 \tau_K}.
\end{align*}
Using \eqref{eq_z_Wachter} and \eqref{eq_x15}, we compute
\begin{equation}
\label{eq_G_prime_z_rho}
  \frac{\partial}{\partial z} \mathcal G_{\tau_K,\tau_M}(z)\Bigr|_{z=z_\rho}=-\frac{\tau_M^2\tau_K\rho^4(\tau_K - 1)(\rho^4(\tau_K-1)^2(\tau_M-1) + 1)}{(\rho^4(\tau_K-1)(\tau_M-1) - 1)(\rho^2(\tau_K-1) + 1)^2 (\rho^2(\tau_K-1)(\tau_M-1) - 1)^2}.
\end{equation}
Next, we start plugging all the computed ingredients into \eqref{eq_alpha_cos_final}. An interesting cancelation happens: the squared factor in the first line\footnote{We recall from the proof of Theorem \ref{Theorem_master_equation} that this factor is the ratio of $\cos \theta_x$ and $\alpha_0$.} of \eqref{eq_alpha_cos_final} simplifies to 1:
$$
1-\tau_K^{-1} -z_\rho Q_x(z_\rho) \left(\tau_K^{-1}+ (1-z_\rho) \mathcal G_{\tau_K,\tau_M}(z_\rho)\right)=1.
$$
The two last lines of \eqref{eq_alpha_cos_final} transform into \eqref{eq_x16} after another compuation. The same simplification happens for the first line of \eqref{eq_beta_cos_final}:
$$
 1-\tau_M^{-1} -z_\rho Q_y(z_\rho) \left(\tau_M^{-1}+(1-z_\rho)\mathcal G_{\tau_K,\tau_M}(z_\rho)+\frac{1-z_\rho}{z_\rho}(\tau_M^{-1}-\tau_K^{-1})\right)=1.
$$
The two last lines of \eqref{eq_beta_cos_final} transform into \eqref{eq_x20}.
\end{proof}
Using $\sin^2\theta=1-\cos^2\theta$ we can further transform the answers and derive \eqref{eq_sx} and \eqref{eq_sy}.
\begin{corollary} \label{Corollary_final_angles} In the setting of Lemma \ref{Lemma_cos_limit_iid}  we have
\begin{equation}
\label{eq_x21}
 \sin^2 \theta_x\to \frac{(1-\rho^2)  (\tau_K - 1)(\rho^2(\tau_M-1) + 1)}{(\rho^2(\tau_K-1)(\tau_M-1) - 1)(\rho^2(\tau_K-1) + 1)},
\end{equation}
\begin{equation}
\label{eq_x22}
\sin^2 \theta_y \to  \frac{(1-\rho^2)  (\tau_M - 1)(\rho^2(\tau_K-1) + 1)}{(\rho^2(\tau_K-1)(\tau_M-1) - 1)(\rho^2(\tau_M-1) + 1)}.
\end{equation}
\end{corollary}

Now we have all the ingredients.

\begin{proof}[Proof of Theorem \ref{Theorem_4moments}] First, suppose that $\rho^2> \rho_c^2$, then by Lemma \ref{Lemma_Wachter_spike_answer}, $z_\rho>\lambda_+$ and the solution to the equation \eqref{eq_CCA_master_asymptotic} converges to $z_\rho$ as $S\to\infty$. We use Theorem \ref{Theorem_master} and Corollary \ref{Corollary_Wachter} and invoke Rouche's theorem to conclude that convergence of the equations implies convergence of the solutions. This implies that the largest canonical correlation $\lambda_1$ converges to $z_\rho$ as $S\to\infty$.

 Because \eqref{eq_CCA_master_asymptotic} is $S\to\infty$ approximation of the equation \eqref{eq_CCA_master} of Theorem \ref{Theorem_master_equation} solved by \emph{each} of the canonical correlations $\lambda_1\ge \lambda_2 \dots \ge \lambda_K$, and the limiting equation by Lemma \ref{Lemma_Wachter_spike_answer} has only one solution larger than $\lambda_+$, we conclude that $\limsup_{S\to\infty} \lambda_2\le \lambda_+$. On the other hand, by the interlacing inequalities of Lemma \ref{Lemma_interlacing}, $\lambda_2\ge c_2^2$, and the latter converges to $\lambda_+$ by Theorem \ref{Theorem_Wachter}, implying $\liminf_{S\to\infty} \lambda_2=\lambda_+$. We conclude that $\lambda_2$ converges in probability to $\lambda_+$ as $S\to\infty$. Hence, \eqref{eq_canonical_limit} is proven.

 The limit of the angles $\theta_x$ and $\theta_y$ is given in Theorem \ref{Theorem_master} by the formulas \eqref{eq_alpha_cos_final} and \eqref{eq_beta_cos_final}. By Corollary \ref{Corollary_final_angles}, this formulas lead to \eqref{eq_sx} and \eqref{eq_sy}. Hence, \eqref{eq_vector_limit1} and \eqref{eq_vector_limit2} are proven.

 \smallskip

 Second, suppose that  $\rho^2\le \tfrac{1}{\sqrt{(\tau_M-1)(\tau_K-1)}}$. Then by Lemma \ref{Lemma_Wachter_spike_answer}, the equation \eqref{eq_CCA_master_asymptotic}, does not have solutions larger and bounded away from $\lambda_+$ as $S\to\infty$. Hence, by Theorem \ref{Theorem_master}, $\limsup_{S\to\infty} \lambda_1=\lambda_+$. On the other hand, by Lemma \ref{Lemma_interlacing}, $\lambda_1\ge c_1^2$, and the latter converges to $\lambda_+$ by Theorem \ref{Theorem_Wachter}. We conclude that $\lambda_1$ converges in probability to $\lambda_+$ as $S\to \infty$, thus, proving \eqref{eq_no_spike}.
\end{proof}

\begin{proof}[Proof of Theorem \ref{Theorem_basic_setting}]
We would like to show that Assumption \ref{ass_basic} implies Assumption \ref{ass_4moments} and, therefore, Theorem \ref{Theorem_4moments} implies Theorem \ref{Theorem_basic_setting}.

The parts in Assumption \ref{ass_4moments} about being fourth-moment Gaussian and about the existence of the $(4+\tau)$th moments are automatic in the Gaussian setting. It remains to choose the matrices $A$ and $B$. Let us introduce a positive definite symmetric bilinear form $\Xi^x$ on vectors in $\mathbb R^K$, given by
 $$
  \Xi^u(\be,
   \bz)=\E \bigl[(\be^\T \u)  ( \bz^\T \u) \bigr].
 $$
Essentially, $\Xi^u$ is the covariance matrix of $\u$. Let $\mathfrak L^u$ be $(K-1)$--dimensional subspace in $\mathbb R^K$ consisting of vectors $\Xi^u$--orthogonal to $\ba$:
 $$
  \mathfrak L^u=\bigl\{\be\in\mathbb R^K \mid \Xi^u(\be,\ba)=0\bigr\}.
 $$
Similarly, we let $\Xi^v$ be the covariance matrix of $\v$ and let $\mathfrak L^v$ be $(M-1)$--dimensional subspace in $\mathbb R^M$ consisting of vectors $\Xi^v$--orthogonal to $\bb$. Note that the spaces $\mathfrak L^u$ and $\mathfrak L^v$ are linear spaces of the vectors $\bg$ mentioned in Assumption \ref{ass_basic}.

Choose a $\Xi^u$--orthonormal basis $\bg^{1;u},\dots, \bg^{K-1;u}$ in $\mathfrak L^{u}$:
 $$
  \Xi^u(\bg^{i;u},\bg^{j;u})=\delta_{i=j}, \quad 1\le i,j, \le K-1.
 $$
Also choose a $\Xi^v$--orthonormal basis $\bg^{1;v},\dots,\bg^{M-1;v}$ in $\mathfrak L^v$. Define $A$ to be $(K-1)\times K$ matrix whose $i$--th row is formed by the coordinates of $\bg^{i,u}$, and define $B$ to be $(M-1)\times M$ matrix whose $i$--th row is formed by the coordinates of $\bg^{i,v}$.

The conditions of Assumption \ref{ass_4moments} now follow from the conditions of Assumption \ref{ass_basic} and the fact that for Gaussian vectors being uncorrelated implies being independent. \end{proof}

\subsection{Proof of Theorem \ref{Theorem_non_iid_signal}}
The idea of the proof is to use rotational invariance of the Gaussian law to reduce Theorem \ref{Theorem_non_iid_signal} to Theorems \ref{Theorem_basic_setting} or \ref{Theorem_4moments}.

We rely on the following basic property of Gaussian distributions.
\begin{lemma} \label{Lemma_Gaussian_rotation}
 Let $\W$ be $N\times S$ random matrix, such that each column is a mean $0$ Gaussian vector (of an arbitrary covariance) and the columns are i.i.d.. In addition, let $O$ be $S\times S$ orthogonal matrix which is independent from $\W$. Then:
 \begin{enumerate}
  \item $\W O$ has the same distribution as $\W$.
  \item $\W O$ is independent from $O$.
 \end{enumerate}
\end{lemma}
\begin{proof}
 Let us condition on the value of $O$. The random vector formed by vectorizing $\W O$ is a linear transformation of the vectorization of $\W$; hence, it is Gaussian. The orthogonality of $O$ implies that the covariance structure of the matrix elements of $\W O$ is the same as the one for $\W$. Since the laws of mean $0$ Gaussian vectors are uniquely determined by the covariances, we conclude that the distributions of $\W$ and $\W O$ coincide conditionally on $O$.

 Because the law of $\W O$ is the same for any choice of $O$, we also conclude the independence between $\W O$ and $O$.
\end{proof}

Throughout the proof of Theorem \ref{Theorem_non_iid_signal} we condition on the vectors $\x=\U^\T \ba$ and $\y =\V^\T \bb$ in Assumption \ref{ass_cor_signal} and treat them as deterministic. Because the canonical correlations and corresponding angles between true and estimated canonical variables are unchanged when we rescale $\x$ or $\y$, we can and will assume without loss of generality that they are normalized so that $\x^\T \x=\y^\T \y=S$. Therefore, also
$$
 \x^\T \y= \pm\, S\, \hat r,
$$
where $\hat r$ is given in Eq.~\eqref{eq_sample_r}.
The choice of the sign in the last formula is merely a question of the definition of $\hat r$, because \eqref{eq_sample_r} only fixes its square. Hence, we assume without loss of generality that the sign is $+$.

We introduce an additional uniformly random $S\times S$ orthogonal matrix $O$, which is independent from the rest of the data in Theorem \ref{Theorem_non_iid_signal}. Note that an orthogonal transformation in $S$--dimensional space does not change the scalar products, lengths, and angles between vectors. Therefore, is we replace in Theorem \ref{Theorem_non_iid_signal} the matrices $\U$ and $\V$ with $\U O$ and $\V O$, respectively, then the distributions of squared canonical correlations between $\U$ and $\V$ and corresponding angles between true and estimated canonical variables is unchanged. Hence, it is sufficient to establish the conclusion of Theorem \ref{Theorem_non_iid_signal} for data matrices $\U O$ and $\V O$.

For the $A$ and $B$ of Assumption \ref{ass_cor_signal}, by Lemma \ref{Lemma_Gaussian_rotation},  $A \U O$ and $B \V O$ are matrices of i.i.d.\ Gaussian random variables and they are independent from $O$ and hence from the vectors $O^\T \x =O^\T \U^\T \ba$ and $O^\T \y =O^\T \V^\T \bb$. Therefore, the noise part $A \U O$ and $B \V O$ matches the i.i.d.\ setting of Theorems \ref{Theorem_basic_setting} and \ref{Theorem_4moments}.

Let us look at the signal part: the pair of the $S$--dimensional vectors $(O^\T \x, O^\T \y)$ is obtained from the pair of vectors $(\x,\y)$ by applying uniformly random orthogonal rotation matrix. Denote $(\u^*,\v^*)=(O^\T \x, O^\T \y)$. At this point, the only difference between the present setting and the one of Theorem \ref{Theorem_basic_setting} is whether $\u^*$ and $\v^*$ have i.i.d.\ Gaussian components or not. Recall that the only place where the i.i.d.\ Gaussian assumption on the signal vector was used on our path to Theorem \ref{Theorem_basic_setting} through Theorems \ref{Theorem_master_equation} and \ref{Theorem_master} is in Lemma \ref{Lemma_asymptotic_approx}. Hence, to finish the proof of Theorem \ref{Theorem_non_iid_signal} we show the following generalization:
\begin{lemma} Set $C_{uu}=C_{vv}=1$, $C_{uv}=C_{vu}=\hat r$. Then the asymptotic approximations of Lemma \ref{Lemma_asymptotic_approx} remain true for $(\u^*,\v^*)\stackrel{d}{:=}(O^\T \x, O^\T \y)$.
\end{lemma}
\begin{proof}
 The random vectors $O^\T \x$ and $O^\T \y$ do not have i.i.d.\ components, however, we can construct them from vectors with i.i.d.\ components. For that, let us write $\y=\hat r \x + \sqrt{1-\hat r^2}\, \boldsymbol \nu$, where $\boldsymbol \nu$ is a vector orthogonal to $\x$. Our normalizations imply that the squared length of $\boldsymbol \nu$ is $S$. Observe that $(O^\T \x, O^\T \boldsymbol \nu)$ is a uniformly random pair of orthogonal vectors of length $\sqrt{S}$ in $S$--dimensional space. Here is an alternative way to construct such a pair: Take two independent vectors $\boldsymbol \xi$ and $\boldsymbol \psi$ with i.i.d.\ $\mathcal N(0,1)$ components, represented as $S\times 1$ matrices. Set
 \begin{equation}
 \label{eq_x23}
  \widetilde \x = \sqrt{S} \frac{\boldsymbol \xi}{\sqrt{\boldsymbol \xi^\T\boldsymbol \xi}}, \qquad \widetilde {\boldsymbol \nu} = \sqrt{S}   \frac{\boldsymbol \psi - \boldsymbol \xi \frac{\boldsymbol \xi^\T\boldsymbol \psi}{\boldsymbol \xi^\T\boldsymbol \xi}}{\sqrt{\bigl(\boldsymbol \psi - \boldsymbol \xi \frac{\boldsymbol \xi^\T\boldsymbol \psi}{\boldsymbol \xi^\T\boldsymbol \xi}\bigr)^\T\bigl(\boldsymbol \psi - \boldsymbol \xi \frac{\boldsymbol \xi^\T\boldsymbol \psi}{\boldsymbol \xi^\T\boldsymbol \xi}\bigr)}}.
 \end{equation}
 The invariance of the i.i.d.\ Gaussian vectors $\boldsymbol \xi$ and $\boldsymbol \psi$ under orthogonal transformations readily implies that $(\widetilde \x, \widetilde {\boldsymbol \nu})$ has the same distribution as $(O^\T \x, O^\T \boldsymbol \nu)$. Hence, we can also write
 \begin{equation}
 \label{eq_x24}
   (O^\T \x,\, O^\T \y)\stackrel{d}{=} \left(\widetilde \x,\, \hat r \widetilde \x + \sqrt{1-\hat r^2}\, \widetilde{\boldsymbol \nu}\right).
 \end{equation}
 Note that all the random constants appearing in \eqref{eq_x23} have straightforward deterministic limits by the law of large numbers for i.i.d.\ random variables: as  $S\to \infty$
 $$
  \frac{\sqrt{S}}{\sqrt{\boldsymbol \xi^\T\boldsymbol \xi}}\to 1 ,\qquad \frac{\boldsymbol \xi^\T\boldsymbol \psi}{\boldsymbol \xi^\T\boldsymbol \xi}\to 0, \qquad \frac{\sqrt{S}}{\sqrt{\bigl(\boldsymbol \psi - \boldsymbol \xi \frac{\boldsymbol \xi^\T\boldsymbol \psi}{\boldsymbol \xi^\T\boldsymbol \xi}\bigr)^\T\bigl(\boldsymbol \psi - \boldsymbol \xi \frac{\boldsymbol \xi^\T\boldsymbol \psi}{\boldsymbol \xi^\T\boldsymbol \xi}\bigr)}} \to 1.
 $$
 Hence, we have the following chain of reductions: the conclusion of Lemma \ref{Lemma_asymptotic_approx} holds for $(\u^*,\v^*)\stackrel{d}{:=} (\boldsymbol \xi,\boldsymbol \psi)$ with $C_{uu}=C_{vv}=1$, $C_{uv}=C_{vu}=0$, therefore, it also holds for $(\u^*,\v^*)\stackrel{d}{:=} (\widetilde \x, \widetilde {\boldsymbol \nu})$ with $C_{uu}=C_{vv}=1$, $C_{uv}=C_{vu}=0$, and therefore, it also holds for $(\u^*,\v^*)\stackrel{d}{:=}  (O^\T \x,\, O^\T \y)$  with $C_{uu}=C_{vv}=1$, $C_{uv}=C_{vu}=\hat r$.
\end{proof}

\subsection{Proof of Theorem \ref{Theorem_basic_multi}} \label{Section_LLN_multi} By the same argument as in Lemma \ref{Lemma_simple_data}, we can assume without loss of generality that the signal vectors are represented by the first $\mathbbm q$ rows in $\U$ and by the first $\mathbbm q$ rows in $\V$. The remaining $K-\mathbbm q$ rows in $\U$ and $M-\mathbbm q$ rows in $\V$ represent the noise.

Next, we apply Theorem \ref{Theorem_master} $\mathbbm q$ times: we start from $(K-\mathbbm q)\times S$ and $(M-\mathbbm q)\times S$ matrices $\tilde \U$ and $\tilde \V$ representing noise and then add the rows representing signal one by one. We claim that after the addition of $q\le \mathbbm q$ signals, the squared canonical correlations $c_1^2\ge \dots\ge c^2_{K-\mathbbm q+q}$ appearing in \eqref{eq_G_def} have the three following features as $S\to\infty$: fix arbitrary small $\eps>0$,
\begin{enumerate}
 \item All but finitely many squared canonical correlations belong to the segment ${[\lambda_--\eps,\lambda_++\eps]}$ as $S\to\infty$ and their empirical distribution converges to the Wachter law, as in Theorem \ref{Theorem_Wachter};
 \item There might be several outliers in $\eps$--neighborhoods of points $z_{\rho[1]},\dots, z_{\rho[q]}$ corresponding to those $\rho[i]$, $1\le i \le q$, which are larger than $\tfrac{1}{\sqrt{(\tau_M-1)(\tau_K-1)}}$;
 \item There are no other squared canonical correlations outside $[\lambda_--\eps,\lambda_++\eps]$ beyond those described in $(2)$ above.
\end{enumerate}
The validity of these features for all $q=0,1,\dots,\mathbbm q$ is proven inductively in $q$: for $q=0$ we are in the pure noise situation and use Theorem \ref{Theorem_Wachter}. For the inductive step, note that the first feature follows by applying Lemma \ref{Lemma_interlacing}, which guarantees that the empirical distribution does not change much on each step of adding another pair of signal vectors. The second and third features follow from the first one: away from the outliers the function $G(z)$ of \eqref{eq_G_def} converges towards $\mathcal G_{\tau_K,\tau_M}(z)$ and then the results of Lemma \ref{Lemma_Wachter_spike_answer} and Remark \ref{Remark_no_bottom_outlier} applied to the equation \eqref{eq_CCA_master_asymptotic} of  Theorem \ref{Theorem_master} give the location of the next possible outlier.

We conclude that in Theorem \ref{Theorem_basic_multi}, any squared sample canonical correlations, which remain outside $[\lambda_--\eps,\lambda_++\eps]$, should converge to one of the numbers $z_{\rho[q]}$ as $S\to\infty$. It remains to show that for each $q$ such that $\rho^2[q]> \tfrac{1}{\sqrt{(\tau_M-1)(\tau_K-1)}}$, there is exactly one canonical correlation converging to $z_{\rho[q]}$ and that formulas  \eqref{eq_vector_limit1_multi}, \eqref{eq_vector_limit2_multi} hold. For that notice that in the above inductive procedure the order of addition of the rows representing the signals does not matter for the final result. In particular, we can assume that the rows corresponding to $\rho^2[q]$ are the last ones to be added; in this situation, we had no outliers in $\eps$--neighborhood of $z_{\rho[q]}$ before the last step\footnote{Here we use the fact that $\rho^2[1],\dots,\rho^2[\mathbbm q]$ are all distinct.}, and the result follows by Theorem \ref{Theorem_master} with the answers simplified through $G(z)\approx \mathcal G_{\tau_K,\tau_M}(z)$ by the computations of Appendix \ref{Section_Wachter_specialization}.

We finish the proof of the Gaussian fluctuations part of Theorem \ref{Theorem_basic_multi} in Section \ref{Section_CLT_multiple}.

\section{Gaussian fluctuations of the angles and the outlier}\label{Section_appendix_angle_CLT}

In this section we prove Theorem \ref{Theorem_joint_CLT_main}.

\subsection{Detailed statement}
\label{Section_detailed_CLT}

In Theorem \ref{Theorem_basic_setting} the asymptotic parameters were subject to the assumption $S/K\to \tau_K$, $S/M\to\tau_M$, $r^2\to \rho^2$, while for the Gaussian convergence statements we need to assume a fast enough speed of convergence, as otherwise the corrections contribute to the answers. There are two ways to do this. One option is to assume all parameters to be $S$--dependent, i.e.\ we simply set
\begin{equation}
\label{eq_params_exact}
 \frac{S}{K}=\tau_K=\tau_K(S), \qquad \frac{S}{M}=\tau_M=\tau_M(S), \qquad r^2=\rho^2=\rho^2(S).
\end{equation}
After such redefinition, the statement of Theorem \ref{Theorem_basic_setting} continues to hold, if we replace \eqref{eq_canonical_limit}, \eqref{eq_vector_limit1}, \eqref{eq_vector_limit2}, \eqref{eq_no_spike} by the statements that the difference between the left-hand side and right-hand side tends to $0$. We also need to replace the inequalities on $\tau_K,\tau_M,\rho^2$ with quantitative versions by introducing an $S$--independent $\eps>0$ and assuming $1+\eps<\tau_K<\eps^{-1}$, $1+\eps<\tau_M<\eps^{-1}$, $\rho^2>\rho^2_c+\eps$ (in the super-critical case) for all large enough $S$.

Alternatively, we can have $S$--independent $\tau_K, \tau_M, \rho^2$ if we assume
\begin{equation}
\label{eq_params_exact_almost}
 \frac{S}{K}=\tau_K+o\left(S^{-1/2}\right), \qquad \frac{S}{M}=\tau_M+o\left(S^{-1/2}\right), \qquad r^2=\rho^2+o\left(S^{-1/2}\right), \qquad S\to\infty.
\end{equation}

We follow the rank one Gaussian setting of Assumption \ref{ass_basic} and use notations of Section \ref{Section_setting}.

\begin{theorem}\label{Theorem_joint_CLT}
Under assumptions of Theorem \ref{Theorem_basic_setting}, and additionally  \eqref{eq_params_exact_almost},  define
\begin{equation}\label{eq_Z_def}
        Z_x=\sqrt S\bigl(\sin^2\theta_x-\s_{x}\bigr),
        \qquad
        Z_y=\sqrt S\bigl(\sin^2\theta_y-\s_{y}\bigr),
        \qquad
        Z_z=\sqrt S\bigl(\lambda_1-z_\rho\bigr).
\end{equation}
For $\rho_c^2<\rho^2<1$ the random vector $(Z_x,Z_y,Z_z)$ is asymptotically Gaussian: in distribution
\begin{equation}\label{eq_joint_CLT_appendix}
       \lim_{S\to\infty}  \bigl(
        Z_x, Z_y, Z_z\bigr)^{\T}
        \stackrel{d}{=}
        \mathcal N(0,\Sigma),
\end{equation}
where the asymptotic covariance $\Sigma$ of $(Z_x,Z_y,Z_z)$ is given by  \eqref{eq_cov_Sigma_inverse}, \eqref{eq_covariance_factorized} or, equivalently,
\eqref{eq_sigmax}, \eqref{eq_sigmay}, \eqref{eq_sigmaz}, \eqref{eq__sigmaxz}, \eqref{eq_sigmayz}, \eqref{eq_sigmaxy}.
\end{theorem}

The $3\times 3$ covariance matrix $\Sigma$ is defined through
\begin{equation}
\label{eq_covariance_def}
 \Sigma:=\lim_{S\to\infty} \E \begin{pmatrix}
        Z_x\\ Z_y\\ Z_z
        \end{pmatrix} \begin{pmatrix}
        Z_x& Z_y& Z_z
        \end{pmatrix}=        \begin{pmatrix}
        \sigma_{xx}&\sigma_{xy}&\sigma_{xz}\\
        \sigma_{xy}&\sigma_{yy}&\sigma_{yz}\\
        \sigma_{xz}&\sigma_{yz}&\sigma_{zz}
        \end{pmatrix}.
\end{equation}
The matrix elements admit explicit but complicated rational expressions through parameters \eqref{eq_params_exact}, see formulas \eqref{eq_sigmax}--\eqref{eq_sigmaxy} below. The determinant $\det\Sigma$ and the inverse matrix $\Sigma^{-1}$ turn out to be simpler. We introduce notational shortcuts:
\begin{equation}\label{eq_ab_def}
        a=\tau_K-1,\qquad b=\tau_M-1,\qquad q=\rho^2,  \qquad \Delta=abq-1, \qquad
        \Gamma=abq^2-1 .
\end{equation}
The supercritical condition $\rho^2>\rho_c^2$ becomes $\Gamma>0$, while \eqref{eq_zrho}, \eqref{eq_sx}, \eqref{eq_sy} turn into:
\begin{equation}\label{eq_zqsxsy_CLT}
        z_\rho=\frac{(aq+1)(bq+1)}{q(a+1)(b+1)}, \qquad
        \s_x=\frac{a(1-q)(bq+1)}{\Delta(aq+1)},
        \qquad
        \s_y=\frac{b(1-q)(aq+1)}{\Delta(bq+1)}.
\end{equation}
In this notation, we have a very short formula for the determinant:
\begin{equation}\label{eq_detSigma}
        \det\Sigma=\frac{16a^2b^2(1-q)^6\Gamma}{q\Delta^4}.
\end{equation}
Further quantities entering the definition of $\Sigma$ are:
\begin{equation}\label{eq_dxdy_def}
        d_x=\frac{2abq^2(a+1)^2(b+1)(q-1)}{(aq+1)\Delta^2\Gamma},\qquad
        d_y=\frac{2abq^2(a+1)(b+1)^2(q-1)}{(bq+1)\Delta^2\Gamma},
\end{equation}
and a $3\times 3$ matrix
\begin{multline} \label{eq_cov_Sigma_inverse}
\widetilde \Sigma^{-1}
\\=
 \begin{pmatrix}
\frac{(aq+1)^2\Delta^2(2ab^2q^3+b^2q^2+1)}
{4ab(1-q)^2 q(a+1)^2(bq+1)^2\Gamma}
&
-\frac{\Delta^2(3abq^2+aq+bq-1)}
{4ab(1-q)^2 q(a+1)(b+1)\Gamma}
&
\frac{(b+1)(ab^2q^2+2abq^2-3abq-bq+1)}{4ab(1-q)^2 (bq+1)^2}
\\
\\
-\frac{\Delta^2(3abq^2+aq+bq-1)}
{4ab(1-q)^2 q(a+1)(b+1)\Gamma}
&
\frac{(bq+1)^2\Delta^2(2a^2bq^3+a^2q^2+1)}
{4ab(1-q)^2 q(b+1)^2(aq+1)^2\Gamma}
&
\frac{(a+1)(a^2bq^2+2abq^2-3abq-aq+1)}{4ab(1-q)^2 (aq+1)^2}
\\
\\
\frac{(b+1)(ab^2q^2+2abq^2-3abq-bq+1)}{4ab(1-q)^2 (bq+1)^2}
&
\frac{(a+1)(a^2bq^2+2abq^2-3abq-aq+1)}{4ab(1-q)^2 (aq+1)^2}
&
\frac{q(a+1)^2(b+1)^2\widetilde P_{zz}}
{4ab(1-q)^2 (aq+1)^2(bq+1)^2\Delta^2\Gamma}
\end{pmatrix},
\end{multline}
where
\begin{multline}
\widetilde P_{zz} = 1 - 6ab q + 3ab(3ab + 2) q^2 + 2ab \bigl(ab(a+b) + a + b - 6ab\bigr) q^3 \\
+ (ab)^2 \bigl(1 - 8(a+b) - 2ab\bigr) q^4 + 2(ab)^2 \bigl(ab(a+b) + a + b + ab\bigr) q^5 + (ab)^4 q^6.
\end{multline}

In terms of $\widetilde \Sigma$ and $d_x$, $d_y$, the covariance matrix admits a factorized form:
\begin{equation}\label{eq_covariance_factorized}
        \Sigma^{-1}=   \begin{pmatrix}
        1&0&0\\
        0&1& 0\\
        -d_x&-d_y&1
        \end{pmatrix}\widetilde\Sigma^{-1}    \begin{pmatrix}
        1&0&-d_x\\
        0&1&-d_y\\
        0&0&1
        \end{pmatrix}.
\end{equation}

Alternatively, we can record the rational matrix elements of $\Sigma$ directly:
\begin{equation}\label{eq_sigmax}
 \sigma_{xx}=
 \frac{2a(a+1)^2(1-q)^2}{(aq+1)^4\Delta^4}
 \left[
 \frac{(bq+1)^2\Delta^2 R_x(q)}{\Gamma}
 -q(b+1)T_x(q)
 \right],
\end{equation}
with the polynomials entering into the $x$-angle variance given by
\begin{align}
\label{eq_Rx_Tx}
        R_x(q)=&1-2aq+a(a+b)q^2+6a^2bq^3+a^3bq^4,\\
\notag  T_x(q)=&2a-a(a-b)(b+1)q
 +2ab(a^2+ab-3a-b)q^2 \\
 &+a^2b\bigl((7b-1)a+b(b-7)\bigr)q^3
 +2a^3b^2(2b-1)q^4. \notag
\end{align}
The $y$-angle variance is obtained by exchanging the roles of $a$ and $b$:
\begin{equation}\label{eq_sigmay}
        \sigma_{yy}=\sigma_{xx}\big|_{a\leftrightarrow b}.
\end{equation}
The outlier eigenvalue variance is the simplest:
\begin{equation}\label{eq_sigmaz}
        \sigma_{zz}=
        \frac{2(1-q)^2(2abq+a+b)\Gamma}{q^2(a+1)^2(b+1)^2}.
\end{equation}
The two angle--outlier covariances are
\begin{multline}\label{eq__sigmaxz}
        \sigma_{xz}
        =
        -\frac{2a(1-q)^2}{q(b+1)(aq+1)^2\Delta^2} \Bigl[
1+b(b+1)q
+\bigl(a^2b+2ab^2-2ab-b^2\bigr)q^2\\
+ab\bigl(a(3b-1)-4b\bigr)q^3
+a^2b^2(b-2)q^4\Bigr],
\end{multline}
and, by exchanging the roles of $a$ and $b$
\begin{equation}\label{eq_sigmayz}
        \sigma_{yz}
        =\sigma_{xz}\big|_{a\leftrightarrow b}.
\end{equation}
Finally,
\begin{equation}\label{eq_sigmaxy}
        \sigma_{xy}
        =
        \frac{2abq(1-q)^2(a+1)(b+1)}
        {(aq+1)^2(bq+1)^2\Delta^4\Gamma}
        R_{xy},
\end{equation}
where
\begin{equation}\label{eq_AppE_Rxy}
\begin{aligned}
R_{xy}
=&a+b+4abq
+\bigl(a^3b-3a^2b+a^2+ab^3-3ab^2+b^2\bigr)q^2
\\
&+4ab\bigl(a^2b-a^2+ab^2-4ab+a-b^2+b\bigr)q^3\\
&+ab\bigl(a^3b+10a^2b^2-13a^2b+a^2
+ab^3-13ab^2+10ab+b^2\bigr)q^4\\
&+4a^2b^2\bigl(a^2b+ab^2-3ab+a+b\bigr)q^5
+a^3b^3(2ab-a-b+2)q^6.
\end{aligned}
\end{equation}

\subsection{Outline of the proof.}
The rest of this appendix proves Theorem \ref{Theorem_joint_CLT}. The basic idea is to perform second-order expansion in the formulas of Theorem \ref{Theorem_master_equation}. For Theorem \ref{Theorem_basic_setting} we applied the Law of Large Numbers to the sums in Lemma \ref{Lemma_asymptotic_approx}, and now we need to apply the Central Limit Theorem instead. The Gaussian fluctuations are coming from randomness in scalar products $\langle \u^*, \u_i\rangle^2$, $\langle \v^*, \u_i\rangle^2$, $\langle \u^*, \v_j\rangle^2$, $\langle \v^*, \v_j\rangle^2$, while the fluctuations of the eigenvalues $c_i$ are of smaller order and do not contribute to the leading Gaussian asymptotics.

The hardest part of the proof of Theorem \ref{Theorem_joint_CLT} is computational: application of CLT to various ingredients in Theorem \ref{Theorem_joint_CLT} gives rise to a $15$--dimensional Gaussian vector with a complicated covariance. Three linear combinations of the $15$ coordinates eventually lead to $(Z_x,Z_y,Z_z)$. As an important simplification step, we show already on an early stage of the proof that all covariances of interest are rational functions of $a$, $b$, and $q$. Hence, all later steps become exact algebraic manipulations with rational functions. The overall computational complexity of the argument is the main reason, why we present it only in the setting of Theorem \ref{Theorem_basic_setting}, rather than more advanced cases of Theorem \ref{Theorem_4moments} and Theorem \ref{Theorem_master}. We believe that the asymptotic Gaussianity of $(Z_x,Z_y,Z_z)$ continues to hold in the settings of all the extensions of Theorem \ref{Theorem_basic_setting} discussed in previous sections. Yet, the covariance matrix $\Sigma$ changes, computations become heavier, and we do not discuss them in this paper; see Section \ref{Section_appendix_MC} for numerical simulations. We do not pursue a generalization of Theorem \ref{Theorem_joint_CLT} to the non-Gaussian setting of Theorem \ref{Theorem_4moments}: without Gaussianity certain independent terms in the proof (e.g., the scalar products $\langle \u^*, \u_i\rangle^2$, $i=1,2,\dots$) become only uncorrelated which would make application of the Central Limit Theorem much trickier.

The factorized form \eqref{eq_covariance_factorized} appears naturally in our arguments. It reflects two components contributing to the fluctuations of the angles $\theta_x$ and $\theta_y$ through Theorem \ref{Theorem_master_equation}: fluctuations of the outlier eigenvalue $z$ (obtained as a solution to \eqref{eq_CCA_master}) and fluctuations of the sums \eqref{eq_CCA_alpha}--\eqref{eq_CCA_beta_Q} conditional on $z$. In particular, $d_x$ and $d_y$ in the triangular matrix in \eqref{eq_covariance_factorized} are the partial derivatives of $\s_x$ and $\s_y$ in $z$.

\subsection{The $15$--dimensional vector}\label{subsec_15_vector}
We encode various parts  of the master equation of Theorem \ref{Theorem_master_equation} in a vector
\begin{equation}\label{eq_augW_def}
 \mathcal W(z)=(L,P,Q,U,N_0,N_1,D_{00},D_{01},D_{11},V,M_0,M_1,E_{00},E_{01},E_{11}),
\end{equation}
whose components we define in this section. The first three components correspond to the equation \eqref{eq_CCA_master}, written in the compressed form:
\begin{equation}\label{eq_F_def}
        F\bigl(\mathcal W(z),z\bigr)=L^2-zPQ=0.
\end{equation}
Namely:
\begin{align*}
 L&=\frac{1}{S} \left[\langle \u^*, \v^*\rangle + \sum_{j=1}^{M-1}  \frac{\langle\u^*, \v_j\rangle(c_j \langle \v^*, \u_j\rangle -z \langle \v^*, \v_j\rangle)}{z-c_j^2} -z  \sum_{i=1}^{K-1}    \frac{\langle\u^*, \u_i\rangle(\langle \v^*, \u_i\rangle - c_i \langle \v^*, \v_i\rangle)}{z - c_i^2}   \right],\\
 P&=\frac{1}{S}\left[-\langle\u^*,\u^*\rangle+\sum_{j=1}^{M-1}\frac{  \langle \u^*, \v_j \rangle^2 -2 c_j  \langle\u^*, \v_j\rangle  \langle \u^*, \u_j\rangle }{z-c_j^2} +z\sum_{i=1}^{K-1}   \frac{\langle \u^*, \u_i\rangle^2}{z - c_i^2} \right],
\\
 Q&=\frac{1}{S} \left[-\langle\v^*,\v^*\rangle+\sum_{i=1}^{K-1}  \frac{\langle \v^*, \u_i\rangle^2 - 2  c_i \langle \v^*, \u_i\rangle \langle \v^*, \v_i\rangle}{z -c_i^2} +z\sum_{j=1}^{M-1} \frac{\langle\v^*,\v_j\rangle^2}{z - c_j^2}\right].
\end{align*}
The next block of components of $\mathcal W$ enters the formula for the angle $\theta_x$. We recall that the $\U$-side sample canonical variable in Theorem \ref{Theorem_master_equation} is $\widehat \x=\alpha_0\u^*+
        \sum_{i=1}^{K-1}\alpha_i\u_i$ and from \eqref{eq_2equations_solved} in Theorem \ref{Theorem_master_equation} we have
\begin{equation}
        \frac{1}{\alpha_0} \widehat\x=\u^*+\nu_0+\frac{P}{L}\nu_1,
\end{equation}
where
\begin{equation}
 \nu_0=\sum_{i=1}^{K-1}
       {c_i\langle \u^*,\v_i\rangle-z\langle \u^*,\u_i\rangle\over z-c_i^2}\,\u_i, \qquad
 \nu_1=\sum_{i=1}^{K-1}
       {-z\bigl(\langle \v^*,\u_i\rangle-c_i\langle \v^*,\v_i\rangle\bigr)\over z-c_i^2}\,\u_i.
\end{equation}
Define
\begin{equation}
        U=\frac{1}{S}\langle \u^*,\u^*\rangle,
        \qquad
        N_i=\frac{1}{S}\langle \u^*,\nu_i\rangle,
        \qquad
        D_{ij}=\frac{1}{S} \langle \nu_i,\nu_j\rangle,
        \qquad i,j\in\{0,1\}.
\end{equation}
Put
\begin{equation}
\label{eq_x28}
        B_x=N_0+\frac{P}{L}N_1,
        \qquad
        D_x=D_{00}+2\frac{P}{L}D_{01}+\frac{P^2}{L^2}D_{11}.
\end{equation}
Then the result of Theorem \ref{Theorem_master_equation} gets encoded in terms of $\mathcal W(z)$ vector as:
\begin{equation}\label{eq_cosx_W}
        \cos^2\theta_x= \frac{\bigl\langle \u^*, \frac{1}{\alpha_0}\widehat \x\bigr\rangle^2}{\bigl\langle \u^*, \u^*\bigr\rangle\bigl\langle \frac{1}{\alpha_0}\widehat \x,\frac{1}{\alpha_0}\widehat \x \bigr\rangle}
        ={(U+B_x)^2\over U(U+2B_x+D_x)}, \qquad \text{with }z \text{ found from \eqref{eq_F_def}}.
\end{equation}
For the $\V$-side sample canonical variable, we similarly introduce
\begin{equation}
 \omega_0=\sum_{i=1}^{M-1}
       {c_i\langle \v^*,\u_i\rangle-z\langle \v^*,\v_i\rangle\over z-c_i^2}\,\v_i, \qquad
 \omega_1=\sum_{i=1}^{M-1}
       {-z\bigl(\langle \u^*,\v_i\rangle-c_i\langle \u^*,\u_i\rangle\bigr)\over z-c_i^2}\,\v_i.
\end{equation}
\begin{equation}
        V=\frac{1}{S}\langle \v^*,\v^*\rangle,
        \qquad
        M_i=\frac{1}{S}\langle \v^*,\omega_i\rangle,
        \qquad
        E_{ij}=\frac{1}{S}\langle \omega_i,\omega_j\rangle,
        \qquad i,j\in\{0,1\}.
\end{equation}
\begin{equation}
        B_y=M_0+\frac{Q}{L}M_1,
        \qquad
        D_y=E_{00}+2\frac{Q}{L}E_{01}+\frac{Q^2}{L^2}E_{11}.
\end{equation}
The $\V$--side angle gets encoded in terms of $\mathcal W(z)$ as
\begin{equation}\label{eq_cosy_W}
        \cos^2\theta_y
        ={(V+B_y)^2\over V(V+2B_y+D_y)}, \qquad \text{with }z \text{ found from \eqref{eq_F_def}}.
\end{equation}

\subsection{Asymptotics of $\mathcal W(z)$} \label{Section_asymptotics_W}

We are going to relate the desired $(Z_x,Z_y,Z_z)$ vector to $\mathcal W(z)$ and then use LLN and CLT for the latter. We start from the second part. Let us recall the probabilistic setting under which we prove Theorem \ref{Theorem_joint_CLT}:

\begin{enumerate}
 \item $(\u^*,\v^*)$ is $S\times 2$ mean $0$ Gaussian matrix with i.i.d.\ rows, each being a two-dimensional Gaussian vector of covariance\footnote{The outcome of Theorem \ref{Theorem_master_equation} is unchanged under rescaling of $\u^*$ or $\v^*$, and therefore there is no loss of generality in setting $C_{uu}=C_{vv}=1$.}
     $$
     \begin{pmatrix} C_{uu} & C_{uv}\\ C_{vu} & C_{vv}\end{pmatrix}=  \begin{pmatrix} 1 & \sqrt{q}\\ \sqrt{q} & 1\end{pmatrix}.
     $$
 \item $\u_i$ are orthonormal $S$--dimensional vectors, $\v_j$ are orthonormal $S$--dimensional vectors, $\langle \u_i,\v_j\rangle=c_i\delta_{i=j}$. All these vectors are independent of $(\u^*,\v^*)$.

 \item $c_j=0$ for $j\ge K$, while $c_1^2\ge \dots \ge c_{K-1}^2$ converge towards the Wachter distribution as $S\to\infty$, as in Theorem \ref{Theorem_Wachter}. In particular, using the Stieltjes transform notation of Lemma \ref{Lemma_Wachter_modified_Stieltjes}, we have for $z\in\mathbb C\setminus [\lambda_-,\lambda_+]$
     \begin{equation}
     \label{eq_Wachter_remainder}
      \frac{1}{S}\sum_{i=1}^{K-1}\frac{1}{z-c_i^2}=\mathcal G_{\tau_K,\tau_M}(z)+ O(S^{-1}), \qquad S\to\infty.
     \end{equation}
     It is important for us that the remainder is small: our arguments require $o(S^{-1/2})$ rather than merely $o(1)$.
     $O(S^{-1})$ term in \eqref{eq_Wachter_remainder} should be interpreted as $S^{-1} \xi$, where all absolute moments $\E |\xi|^k$, $k=1,2,\dots,$ for a random variable $\xi$ stay bounded as $S\to\infty$.
     This holds for $z$ bounded away from the support of the Wachter law: there are many statements of this kind in the literature, see e.g.,
     \citet[Proposition 3.5]{bourgade2022optimal}, \cite[Remark 7.8]{huang2024convergence}, \citet[Lemma B.2 on arXiv or Lemma S5.2 in the supplement to the published version]{bao2019canonical}, and \citet[Theorem 2.14]{FanYang}.
     The convergence in \eqref{eq_Wachter_remainder} is locally-uniform over complex $z$ outside the Wachter support, and therefore the statement can be differentiated in $z$ any number of times.
\end{enumerate}

\begin{proposition} \label{Proposition_W_limit}
 Using the notations \eqref{eq_ab_def}, \eqref{eq_zqsxsy_CLT}, under the super-criticality assumption $\Gamma>0$, we have:
 \begin{equation}
 \label{eq_W_asymptotics}
  \lim_{S\to\infty} \mathcal W(z_\rho)= m_*, \qquad  \lim_{S\to\infty} \frac{\partial}{\partial z}\mathcal W(z)\Bigr|_{z=z_\rho} =m'_*,\qquad \lim_{S\to\infty} \sqrt{S}\bigl(\mathcal W(z_\rho)-m_*\bigr)\stackrel{d}{=}{\mathcal N}_{\mathcal W}.
 \end{equation}
 The first two convergences are in probability, and the third one is in distribution.
 $\mathcal N_{\mathcal W}$ is a $15$--dimensional mean $0$ Gaussian vector, whose covariance structure is explained in Section \ref{Section_W_covariance}; $m_*$ and $m'_*$ are deterministic vectors, whose coordinates using the labeling \eqref{eq_augW_def} are:
\begin{equation} \label{eq_W_limit}
 \begin{gathered}
 L_*={\Delta\over \sqrt q\,(a+1)(b+1)},\qquad
 P_*=-{\Delta\over (a+1)(bq+1)},\qquad
 Q_*=-{\Delta\over (b+1)(aq+1)},\\
 U_*=1,
 \qquad
 N_{0,*}=-{1\over a+1},
 \qquad
 N_{1,*}=-{bq+1\over \sqrt q\,(a+1)(b+1)},\\
 D_{00,*}={a(bq^2+1)\over (a+1)\Gamma},
 \qquad
 D_{01,*}={a\sqrt q\,(bq+1)^2\over (a+1)(b+1)\Gamma},
 \qquad
 D_{11,*}={a(bq+1)^2\over (a+1)(b+1)\Gamma},\\
 V_*=1,
 \qquad
 M_{0,*}=-{1\over b+1},
 \qquad
 M_{1,*}=-{aq+1\over \sqrt q\,(a+1)(b+1)},\\
 E_{00,*}={b(aq^2+1)\over (b+1)\Gamma},
 \qquad
 E_{01,*}={b\sqrt q\,(aq+1)^2\over (a+1)(b+1)\Gamma},
 \qquad
 E_{11,*}={b(aq+1)^2\over (a+1)(b+1)\Gamma},
\end{gathered}
\end{equation}
and

\begin{equation} \label{eq_W_prime_limit}
\begin{gathered}
L'_*={\sqrt q\over \Gamma},\qquad
P'_*=-{bq^2(a+1)(b+1)\over (bq+1)^2\Gamma},\qquad
Q'_*=-{aq^2(a+1)(b+1)\over (aq+1)^2\Gamma},\\
U'_*=0,
\qquad
N'_{0,*}=0,
\qquad
N'_{1,*}={\sqrt q\over \Gamma},\qquad
D'_{00,*}=-{2abq^3(a+1)(b+1)\over \Gamma^3},\\
D'_{01,*}=-{2abq^{5/2}(aq+1)(bq+1)\over \Gamma^3},
\qquad
D'_{11,*}=-{2abq^2(aq+1)(bq+1)\over \Gamma^3},\\
V'_*=0,
\qquad
M'_{0,*}=0,
\qquad
M'_{1,*}={\sqrt q\over \Gamma},\qquad
E'_{00,*}=-{2abq^3(a+1)(b+1)\over \Gamma^3},\\
E'_{01,*}=-{2abq^{5/2}(aq+1)(bq+1)\over \Gamma^3},
\qquad
E'_{11,*}=-{2abq^2(aq+1)(bq+1)\over \Gamma^3}.
\end{gathered}
\end{equation}
\end{proposition}
\begin{proof}
 For the limit of $\mathcal W(z_\rho)$, we use Lemma \ref{Lemma_asymptotic_approx} and repeat the arguments of Sections \ref{Section_proof_most_general} and \ref{Section_Wachter_specialization}. The asymptotics of the sums $\sum \frac{1}{z-c_i^2}$ and $\sum \frac{1}{(z-c_i^2)^2}$ is directly captured by \eqref{eq_Wachter_remainder} and its $z$--derivative. When $c_i^2$ appears in the numerator we can use
 $$
  \frac{c_i^2}{z-c_i^2}=\frac{z}{z-c_i^2}-1
 $$
 to reduce to the previous case. Hence, for any $z$ bounded away from the support of the Wachter distribution, the individual coordinates of $\mathcal W(z)$ have the limits given by
 {\small
\begin{align*}
\tilde L(z)&=\sqrt q\left(\frac{b}{b+1}-z\frac{1}{a+1}-z(1-z)\mathcal G_{\tau_K,\tau_M}(z)\right), \qquad
\tilde P(z)=\frac{1-a}{a+1}+(1-z)\mathcal G_{\tau_K,\tau_M}(z)+{\frac{1}{b+1}-\frac{1}{a+1}\over z},\\
\tilde Q(z)&=\frac{1}{a+1}-\frac{b}{b+1}+(1-z)\mathcal G_{\tau_K,\tau_M}(z),\qquad
\tilde U(z)=1,\qquad
\tilde N_0(z)=-\frac{1}{a+1},\\
\tilde N_1(z)&=-\sqrt q z\left((1-z)\mathcal G_{\tau_K,\tau_M}(z)+\frac{1}{a+1}\right),\qquad
\tilde D_{00}(z)=z(z-1)\mathcal G'_{\tau_K,\tau_M}(z)+(2z-1)\mathcal G_{\tau_K,\tau_M}(z),
\end{align*}
\begin{align*}
\tilde D_{01}(z)&=\sqrt q \tilde D_{11}(z),\qquad
\tilde D_{11}(z)=z^2\mathcal G_{\tau_K,\tau_M}(z)-z^2(1-z)\mathcal G'_{\tau_K,\tau_M}(z),\\
\tilde V(z)&=1,\qquad
\tilde M_0(z)=-\frac{1}{b+1},\qquad
\tilde M_1(z)=-\sqrt q\left(z(1-z)\mathcal G_{\tau_K,\tau_M}(z)+\frac{z}{a+1}+\frac{1}{b+1}-\frac{1}{a+1}\right),\\
\tilde E_{00}(z)&=\frac{1}{b+1}-\frac{1}{a+1}+z(z-1)\mathcal G'_{\tau_K,\tau_M}(z)+(2z-1)\mathcal G_{\tau_K,\tau_M}(z),\qquad
\tilde E_{01}(z)=\sqrt q \tilde E_{11}(z),\\
\tilde E_{11}(z)&=D_{11}(z)+\frac{1}{b+1}-\frac{1}{a+1} .
\end{align*}
}
In order to get $m_*$ and $m'_*$, we need to plug $z=z_\rho$ into these formulas and their $z$--derivatives. The answer gets expressed through $\mathcal G_{\tau_K,\tau_M}(z_\rho)$, $\mathcal G'_{\tau_K,\tau_M}(z_\rho)$, and $\mathcal G''_{\tau_K,\tau_M}(z_\rho)$. The first two were computed in \eqref{eq_G_z_rho} and \eqref{eq_G_prime_z_rho}, recasted in \eqref{eq_ab_def} notation as:
 \begin{equation}
 \label{eq_G_G_prime}
   \mathcal G_{\tau_K,\tau_M}(z_\rho)=\frac{a(b+1)q}{(aq+1)\Delta}, \qquad  \mathcal G'_{\tau_K,\tau_M}(z_\rho)=-\frac{(b+1)^2(a+1)q^2 a(a^2bq^2 + 1)}{\Gamma(aq + 1)^2 \Delta^2}.
 \end{equation}
The second derivative is computed as
\begin{multline}
\label{eq_G_double_prime_z_rho}
\mathcal G''_{\tau_K,\tau_M}(z_\rho)
\\
=
\frac{
2aq^3(a+1)^2(b+1)^3
\Big[
a^5b^3q^6
+3a^3b^2q^4
+a^3b^2q^3
-a^3bq^3
+a^2b^2q^3
-a^2bq^3
-3a^2bq^2
-1
\Big]}{
(aq+1)^3\Delta^3\Gamma^3
}.
\end{multline}
The last formula can be obtained from $\mathcal G'$ by reconstructing the dependence of $\Gamma$ on $q$ and then differentiating using\footnote{Note that we can not directly apply \eqref{eq_dz_dq} to get $m'_*$ from $m_*$, because the latter combines dependencies on both $z$ and $q$.}
\begin{equation}\label{eq_dz_dq}
        {\dd\over \dd z}= {q^2(a+1)(b+1)\over \Gamma}{\dd\over \dd q},
\end{equation}
which follows by differentiating $q\mapsto z_\rho$ dependence in \eqref{eq_ab_def}. Direct substitution of $\mathcal G$ and its derivatives leads to \eqref{eq_W_limit} and \eqref{eq_W_prime_limit}.

\smallskip

For the Gaussian limit, we note that since $\u^*$ and $\v^*$ have i.i.d.\ $\mathcal N(0,1)$ components and because being uncorrelated for Gaussian vectors implies being independent, all scalar products $\langle \u^*,\u_i\rangle$, $\langle \v^*,\u_i\rangle$, $\langle \u^*,\v_i\rangle$, $\langle \v^*,\v_i\rangle$ are $\mathcal N(0,1)$, which are correlated among the four, but independent over the $i$ indices. Therefore, conditionally on $(\u_i, \v_i, c_i)_{i=1}^{K-1}$, $(\v_j)_{j=K}^{M-1}$, the components of $\mathcal W(z)$ are defined through sums of independent Gaussian random variables, which makes CLT of the form \eqref{eq_W_asymptotics} immediate: we just need to ignore the higher order terms with $o(S^{-1/2})$ total contributions. The asymptotic covariance between components of $\mathcal W(z)$ gets expressed in terms of sums involving $c_i^2$, and using the convergence of their empirical distributions towards the Wachter law as in \eqref{eq_Wachter_remainder}, the conditional CLT turns into an unconditional statement as $S\to\infty$. Finally, the asymptotic mean of $\sqrt{S}\bigl(\mathcal W(z_\rho)-m_*\bigr)$ vanishes as $S\to\infty$ because of the smallness of error term in \eqref{eq_Wachter_remainder}.
\end{proof}

\subsection{Expression of $(Z_x,Z_y,Z_z)$ through $\mathcal W(z)$} \label{Section_expressions_through_W} Our next goal is to write down formulas expressing $Z_x$, $Z_y$ and $Z_z$ through the vector $\sqrt{S}\bigl(\mathcal W(z)-m_*\bigr)$. Using the notation \eqref{eq_augW_def}, set
\[
        \delta R:=\sqrt{S}\bigl(R(z_\rho)-R_*\bigr),
        \qquad
        R\in\{L,P,Q,U,N_0,N_1,D_{00},D_{01},D_{11},V,M_0,M_1,E_{00},E_{01},E_{11}\}.
\]

\begin{lemma} \label{Lemma_Zz_as_linear_form}
 With $o_p(1)$ meaning a term which tends to $0$ in probability, we have:
 \begin{equation} \label{eq_Zz_as_linear_form}
   Z_z= \frac{(a+1)(b+1)\Gamma}{\Delta^2}  \Bigl[ \delta L\cdot 2 L_* - \delta P \cdot z_\rho  Q_* - \delta Q\cdot z_\rho P_* \Bigr]+o_p(1).
 \end{equation}
\end{lemma}
\begin{proof}
By Theorem \ref{Theorem_master_equation}, the outlier $\lambda_1$ is a random (largest) solution of the scalar equation
\begin{equation*}
 L^2(\lambda_1)-\lambda_1 P(\lambda_1) Q(\lambda_1)=0.
\end{equation*}
We linearize the equation for $\lambda_1$ near $z_\rho$ by rewriting it as
\begin{multline*}
 0=\bigl[L_*+ (L(z_\rho)-L_*)+(L(\lambda_1)-L(z_\rho))\bigr]^2\\- \bigl[z_\rho+(\lambda_1-z_\rho)\bigr] \bigl[P_*+ (P(z_\rho)-P_*)+(P(\lambda_1)-P(z_\rho))\bigr] \bigl[Q_*+ (Q(z_\rho)-Q_*)+(Q(\lambda_1)-Q(z_\rho))\bigr].
\end{multline*}
Using Proposition \ref{Proposition_W_limit}, we further rewrite this as
\begin{multline}
\label{eq_x41}
 S^{-1/2}o_p(1)+O\bigl((\lambda_1-z_\rho)^2\bigr)= \bigl[L_*^2-z_\rho P_* Q_*\bigr]+ (\lambda_1-z_\rho) \left[2L_* L'_*   - z_\rho P'_* Q_* - z_\rho P_* Q'_*- P_* Q_*\right]
 \\+ (L(z_\rho)-L_*)\cdot 2 L_* - (P(z_\rho)-P_*) \cdot z_\rho  Q_* - (Q(z_\rho)-Q_*)\cdot z_\rho P_*.
\end{multline}
Note that the first term vanishes, $\bigl[L_*^2-z_\rho P_* Q_*\bigr]=0$, due to definition of $z_\rho$. For the second term,
a calculation using \eqref{eq_zqsxsy_CLT}, \eqref{eq_W_limit} and  \eqref{eq_W_prime_limit} gives
\begin{equation}\label{eq_AppE_Fdot}
        \left[2L_* L'_*   - z_\rho P'_* Q_* - z_\rho P_* Q'_*- P_* Q_*\right]=-\frac{\Delta^2}{(a+1)(b+1)\Gamma}.
\end{equation}
Since $\lambda_1-z_\rho= S^{-1/2} Z_z$, we arrive at \eqref{eq_Zz_as_linear_form}.
\end{proof}

\begin{lemma}\label{Lemma_Zx_as_linear_form}
Define the first-order variations
\begin{equation}\label{eq_delta_Ax_Bx_Dx}
\begin{aligned}
        \delta B_x
        &:=\delta N_0+\frac{P_*}{L_*}\delta N_1+N_{1,*}\frac{P_*}{L_*}\left(\frac{\delta P}{P_*}-\frac{\delta L}{L_*}\right),\\
        \delta D_x
        &:=\delta D_{00}+2\frac{P_*}{L_*}\delta D_{01}
          +\frac{P^2_*}{L^2_*}\delta D_{11}
          +2\left(D_{01,*}+\frac{P_*}{L_*}D_{11,*}\right)\frac{P_*}{L_*}\left(\frac{\delta P}{P_*}-\frac{\delta L}{L_*}\right).
\end{aligned}
\end{equation}
Then, with \(o_p(1)\) denoting a term converging to zero in probability, and $d_x$ given in \eqref{eq_dxdy_def}:
\begin{equation}\label{eq_Zx_as_linear_form}
        Z_x
        =
       \Bigl[
        (1-\s_x)^2\delta D_x
        -(1-\s_x)\s_x(\delta U+2\delta B_x)
        \Bigr]
        +d_xZ_z+o_p(1).
\end{equation}
\end{lemma}

\begin{proof}
By \eqref{eq_cosx_W}, and extending the $*$ notation as combinations of terms in \eqref{eq_x28}:
\begin{equation}
\label{eq_x29}
        \sin^2\theta_x
        =1-{(U+B_x)^2\over U(U+2B_x+D_x)}, \qquad
        \s_x=1-{(U_*+B_{x,*})^2\over U_*(U_*+2B_{x,*}+D_{x,*})}.
\end{equation}
Lemma \ref{Lemma_Zz_as_linear_form} gives
\(\lambda_1-z_\rho=O_p(S^{-1/2})\). Using
Proposition \ref{Proposition_W_limit}, we have
\begin{equation}\label{eq_W_lambda_expansion}
        \mathcal W(\lambda_1)
        =
        m_*+\bigl(\mathcal W(z_\rho)-m_*\bigr)
        +\bigl(\lambda_1-z_\rho\bigr)m'_*
        +S^{-1/2} o_p(1).
\end{equation}
Since the second and third terms are of order $S^{-1/2}$, we can use this to linearize $\sin^2\theta_x$, treated as a function of $\mathcal W(\lambda_1)$, near $m_*$. Comparing \eqref{eq_x28} with \eqref{eq_delta_Ax_Bx_Dx} and differentiating:
\begin{equation}
\label{eq_x30}
  B_x(\lambda_1)= B_{x,*} + S^{-1/2}\delta B_x+ \bigl(\lambda_1-z_\rho\bigr)\left( N'_{0,*}+\frac{P_*}{L_*}N'_{1,*}+\frac{P'_*}{L_*}N_{1,*}-\frac{P_*L'_*}{L^2_*}N_{1,*}\right)+S^{-1/2} o_p(1),
\end{equation}
\begin{multline} \label{eq_x31}
 D_x(\lambda_1)=D_{x,*}+S^{-1/2} \delta D_x + S^{-1/2} o_p(1)\\ + \bigl(\lambda_1-z_\rho\bigr)\left(
 D'_{00,*}+2\frac{P_*}{L_*}D'_{01,*}
          +\frac{P^2_*}{L^2_*}D'_{11,*}
          +2\left(D_{01,*}+\frac{P_*}{L_*}D_{11,*}\right)\frac{P_*}{L_*}\left(\frac{P'_*}{P_*}-\frac{L'_*}{L_*}\right)\right).
\end{multline}
On the other hand, differentiating \eqref{eq_x29}, we get
\begin{multline}
\label{eq_x32}
\sin^2\theta_x=\s_x-\frac{S^{-1/2}(U_*+B_{x,*})^2}{U_*(U_*+2B_{x,*}+D_{x,*})} \left(
        2\frac{\delta U+\delta B_x}{U_*+B_{x,*}}
        -\frac{\delta U}{U_*}
        -\frac{\delta U+2\delta B_x+\delta D_x}{U_*+2B_{x,*}+D_{x,*}} \right)
        \\ - (\lambda_1-z_\rho) \frac{(U_*+B_{x,*})^2}{U_*(U_*+2B_{x,*}+D_{x,*})} \left(
        2\frac{U'_*+B'_{x,*}}{U_*+B_{x,*}}
        -\frac{U'_*}{U_*}
        -\frac{U'_*+2B'_{x,*}+D'_{x,*}}{U_*+2B_{x,*}+D_{x,*}} \right) + S^{-1/2} o_p(1),
\end{multline}
where $B'_{x,*}$ and $D'_{x,*}$ are coefficients of $(\lambda_1-z_\rho)$ in \eqref{eq_x30} and \eqref{eq_x31}, respectively. Using \eqref{eq_W_limit}, we have:
$$
U_*=1, \qquad B_{x,*}=N_{0,*}+\frac{P_*}{L_*}N_{1,*}=0, \quad\text{and}\quad
(U_*+2B_{x,*}+D_{x,*})^{-1}=1-\s_x.$$
Hence, \eqref{eq_x32} simplifies to
\begin{multline}
\label{eq_x33}
\sin^2\theta_x=\s_x-S^{-1/2}(1-\s_x) \bigl(
        2(\delta U+\delta B_x)
        -\delta U
        -(1-\s_x)(\delta U+2\delta B_x+\delta D_x)\bigr)
        \\ - (\lambda_1-z_\rho) (1-\s_x) \bigl(
        2(U'_*+B'_{x,*})
        -U'_*
        -(1-\s_x)(U'_*+2B'_{x,*}+D'_{x,*}) \bigr) + S^{-1/2} o_p(1).
\end{multline}
 Plugging \eqref{eq_W_prime_limit} and comparing with the definition of $d_x$ in \eqref{eq_dxdy_def}, we arrive at \eqref{eq_Zx_as_linear_form}.
\end{proof}

By the same argument as in Lemma \ref{Lemma_Zx_as_linear_form}, we also get the following statement.

\begin{lemma}\label{Lemma_Zy_as_linear_form}
Define the first-order variations
\begin{equation}\label{eq_delta_Ay_By_Dy}
\begin{aligned}
        \delta B_y
        &:=
        \delta M_0+\frac{Q_*}{L_*}\delta M_1
        +M_{1,*}\frac{Q_*}{L_*}
        \left(\frac{\delta Q}{Q_*}-\frac{\delta L}{L_*}\right),\\
        \delta D_y
        &:=
        \delta E_{00}+2\frac{Q_*}{L_*}\delta E_{01}
        +\frac{Q_*^2}{L_*^2}\delta E_{11}
        +2\left(E_{01,*}+\frac{Q_*}{L_*}E_{11,*}\right)
        \frac{Q_*}{L_*}
        \left(\frac{\delta Q}{Q_*}-\frac{\delta L}{L_*}\right).
\end{aligned}
\end{equation}
Then, with \(o_p(1)\) denoting a term converging to zero in probability, and
\(d_y\) given in \eqref{eq_dxdy_def},
\begin{equation}\label{eq_Zy_as_linear_form}
        Z_y
        =
        (1-\s_y)^2\delta D_y
        -(1-\s_y)\s_y(\delta V+2\delta B_y)
        +d_yZ_z+o_p(1).
\end{equation}
\end{lemma}

\subsection{Asymptotic covariance of $\mathcal W(z_\rho)$}\label{Section_W_covariance} The goal of this section is to explain how the covariance of the Gaussian vector ${\mathcal N}_{\mathcal W}$ in Proposition \ref{Proposition_W_limit} is computed. We start by conditioning on all $\u_i$, $\v_i$, and $c_i$. We specify a linear basis of the ambient $S$--dimensional space, which consists of three parts:
\begin{enumerate}
\item $2K-2$ vectors split into $K-1$ pairs $\u_i,\v_i$, $1\le i \le K-1$.
\item $M-K$ vectors $\v_j$, $K\le j\le M-1$.
\item $S-M-K+2$ vectors $\w_k$, $1\le k\le S-M-K+2$, which are orthonormal and orthogonal to all the vectors in the first two parts.
\end{enumerate}

Accordingly, we introduce $2S$ (jointly) Gaussian mean $0$ random variables through scalar products of (Gaussian) $\u_*$ and $\v_*$ with the basis:
\begin{enumerate}
 \item For each $i\le K-1$, in the four-dimensional plane spanned by $\langle \u^*,\u_i\rangle$, $\langle \v^*,\u_i\rangle$, $\langle \u^*,\v_i\rangle$, $\langle \v^*,\v_i\rangle$, we choose four mean $0$ Gaussian random variables, $\xi_i,\eta_i,\bar\xi_i,\bar\eta_i$ with covariance matrix
     \[
        \E (\xi_i,\eta_i,\bar\xi_i,\bar\eta_i)^T(\xi_i,\eta_i,\bar\xi_i,\bar\eta_i)=
        \begin{pmatrix}
        1&\sqrt q&0&0\\
        \sqrt q&1&0&0\\
        0&0&1&\sqrt q\\
        0&0&\sqrt q&1
        \end{pmatrix},
   \]
   and such that
   \begin{align*}
     &\xi_i=\langle \u^*,\u_i\rangle,
        \qquad
        \eta_i=\langle \v^*,\u_i\rangle, \\
      &c_i\xi_i+\sqrt{1-c_i^2}\bar\xi_i=\langle \u^*,\v_i\rangle,
        \qquad
      c_i\eta_i+\sqrt{1-c_i^2}\bar\eta_i=\langle \v^*,\v_i\rangle.
   \end{align*}
   Such choice is possible, because the above linear combinations of $(\xi_i,\eta_i,\bar\xi_i,\bar\eta_i)$ have the same covariances as  $\langle \u^*,\u_i\rangle$, $\langle \v^*,\u_i\rangle$, $\langle \u^*,\v_i\rangle$, $\langle \v^*,\v_i\rangle$, whose covariances we computed in Lemma \ref{Lemma_asymptotic_approx} (recall that we set $C_{uu}=C_{vv}=1$ throughout this section).
 \item For each $K\le j\le M-1$,  we choose two mean $0$ Gaussian random variables $\bar\xi_j,\bar\eta_j$ with covariance matrix
       $$
        \E (\bar\xi_j,\bar\eta_j)^T(\bar\xi_j,\bar\eta_j)=
        \begin{pmatrix}
        1&\sqrt q\\
        \sqrt q&1
        \end{pmatrix},
        $$
   and such that
   \begin{equation}
      \bar\xi_j=\langle \u^*,\v_j\rangle,
        \qquad
       \bar\eta_j=\langle \v^*,\v_j\rangle.
   \end{equation}
  \item For each $1\le k \le S-M-K+2$, we choose two mean $0$ Gaussian random variables $\xi^\perp_k, \eta^\perp_k$, with covariance matrix
         $$
        \E (\xi^\perp_k,\eta^\perp_k)^T(\xi^\perp_k,\eta^\perp_k)=
        \begin{pmatrix}
        1&\sqrt q\\
        \sqrt q&1
        \end{pmatrix},
        $$
   and such that
   \begin{equation}
      \xi^\perp_k=\langle \u^*,\w_k\rangle,
        \qquad
       \eta^\perp_k=\langle \v^*,\w_k\rangle.
   \end{equation}
\end{enumerate}
 Each of these Gaussian random variables has variance $1$ and they are mostly independent, with the only non-vanishing covariances being: $\E \xi_i \eta_i=\E \bar \xi_i \bar \eta_i=\E \bar \xi_j \bar \eta_j=\E \xi^\perp_k\eta^\perp_k=\sqrt{q}$, where $i,j,k$ are as above. The coordinates of $\mathcal W(z)$ are all quadratic forms in these Gaussian random variables, once we notice that, expanding $\u^*$ and $\v^*$ in the basis,
 \begin{align*}
  \langle \u^*, \u^*\rangle&=\sum_{i=1}^{K-1} (\xi_i)^2 +\sum_{i=1}^{K-1}(\bar \xi_i)^2+\sum_{j=K}^{M-1}(\bar\xi_j)^2+\sum_{k=1}^{S-M-K+2} ( \xi^\perp_k)^2,
  \\
   \langle \v^*, \v^*\rangle&= \sum_{i=1}^{K-1} (\eta_i)^2 +\sum_{i=1}^{K-1}(\bar \eta_i)^2+\sum_{j=K}^{M-1}(\bar\eta_j)^2+\sum_{k=1}^{S-M-K+2} ( \eta^\perp_k)^2.
 \end{align*}

The next step in the covariance computation is to decompose \(\mathcal W(z_\rho)\) into independent contributions coming from the three types of
basis directions introduced above, express each contribution as a quadratic form in the Gaussian
variables \(\xi_i,\eta_i,\bar\xi_i,\bar\eta_i\), \(\bar\xi_j,\bar\eta_j\), or
\(\xi_k^\perp,\eta_k^\perp\), and then apply Wick's formula.

Let
\((\xi,\eta,\bar\xi,\bar\eta)\) be a centered Gaussian vector with
\[
        \E\xi^2=\E\eta^2=\E\bar\xi^2=\E\bar\eta^2=1,
        \qquad
        \E\xi\eta=\E\bar\xi\bar\eta=\sqrt q,
\]
and with all other covariances equal to zero.  Wick's formula for centered Gaussian random variables gives
\[
        \operatorname{Cov}(X_1X_2,X_3X_4)
        =\E X_1X_3\,\E X_2X_4+
        \E X_1X_4\,\E X_2X_3
\]
Applying this identity to the quadratic monomials in
\(\xi,\eta,\bar\xi,\bar\eta\), we obtain
\begin{equation}\label{eq_Wick_table_four_variables}
\renewcommand{\arraystretch}{1.15}
\begin{array}{c|cccccccccc}
\operatorname{Cov}(\cdot,\cdot)
&\xi^2&\eta^2&\bar\xi^2&\bar\eta^2&\xi\eta&\bar\xi\bar\eta&\xi\bar\xi&\xi\bar\eta&\eta\bar\xi&\eta\bar\eta\\
\hline
\xi^2&2&2q&0&0&2\sqrt q&0&0&0&0&0\\
\eta^2&2q&2&0&0&2\sqrt q&0&0&0&0&0\\
\bar\xi^2&0&0&2&2q&0&2\sqrt q&0&0&0&0\\
\bar\eta^2&0&0&2q&2&0&2\sqrt q&0&0&0&0\\
\xi\eta&2\sqrt q&2\sqrt q&0&0&1+q&0&0&0&0&0\\
\bar\xi\bar\eta&0&0&2\sqrt q&2\sqrt q&0&1+q&0&0&0&0\\
\xi\bar\xi&0&0&0&0&0&0&1&\sqrt q&\sqrt q&q\\
\xi\bar\eta&0&0&0&0&0&0&\sqrt q&1&q&\sqrt q\\
\eta\bar\xi&0&0&0&0&0&0&\sqrt q&q&1&\sqrt q\\
\eta\bar\eta&0&0&0&0&0&0&q&\sqrt q&\sqrt q&1
\end{array}.
\end{equation}
In particular,  if
\((\xi,\eta)\) is centered Gaussian with \(\E\xi^2=\E\eta^2=1\) and
\(\E\xi\eta=\sqrt q\), then
\begin{equation}\label{eq_Wick_table_two_variables}
\begin{array}{c|ccc}
\operatorname{Cov}(\cdot,\cdot)&\xi^2&\eta^2&\xi\eta\\
\hline
\xi^2&2&2q&2\sqrt q\\
\eta^2&2q&2&2\sqrt q\\
\xi\eta&2\sqrt q&2\sqrt q&1+q
\end{array}.
\end{equation}
We now introduce the word ``block'', which will be used in the rest of the proof.  A block is simply
the collection of all terms in \(\mathcal W(z_\rho)\) produced by one independent group of Gaussian
scalar products.  There are three kinds of blocks.  First, for each canonical pair
\((\u_i,\v_i)\), \(1\le i\le K-1\), there is a four-variable group
\((\xi_i,\eta_i,\bar\xi_i,\bar\eta_i)\).  We call this a paired canonical-direction block.  Second,
for each extra direction \(\v_j\), \(K\le j\le M-1\), there is a two-variable group
\((\bar\xi_j,\bar\eta_j)\).  Third, for each direction \(\w_k\), there is a two-variable group \((\xi_k^\perp,\eta_k^\perp)\).  The three families have
multiplicities
\[
        K-1,\qquad M-K,
        \qquad S-K-M+2,
\]
respectively.  After division by \(S\), these multiplicities converge to
\[
        {1\over a+1},\qquad
        {1\over b+1}-{1\over a+1},
        \qquad
        1-{1\over a+1}-{1\over b+1}.
\]

For a paired canonical-direction block, corresponding to $(\u_i,\v_i)$, set
\[
        s_i=(1-c_i^2)^{1/2},
        \qquad
        d_i=z_\rho-c_i^2, \qquad
        \widetilde\xi_i=c_i\xi_i+s_i\bar\xi_i=\langle \u^*,\v_i\rangle,
        \qquad
        \widetilde\eta_i=c_i\eta_i+s_i\bar\eta_i=\langle \v^*,\v_i\rangle.
\]
Let \(\mathcal P_i\) be the raw contribution of this paired canonical-direction block to
\(\mathcal W(z_\rho)\), in the coordinate order \eqref{eq_augW_def}.  Its coordinates are
\begin{equation}\label{eq_paired_block_contribution}
\begin{aligned}
(\mathcal P_i)_L={}&
\xi_i\eta_i+\bar\xi_i\bar\eta_i
+{\widetilde\xi_i(c_i\eta_i-z_\rho\widetilde\eta_i)\over d_i}
-{z_\rho\xi_i(\eta_i-c_i\widetilde\eta_i)\over d_i},\\
(\mathcal P_i)_P={}&
-\xi_i^2-\bar\xi_i^2
+{\widetilde\xi_i^2-2c_i\widetilde\xi_i\xi_i+z_\rho\xi_i^2\over d_i},\qquad
(\mathcal P_i)_Q={}
-\eta_i^2-\bar\eta_i^2
+{\eta_i^2-2c_i\eta_i\widetilde\eta_i+z_\rho\widetilde\eta_i^2\over d_i},\\
(\mathcal P_i)_U={}&\xi_i^2+\bar\xi_i^2,\qquad
(\mathcal P_i)_{N_0}={}{\xi_i(c_i\widetilde\xi_i-z_\rho\xi_i)\over d_i},       \qquad
(\mathcal P_i)_{N_1}=-{z_\rho\xi_i(\eta_i-c_i\widetilde\eta_i)\over d_i},\\
(\mathcal P_i)_{D_{00}}={}&{(c_i\widetilde\xi_i-z_\rho\xi_i)^2\over d_i^2},
        \qquad
(\mathcal P_i)_{D_{01}}=-{z_\rho(c_i\widetilde\xi_i-z_\rho\xi_i)(\eta_i-c_i\widetilde\eta_i)\over d_i^2},\qquad
(\mathcal P_i)_{D_{11}}={}{z_\rho^2(\eta_i-c_i\widetilde\eta_i)^2\over d_i^2},\\
(\mathcal P_i)_V={}&\eta_i^2+\bar\eta_i^2,\qquad
(\mathcal P_i)_{M_0}={}{\widetilde\eta_i(c_i\eta_i-z_\rho\widetilde\eta_i)\over d_i},
        \qquad
(\mathcal P_i)_{M_1}=-{z_\rho\widetilde\eta_i(\widetilde\xi_i-c_i\xi_i)\over d_i},\\
(\mathcal P_i)_{E_{00}}={}&{(c_i\eta_i-z_\rho\widetilde\eta_i)^2\over d_i^2},
        \qquad
(\mathcal P_i)_{E_{01}}=-{z_\rho(c_i\eta_i-z_\rho\widetilde\eta_i)(\widetilde\xi_i-c_i\xi_i)\over d_i^2},\qquad
(\mathcal P_i)_{E_{11}}={}{z_\rho^2(\widetilde\xi_i-c_i\xi_i)^2\over d_i^2}.
\end{aligned}
\end{equation}
For a block corresponding to  direction $\v_j$, \(K\le j\le M-1\),  its contribution to $\mathcal W(z_\rho)$ is  $\mathcal E_j$ given by:
\begin{equation}\label{eq_extra_block_contribution}
\begin{gathered}
(\mathcal E_j)_P=(z_\rho^{-1}-1)\bar\xi_j^2,
\qquad
(\mathcal E_j)_U=\bar\xi_j^2,
\qquad
(\mathcal E_j)_V=\bar\eta_j^2,\\
(\mathcal E_j)_{M_0}=-\bar\eta_j^2,
\qquad
(\mathcal E_j)_{M_1}=-\bar\xi_j\bar\eta_j,
\qquad
(\mathcal E_j)_{E_{00}}=\bar\eta_j^2,\\
(\mathcal E_j)_{E_{01}}=\bar\xi_j\bar\eta_j,
\qquad
(\mathcal E_j)_{E_{11}}=\bar\xi_j^2,
\end{gathered}
\end{equation}
and all other coordinates of $\mathcal E_j$ are equal to zero.  Finally, for a block corresponding to
$\w_k$, $1\le k \le S-M-K+2$,  its contribution to $\mathcal W(z_\rho)$ is \(\mathcal O_k\) given by:
\begin{equation}\label{eq_perp_block_contribution}
\begin{gathered}
        (\mathcal O_k)_L=\xi_k^\perp\eta_k^\perp,
        \qquad
        (\mathcal O_k)_P=-(\xi_k^\perp)^2,
        \qquad
        (\mathcal O_k)_Q=-(\eta_k^\perp)^2,
        \\
        (\mathcal O_k)_U=(\xi_k^\perp)^2,
        \qquad
        (\mathcal O_k)_V=(\eta_k^\perp)^2,
\end{gathered}
\end{equation}
and all other coordinates equal to zero.  With this notation,
\begin{equation}\label{eq_W_block_decomposition_new}
\sqrt S\bigl(\mathcal W(z_\rho)-\E \mathcal W(z_\rho)\bigr)
=
{1\over\sqrt S}\sum_{i=1}^{K-1}\bigl(\mathcal P_i-\E\mathcal P_i\bigr)
+{1\over\sqrt S}\sum_{j=K}^{M-1}\bigl(\mathcal E_j-\E\mathcal E_j\bigr)
+{1\over\sqrt S}\sum_{k=1}^{S-K-M+2}\bigl(\mathcal O_k-\E\mathcal O_k\bigr)
\end{equation}
The deterministic mean in the left-hand side of \eqref{eq_W_block_decomposition_new} matches the vector $m_*$ in Proposition \ref{Proposition_W_limit} up to $o\bigl(S^{-1/2}\bigr)$ error.
The $S-K+1$ terms in the right-hand side of \eqref{eq_W_block_decomposition_new} are independent conditional on the canonical bases and on the
values of $c_i$.  Hence, the covariance between coordinates of $\mathcal W(z_\rho)$ splits into a sum of $S-K+1$ terms subdivided into three groups. Using the tables \eqref{eq_Wick_table_four_variables} and \eqref{eq_Wick_table_two_variables}, each of the $S-K+1$ terms becomes an explicit expression --- a polynomial in $\sqrt{q}$, $c_i$, $\frac{1}{d_i}$, and $z_\rho$. In each monomial of this polynomial, $\frac{1}{d_i}$ has degree at most $4$ and the degree of $c_i$ is at most the degree of $\frac{1}{d_i}$. Therefore, summing over $i$, $j$, $k$ and applying the convergence toward the Wachter distribution in the same way as in Section \ref{Section_asymptotics_W}, each element of the asymptotic covariance matrix of $\mathcal W(z_\rho)$ is expressed as a rational function in terms of
$$
 \sqrt{q}, \quad a,\quad b,\quad  \mathcal G_{\tau_K,\tau_M}(z_\rho), \quad  \mathcal G'_{\tau_K,\tau_M}(z_\rho),\quad \mathcal G''_{\tau_K,\tau_M}(z_\rho),\quad \mathcal G'''_{\tau_K,\tau_M}(z_\rho).
$$
The last four expressions are explicit rational functions in $q,a,b$ (as explained in Section \ref{Section_asymptotics_W}), and we arrive at the final explicit expression (not shown)  for the matrix elements of $15\times 15$ asymptotic covariance matrix of $\mathcal W(z_\rho)$ as rational functions in $\sqrt{q}$, $a$, and $b$. In the same way, we compute asymptotic (co-)variance for any linear combinations of the coordinates of $\mathcal W(z_\rho)$, such as those appearing in Section \ref{Section_expressions_through_W}.

\subsection{A detailed example: the outlier linear form}\label{subsec_ellz_worked_example}

We illustrate the covariance calculation in full detail for the outlier
linear form corresponding to $Z_z$. Let $w_L$, $w_P$, $w_Q$ be the coordinates of the limiting Gaussian vector $\mathcal N_{\mathcal W}$ in Proposition \ref{Proposition_W_limit} corresponding to $\delta L$, $\delta P$, $\delta Q$, respectively, and define
\begin{equation}\label{eq_ellz_def}
        \ell_z(\mathcal N_{\mathcal W}):=
        { (a+1)(b+1)\Gamma\over \Delta^2}
        \left(2L_*w_L-z_\rho Q_*w_P-z_\rho P_*w_Q\right).
\end{equation}
This is the linear form which appears in Lemma
\ref{Lemma_Zz_as_linear_form}.  To keep the notation readable, set
\begin{equation}\label{eq_ellz_constants}
        \Lambda_z:={ (a+1)(b+1)\Gamma\over \Delta^2},
        \qquad
        \mathfrak l:=\Lambda_z L_*,
        \qquad
        \mathfrak p:=\Lambda_z P_*,
        \qquad
        \mathfrak q:=\Lambda_z Q_* .
\end{equation}
Using \eqref{eq_W_limit}, these constants are
\begin{equation}\label{eq_ellz_constants_explicit}
        \mathfrak l={\Gamma\over \sqrt q\,\Delta},
        \qquad
        \mathfrak p=-{(b+1)\Gamma\over (bq+1)\Delta},
        \qquad
        \mathfrak q=-{(a+1)\Gamma\over (aq+1)\Delta}.
\end{equation}
We compute contributions to
$\mathrm{Var}[\ell_z(\mathcal N_{\mathcal W})]$ from three types of
terms in \eqref{eq_W_block_decomposition_new}.

\noindent {\bf Paired canonical-direction contribution.}
Fix \(1\le i\le K-1\).  Let
\[
        r_i={z_\rho-1\over z_\rho-c_i^2}.
\]
Substituting the \(L\)-, \(P\)-, and \(Q\)-coordinates of
\(\mathcal P_i\) from \eqref{eq_paired_block_contribution} into
\eqref{eq_ellz_def}, and expanding in the monomials appearing in
\eqref{eq_Wick_table_four_variables}, gives
\begin{equation}\label{eq_ellz_paired_expansion}
        \ell_z(\mathcal P_i)
        =
        r_i\Big[
        \mathfrak p z_\rho(1-c_i^2)\eta_i^2
        +\mathfrak q z_\rho\bar\xi_i^2
        +\mathfrak p z_\rho c_i^2\bar\eta_i^2
        +2\mathfrak l c_i^2\bar\xi_i\bar\eta_i
        -2\mathfrak l c_is_i\,\eta_i\bar\xi_i
        -2\mathfrak p z_\rho c_is_i\,\eta_i\bar\eta_i
        \Big].
\end{equation}
Note that the coefficients of $\xi_i^2$ and
$\xi_i\eta_i$ vanished because  $-c_i^2-d_i+z_\rho=0$.

Now apply \eqref{eq_Wick_table_four_variables}. Writing the coefficients in \eqref{eq_ellz_paired_expansion} as
\[
\begin{aligned}
        A_i&:=\mathfrak p z_\rho(1-c_i^2),&
        B_i&:=\mathfrak q z_\rho,&
        C_i&:=\mathfrak p z_\rho c_i^2,\\
        D_i&:=2\mathfrak l c_i^2,&
        E_i&:=-2\mathfrak l c_is_i,&
        F_i&:=-2\mathfrak p z_\rho c_is_i,
\end{aligned}
\]
we get
\begin{multline}\label{eq_ellz_paired_cov_before_t}
        \mathrm{Var}[\ell_z(\mathcal P_i)]
        \\=
        r_i^2\Bigl[
        2A_i^2
        +2B_i^2+2C_i^2+(1+q)D_i^2
        +4qB_iC_i
        +4\sqrt q\,D_i(B_i+C_i)
        +E_i^2+F_i^2+2\sqrt{q}E_iF_i
        \Bigr].
\end{multline}
Collection of powers of $c_i$ gives
\begin{equation}\label{eq_ellz_pair_integrand}
        \mathrm{Var}[\ell_z(\mathcal P_i)]
        =
        r_i^2(B_0+B_1c_i^2+B_2c_i^4),
\end{equation}
where
$$
        B_0=2z_\rho^2(\mathfrak p^2+\mathfrak q^2),\quad
        B_1=4\left(
        \mathfrak l^2
        +2qz_\rho\mathfrak l\mathfrak p
        +2\sqrt q\,z_\rho\mathfrak l\mathfrak q
        +qz_\rho^2\mathfrak p\mathfrak q
        \right),\quad
        B_2=4q\mathfrak l^2.
$$
Therefore, using the Wachter distribution \eqref{eq_Wachter} with Theorem \ref{Theorem_Wachter}, and noting that the normalization in this section is $\frac{1}{S}$, the paired canonical-direction contribution to
$\mathrm{Var}[\ell_z(\mathcal N_{\mathcal W})]$ is
\begin{equation}\label{eq_ellz_pair_contribution}
        \mathcal C_{\mathrm{pair}}^{(z)}
        =
        \int_{\lambda_-}^{\lambda_+}  \frac{(z_\rho-1)^2(B_0+B_1t+B_2t^2)}{(z_\rho-t)^2}
        \,  \frac{\omega_{\tau_K,\tau_M}(t)}{a+1} \dd t.
\end{equation}
We can rewrite this integral in terms of the total mass $\frac{1}{a+1}$, Stieltjes transform $\mathcal G_{\tau_K,\tau_M}(z_\rho)$ and its derivative $\mathcal G'_{\tau_K,\tau_M}(z_\rho)$, whose explicit expressions are given at \eqref{eq_G_G_prime}.
Using
\[
        {t\over (z_\rho-t)^2}
        ={z_\rho\over (z_\rho-t)^2}-{1\over z_\rho-t}, \qquad
        {t^2\over (z_\rho-t)^2}
        ={z_\rho^2\over (z_\rho-t)^2}
        -{2z_\rho\over z_\rho-t}
        +1.
\]
We obtain the fully reduced form
\begin{equation}\label{eq_ellz_pair_contribution_reduced}
        \mathcal C_{\mathrm{pair}}^{(z)}
        =
        (z_\rho-1)^2
        \Bigl[
        -(B_0+z_\rho B_1+z_\rho^2B_2)\mathcal G'_{\tau_K,\tau_M}(z_\rho)
        -(B_1+2z_\rho B_2)\mathcal G_{\tau_K,\tau_M}(z_\rho)
        +B_2\frac{1}{a+1}
        \Bigr].
\end{equation}

\noindent \paragraph{\bf Extra $\v_j$ contribution.}
For \(K\le j\le M-1\), only the \(P\)-coordinate of
\(\mathcal E_j\) enters \(\ell_z\).  From
\eqref{eq_extra_block_contribution}, $
        (\mathcal E_j)_P=(z_\rho^{-1}-1)\bar\xi_j^2,$ hence
$$
        \ell_z(\mathcal E_j)=
        \mathfrak q(z_\rho-1)\bar\xi_j^2 .
$$
By \eqref{eq_Wick_table_two_variables},
$\mathrm{Var}[\bar\xi_j^2]=2,$ and therefore
\begin{equation}\label{eq_ellz_extra_one_direction}
        \mathrm{Var}[\ell_z(\mathcal E_j)]
        =
        2\mathfrak q^2(z_\rho-1)^2.
\end{equation}
There are $M-K$ such terms in the decomposition \eqref{eq_W_block_decomposition_new}.  Since
${M-K\over S}\to {1\over b+1}-{1\over a+1},$
their contribution to the limiting covariance is
\begin{equation}\label{eq_ellz_extra_contribution}
        \mathcal C_{\mathrm{extra}}^{(z)}
        =
        \left({1\over b+1}-{1\over a+1}\right)
        2\mathfrak q^2(z_\rho-1)^2
        =
        {2(a-b)(1-q)^2\Gamma^2
        \over
        q^2(a+1)(b+1)^3(aq+1)^2}.
\end{equation}

\medskip

\noindent {\bf Orthogonal-complement $\w_k$ contribution.}
For the third sum in \eqref{eq_W_block_decomposition_new}, we have
\[
        (\mathcal O_k)_L=\xi_k^\perp\eta_k^\perp,
        \qquad
        (\mathcal O_k)_P=-(\xi_k^\perp)^2,
        \qquad
        (\mathcal O_k)_Q=-(\eta_k^\perp)^2.
\]
Thus
\begin{equation}\label{eq_ellz_perp_expansion}
        \ell_z(\mathcal O_k)
        =
        z_\rho\mathfrak q(\xi_k^\perp)^2
        +z_\rho\mathfrak p(\eta_k^\perp)^2
        +2\mathfrak l\,\xi_k^\perp\eta_k^\perp .
\end{equation}
Applying the Wick table \eqref{eq_Wick_table_two_variables},
we obtain
\begin{equation}\label{eq_ellz_perp_one_direction}
        \mathrm{Var}[\ell_z(\mathcal O_k)]
        =
        2z_\rho^2\mathfrak q^2
        +2z_\rho^2\mathfrak p^2
        +4(1+q)\mathfrak l^2
        +4qz_\rho^2\mathfrak p\mathfrak q
        +8\sqrt q\,z_\rho\mathfrak l(\mathfrak p+\mathfrak q).
\end{equation}
There are $S-K-M+2$ such terms in \eqref{eq_W_block_decomposition_new} and $
        {S-K-M+2\over S}
        \to
        1-{1\over a+1}-{1\over b+1}.$
Therefore,
\begin{multline}\label{eq_ellz_perp_contribution}
        \mathcal C_{\mathrm{orth}}^{(z)}
        =
        \left(1-{1\over a+1}-{1\over b+1}\right)
        \Big[
        2z_\rho^2\mathfrak q^2
        +2z_\rho^2\mathfrak p^2
        +4(1+q)\mathfrak l^2
        +4qz_\rho^2\mathfrak p\mathfrak q
        +8\sqrt q\,z_\rho\mathfrak l(\mathfrak p+\mathfrak q)
        \Big]
         \\
        =
        {2(1-q)^2(ab-1)\Gamma^2
        \over
        q^2(a+1)^3(b+1)^3\Delta^2}
        \Big[
        2a^2b^2q+2a^2bq+a^2
        +2ab^2q
        +2abq+2a+b^2+2b+2
        \Big].
\end{multline}

\medskip

\noindent {\bf Putting the three pieces together.}
Combining \eqref{eq_ellz_pair_contribution_reduced},
\eqref{eq_ellz_extra_contribution}, and
\eqref{eq_ellz_perp_contribution}, we get
\begin{equation}\label{eq_ellz_three_piece_sum}
        \mathrm{Var}[\ell_z(\mathcal N_{\mathcal W})]
        =
        \mathcal C_{\mathrm{pair}}^{(z)}
        +\mathcal C_{\mathrm{extra}}^{(z)}
        +\mathcal C_{\mathrm{orth}}^{(z)}.
\end{equation}
Substituting \eqref{eq_ellz_constants_explicit}, \eqref{eq_G_G_prime}, and \eqref{eq_zqsxsy_CLT}
into the right-hand side of \eqref{eq_ellz_three_piece_sum}, and simplifying,
gives
\begin{equation}\label{eq_ellz_final_variance}
        \mathrm{Var}[\ell_z(\mathcal N_{\mathcal W})]
        =
        {2(1-q)^2(2abq+a+b)\Gamma
        \over
        q^2(a+1)^2(b+1)^2}.
\end{equation}
This is exactly the variance of $Z_z$, i.e.\ $\sigma_{zz}$ in the statement of Theorem \ref{Theorem_joint_CLT}.

\subsection{Collecting the three linear forms}\label{Section_three_linear_forms}
 In this section we combine all the ingredients to finish the proof of Theorem \ref{Theorem_joint_CLT}.

\begin{proposition}
In the notations of Lemmas \ref{Lemma_Zz_as_linear_form}, \ref{Lemma_Zx_as_linear_form}, \ref{Lemma_Zy_as_linear_form}, define the distributional limit:
\begin{equation} \label{eq_x39}
\begin{pmatrix} \tilde Z_x\\ \tilde Z_y\\ \tilde Z_z\end{pmatrix}=\lim_{S\to\infty} \begin{pmatrix}
        (1-\s_x)^2\delta D_x
        -(1-\s_x)\s_x(\delta U+2\delta B_x)
       \\ (1-\s_y)^2\delta D_y
        -(1-\s_y)\s_y(\delta V+2\delta B_y)\\
         \frac{(a+1)(b+1)\Gamma}{\Delta^2}  \Bigl[ \delta L\cdot 2 L_* - \delta P \cdot z_\rho  Q_* - \delta Q\cdot z_\rho P_* \Bigr]
        \end{pmatrix}.
\end{equation}
Then $(\tilde Z_x,\tilde Z_y, \tilde Z_z)$ is a mean $0$ Gaussian vector with covariance matrix
\begin{multline} \label{eq_covariance_auxiliary}
 \E \begin{pmatrix}
  \tilde Z_x \tilde Z_x & \tilde Z_x \tilde Z_y& \tilde Z_x \tilde Z_z\\
  \tilde Z_y \tilde Z_x & \tilde Z_y \tilde Z_y& \tilde Z_y \tilde Z_z\\
  \tilde Z_z \tilde Z_x & \tilde Z_z \tilde Z_y& \tilde Z_z \tilde Z_z
  \end{pmatrix}
  = \widetilde \Sigma
  \\
  =
   \begin{pmatrix}
    \frac{2a(a+1)^2(1-q)^2\widetilde R_{xx}}{
(aq+1)^4\Delta^4\Gamma} & \frac{2abq(a+1)(b+1)(1-q)^2\widetilde R_{xy}}{
(aq+1)^2(bq+1)^2\Delta^4\Gamma} & -\frac{2a(1-q)^2\widetilde R_{xz}}{
q(b+1)(aq+1)^2\Delta^2}
\\
 \frac{2abq(a+1)(b+1)(1-q)^2\widetilde R_{xy}}{
(aq+1)^2(bq+1)^2\Delta^4\Gamma} &  \frac{2b(b+1)^2(1-q)^2\widetilde R_{yy}}{
(bq+1)^4\Delta^4\Gamma} & -\frac{2b(1-q)^2\widetilde R_{yz}}{
q(a+1)(bq+1)^2\Delta^2}
\\
 -\frac{2a(1-q)^2\widetilde R_{xz}}{
q(b+1)(aq+1)^2\Delta^2} & -\frac{2b(1-q)^2\widetilde R_{yz}}{
q(a+1)(bq+1)^2\Delta^2} &
{2(1-q)^2(2abq+a+b)\Gamma\over q^2(a+1)^2(b+1)^2}
   \end{pmatrix},
\end{multline}
where
\begin{align*}
\widetilde R_{xx}={}&
1+2bq+b^2q^2-4abq+2abq^2-6ab^2q^2+2ab^2q^3
+ab^3q^2-4ab^3q^3+ab^3q^4\\
&-2a^2bq^2+4a^2bq^3+4a^2b^2q^2-4a^2b^2q^3+5a^2b^2q^4
+8a^2b^3q^3\\
&-8a^2b^3q^4+4a^2b^3q^5+a^2b^4q^4
+a^3bq^4+4a^3b^2q^3-8a^3b^2q^4\\
&+2a^3b^2q^5+5a^3b^3q^4-16a^3b^3q^5+4a^3b^3q^6
-2a^3b^4q^5-2a^3b^4q^6\\
&+a^4b^2q^4-4a^4b^2q^5+a^4b^2q^6
+2a^4b^3q^5-6a^4b^3q^6+2a^4b^3q^7\\
&+2a^4b^4q^6+2a^4b^5q^7+a^5b^3q^6+2a^5b^4q^7+a^5b^5q^8,
\end{align*}
\begin{align*}
\widetilde R_{xz}={}&
1+bq-b^2q+b^2q^2-2abq+a^2bq^3-a^2bq^2
+2ab^2q^3-4ab^2q^2\\
&\quad
+2a^2b^2q^4-a^2b^2q^3+a^2b^3q^4,
\end{align*}
\begin{align*}
\widetilde R_{xy}={}&
2a^4b^4q^6+a^4b^3q^6+2a^4b^3q^5+a^4b^2q^4
+a^3b^4q^6+2a^3b^4q^5\\
&+2a^3b^3q^6-4a^3b^3q^5+2a^3b^3q^4
+2a^3b^2q^5-11a^3b^2q^4+4a^3b^2q^3\\
&+a^3bq^4-4a^3bq^3+a^3bq^2
+a^2b^4q^4+2a^2b^3q^5-11a^2b^3q^4+4a^2b^3q^3\\
&+2a^2b^2q^4-16a^2b^2q^3+8a^2b^2q^2
+4a^2bq^3-5a^2bq^2+2a^2bq+a^2q^2\\
&+ab^3q^4-4ab^3q^3+ab^3q^2
+4ab^2q^3-5ab^2q^2+2ab^2q+8abq^2-4abq\\
&+2aq-a+b^2q^2+2bq-b,
\end{align*}
and
$$
 \widetilde R_{yy}=\widetilde R_{xx}\big|_{a\leftrightarrow b},\qquad  \widetilde R_{yz}=\widetilde R_{xz}\big|_{a\leftrightarrow b}.
$$
\end{proposition}
\begin{proof}
 The asymptotic Gaussianity was proven in Proposition \ref{Proposition_W_limit}. The $(3,3)$ matrix element of the covariance matrix, $\E \tilde Z_z \tilde Z_z$, was computed in Section \ref{subsec_ellz_worked_example}. The remaining matrix elements are computed by a similar tedious, but straightforward computation.
\end{proof}

\begin{proposition} \label{Proposition_inverse_covariance}
 For the covariance matrix $\widetilde \Sigma$ given by \eqref{eq_covariance_auxiliary}, we have
 \begin{equation}\label{eq_tildeSigma_det_new}
        \det\widetilde\Sigma={16a^2b^2(1-q)^6\Gamma\over q\Delta^4}.
\end{equation}
Its inverse $\widetilde\Sigma^{-1}$ is given by \eqref{eq_cov_Sigma_inverse}.
\end{proposition}
\begin{proof} Direct computation. \end{proof}

\begin{proof}[Proof of Theorem \ref{Theorem_joint_CLT}]

By Proposition \ref{Proposition_W_limit} and Lemmas \ref{Lemma_Zz_as_linear_form}, \ref{Lemma_Zx_as_linear_form}, \ref{Lemma_Zy_as_linear_form}, and using \eqref{eq_x39}:
$$
   \lim_{S\to\infty}  \begin{pmatrix}
        Z_x\\ Z_y\\ Z_z
     \end{pmatrix}
     =
     \begin{pmatrix}
        \widetilde Z_x+ d_x\widetilde Z_z\\
        \widetilde Z_y+ d_y \widetilde Z_z\\
        \widetilde Z_z
     \end{pmatrix}
     =   \begin{pmatrix}
        1&0&d_x\\
        0&1&d_y\\
        0&0&1
        \end{pmatrix}
        \begin{pmatrix}
        \widetilde Z_x\\
        \widetilde Z_y\\
        \widetilde Z_z
     \end{pmatrix}.
$$
Hence, the asymptotic Gaussianity and covariance matrix $\Sigma$ of $(Z_x,Z_y,Z_z)$ are obtained from the ones for $(\widetilde Z_x,\widetilde Z_y,\widetilde Z_z)$ and we have:
$$
 \Sigma=  \begin{pmatrix}
        1&0&d_x\\
        0&1&d_y\\
        0&0&1
        \end{pmatrix}  \widetilde \Sigma  \begin{pmatrix}
        1&0&0\\
        0&1&0\\
        d_x&d_y&1
        \end{pmatrix}.
$$
Inverting the identity and using Proposition \ref{Proposition_inverse_covariance}, we obtain the factorized formula \eqref{eq_covariance_factorized}.  Multiplying out the matrices gives the direct covariance entries
\eqref{eq_sigmax}--\eqref{eq_sigmaxy}.
\end{proof}

\begin{remark}
The square roots of $q$ which appear in the formulas of Proposition \ref{Proposition_W_limit} and the Wick tables of Section \ref{Section_W_covariance} disappear from the final covariance.  This is forced by sign invariance:
changing the sign of the signal vector $\v^*$ replaces $\sqrt q$ by $-\sqrt q$, but leaves
$\lambda_1$, $\sin^2\theta_x$, and $\sin^2\theta_y$ unchanged.  Thus, all final covariance
entries have to be even functions of $\sqrt q$, and hence rational functions of $q$.
\end{remark}

\subsection{Extension to multiple signals}

\label{Section_CLT_multiple}

We now explain that Theorem \ref{Theorem_joint_CLT} extends to the case of multiple signals, and therefore finish the proof of Theorem \ref{Theorem_basic_multi}. In Section \ref{Section_LLN_multi} we proved all parts, except for the Gaussian corrections to the limits. The argument for the latter follows the same iterative scheme as in Section \ref{Section_LLN_multi} and adds signals one by one. We show that when we add the $q$th signal on the $q$th iteration, two statements hold:
\begin{enumerate}
 \item The new sample canonical correlation and corresponding angles satisfy \eqref{eq_joint_CLT_appendix} and the limit is independent of the limits of the previous iterations. (Step 1.)
 \item The (supercritical) canonical correlations and corresponding angles from the previous $(q-1)$ steps change by $O(S^{-1})$ and therefore their limits are still governed by \eqref{eq_joint_CLT_appendix}. (Steps 2 and 3.)
\end{enumerate}

\noindent {\bf Step 1.} For the first part, we repeat the proof of Theorem
\ref{Theorem_joint_CLT}, conditioning throughout on the randomness coming from the noise
and the first $(q-1)$ signals. Compared to the $\mathbbm q=1$ case, the sums contributing
to $\mathcal W(z)$ in \eqref{eq_augW_def} contain $q-1$ additional terms corresponding to
the previously added spikes/signals. Their total contribution is $O(S^{-1})$, and hence
$o(S^{-1/2})$ at the CLT scale, while all the remaining terms have the same asymptotic
behavior as in the $\mathbbm q=1$ case governed by the Wachter law, cf.\ Lemma
\ref{Lemma_G_approximation}. Therefore, the deterministic limits and linearizations of
Section \ref{Section_asymptotics_W} remain unchanged.

It remains to consider the Gaussian fluctuations and their dependence on the previous
iterations. The new jointly Gaussian vectors $\u^*,\v^*$ are
independent of all $\u_i,\v_i,c_i$ from the previous iterations. By rotational invariance, conditionally on the previous iterations
the vectors of scalar products
$
 \bigl(\langle \u^*,\u_i\rangle,\langle \u^*,\v_i\rangle,
       \langle \v^*,\u_i\rangle,\langle \v^*,\v_i\rangle\bigr)
$
form independent Gaussian blocks across $i$; within each block their covariance depends on
$c_i$ (and on $\rho^2[q]$). Hence, conditionally on the previous iterations, the law of large
numbers and the central limit theorem used in Proposition \ref{Proposition_W_limit} apply in
exactly the same way as in the one-signal case. The resulting limiting means and covariances
depend on $\{c_i\}$ only through normalized weighted sums, which converge in probability to
deterministic integrals with respect to the Wachter law. Thus, the conditional limiting
Gaussian distribution is nonrandom and coincides with the one in Proposition
\ref{Proposition_W_limit}, with $\rho^2=\rho^2[q]$. Hence, removing the conditioning by taking
expectations, we conclude that this limit is independent of the randomness generated 
by the first $(q-1)$ iterations. Consequently, the Gaussian limit in \eqref{eq_joint_CLT_appendix} for the
$q$--th signal is independent of the Gaussian limits obtained on the previous iterations.

\medskip

\noindent {\bf Step 2.} For the canonical correlations in the second part, we fix $k<q$ and let $\lambda_k$ denote the $k$-th largest squared sample canonical correlation after $(q-1)$ iterations, corresponding to a supercritical $\rho^2[k]>\rho_c^2$ in the setting of Theorem \ref{Theorem_basic_multi}. By the argument of Section \ref{Section_LLN_multi}, we already know that $\lambda_k$ is close to $z_{\rho[k]}$ and that after the $q$th iteration, we again are going to have a sample squared canonical correlation close to $z_{\rho[k]}$. The latter is found by solving \eqref{eq_CCA_master} for $z$ close to $z_{\rho[k]}$, in which $\{c_j^2\}$ are squared canonical correlations after $q-1$ iterations (in particular, $c_k^2=\lambda_k$). Like in \eqref{eq_CCA_master_asymptotic}, Lemma \ref{Lemma_Wachter_spike_answer}, and Lemma \ref{Lemma_Zz_as_linear_form}, we make an asymptotic approximation of the equation \eqref{eq_CCA_master}. In the notations of Lemma \ref{Lemma_Zz_as_linear_form}, we are solving
\begin{equation}
\label{eq_x42}
 L^2(z)-z P(z) Q(z)=0, \qquad z\quad \text{ near }\quad z_{\rho[k]}.
\end{equation}
Note that in this equation we use $\rho=\rho^2[q]$ (which enters in the definition through the correlation structure of $(\u^*,\v^*)$) when constructing $L$, $P$, and $Q$, while $\lambda_k$ corresponds to a different $\rho=\rho^2[k]\ne \rho^2[q]$. Hence, in contrast to the linearization \eqref{eq_x41} in the proof of Lemma \ref{Lemma_Zz_as_linear_form}, the equation will no longer vanish in the leading order. In addition, the terms involving $(z-\lambda_k)^{-1}$ now need a special treatment. Let $\tilde L$, $\tilde P$, $\tilde Q$ denote the large $S$ limits of $L(z_{\rho[k]})$, $P(z_{\rho[k]})$, $Q(z_{\rho[k]})$, respectively, with terms involving $(z-\lambda_k)^{-1}$ omitted; these limits are obtained as in the proofs of Lemma \ref{Lemma_Wachter_spike_answer} and Proposition \ref{Proposition_W_limit}. The asymptotic approximation of \eqref{eq_x42} becomes (we recall from Lemma \ref{Lemma_number_of_roots} that the pole at $z=\lambda_k$ is simple and omit the terms involving $(z-\lambda_k)^{-1}$ before invoking the asymptotic approximations of $L,P,Q$):
\begin{multline}
\label{eq_x43}
 o_p(1)+\tilde L^2-z_{\rho[k]} \tilde P \tilde Q=
   \frac{1}{z-\lambda_k} \cdot \frac{1}{S} \cdot \Biggl[
   \lambda_k \left(\langle \u^*, \v_k \rangle -\sqrt{\lambda_k}  \langle \u^*, \u_k\rangle\right)^2\tilde Q \\+ \lambda_k   \left(\langle \v^*, \u_k\rangle -  \sqrt{\lambda_k} \langle \v^*, \v_k\rangle\right)^2\tilde P  - 2\sqrt{\lambda_k}\left(\langle\u^*, \v_k\rangle-\sqrt{\lambda_k}\langle\u^*, \u_k\rangle\right) \left(\langle \v^*, \u_k\rangle -\sqrt{\lambda_k} \langle \v^*, \v_k\rangle\right)\tilde L
   \Biggr].
\end{multline}
Since the left-hand side is finite and non-zero (by the uniqueness of the solution of the equation \eqref{eq_x13} in Lemma \ref{Lemma_Wachter_spike_answer}), $z-\lambda_k$ should be of order $S^{-1}$ in order for \eqref{eq_x43} to be satisfied.

\medskip

\noindent {\bf Step 3.} It remains to prove that the two angles $\theta_x$, $\theta_y$ corresponding to $\lambda_k$ also change by $O(S^{-1})$. This requires some care: the eigenvector corresponding to $\lambda_k$ changes by a vector whose length is larger than $O(S^{-1})$. However, the change is in a direction orthogonal to the $k$--th signal direction. We present the argument only for $\theta_x$.

We are going to use Theorem \ref{Theorem_master_equation}. We let $\widetilde \U$ and $\widetilde \V$ be the spaces after the first $(q-1)$ iterations and let $\x^q$, $\y^q$ be the signal vectors added to these spaces at the $q$th iteration. Let $P_{\widetilde \U}$ and $P_{\widetilde \V}$ be orthogonal projectors on $\widetilde \U$ and $\widetilde \V$, respectively, and introduce
\begin{equation}
\label{eq_x40}
 \u^*=\frac{\x^q- P_{\widetilde \U} \x^q}{\|\x^q- P_{\widetilde \U} \x^q\|}, \qquad \v^*=\frac{\y^q- P_{\widetilde \V} \y^q}{\|\y^q- P_{\widetilde \V} \y^q\|}.
\end{equation}
Note that $\u^*$ is a  unit vector orthogonal to $\widetilde \U$ and $\v^*$ is a  unit vector orthogonal to $\widetilde \V$. Adding $\u^*$ and $\v^*$ to $\widetilde \U$ and $\widetilde \V$ is the same as adding $\x^q$ and $\y^q$, leading to exactly the same canonical correlations and variables. We apply Theorem \ref{Theorem_master_equation} for this choice\footnote{In all the previous proofs, we applied Theorem \ref{Theorem_master_equation} directly for $(\u^*,\v^*)=(\x^q,\y^q)$. This time a pre-orthogonalization step is necessary in order to better control the remainders in asymptotic expansions.} of $\u^*$ and $\v^*$, so that $\alpha_0\u^*+\sum_{i=1}^{K-1} \alpha_i \u_i$ and $\beta_0 \v^*+\sum_{j=1}^{M-1} \beta_j \v_j$ is the pair of canonical variables after the $q$th iteration corresponding to the squared sample canonical correlations close to $\lambda_k$. Here $(\u_k,\v_k)$ is the pair of canonical variables after the $(q-1)$st iteration corresponding to the $c_k^2=\lambda_k$ correlation. In addition, we let $(\ddot \u,\ddot \v)$ denote the pair of signal vectors added at the $k$th iteration (in particular, $\ddot \u\in\widetilde \U$ and $\ddot \v\in\widetilde \V$). Our goal is to show that the angle $\angle(\ddot \u, \alpha_0\u^*+\sum_{i=1}^{K-1} \alpha_i \u_i)$ is at distance $O(S^{-1})$ from the angle $\angle(\ddot \u, \u_k)$. Here both $\alpha_0\u^*+\sum_{i=1}^{K-1} \alpha_i \u_i$ and $\u_k$ are normalized to be unit vectors. Hence, our goal is to show:
\begin{equation}
\label{eq_x44}
 \frac{\langle \ddot \u, \alpha_0\u^*+\sum_{i=1}^{K-1} \alpha_i \u_i\rangle}{\sqrt{\langle \ddot \u,\ddot \u\rangle}}-\frac{\langle \ddot \u, \u_k \rangle}{\sqrt{\langle \ddot \u,\ddot \u\rangle}}\stackrel{?}{=}O(S^{-1}).
\end{equation}
Recall the formulas \eqref{eq_x8}  for $\alpha_i$, where $z$ solves \eqref{eq_x43}, and $\langle \u^*,\u_i\rangle=\langle \v^*,\v_i\rangle=0$ due to pre-orthogonalization:
\begin{equation}
\label{eq_x46}
   \alpha_i=\alpha_0 \frac{1}{z - c_i^2} \Bigl[ c_i \langle \u^*, \v_i\rangle -z \mathfrak Q_\alpha(z) \langle \v^*, \u_i\rangle \Bigr].
\end{equation}
The prefactor $(z-c_i^2)^{-1}$ is of finite order for $i\ne k$ and of order $S$ (or larger) for $i=k$. All the scalar products in \eqref{eq_x46} are of order $S^{-1/2}$, because we normalized $\u^*$ and $\v^*$. Hence, $\alpha_0$ is $\alpha_k$ times a factor of order $S^{-1/2}$, while all other $\alpha_i$, $i\ge 1$, are $\alpha_0$ times factors of order $S^{-1/2}$. The normalization condition \eqref{eq_normalization} is simply $\sum_i \alpha_i^2=1$, because of the pre-orthogonalization step. After plugging expressions of all $\alpha_i$ through $\alpha_k$, the normalization becomes $\alpha_k^2(1+O(S^{-1}))=1$, which implies\footnote{Since canonical variables are defined up to signs, we can assume $\alpha_k>0$ without loss of generality.} $\alpha_k=1+O(S^{-1})$ and $\alpha_0=O(S^{-1/2})$, $\alpha_i=O(S^{-1})$, $i\ge 1$, $i\ne k$. Hence, in order to prove \eqref{eq_x44}, it remains to show
\begin{equation}
\label{eq_x45}
 \frac{\langle \ddot \u, \alpha_0\u^*+\sum_{i\ne k} \alpha_i \u_i\rangle}{\sqrt{\langle \ddot \u,\ddot \u\rangle}}\stackrel{?}{=}O(S^{-1}).
\end{equation}
Since $\ddot \u\in\widetilde \U$, while $\u^*$ is orthogonal to this space, the term involving $\alpha_0$ in \eqref{eq_x45} vanishes. For the rest, we use \eqref{eq_x46} to rewrite the left-hand side of \eqref{eq_x45} as
\begin{equation*}
 \alpha_0 \Bigl[\langle\u^*,\mathbf a \rangle+ \langle\v^*,\mathbf b \rangle\Bigr], \qquad \mathbf a=\frac{1}{\sqrt{\langle \ddot \u,\ddot \u\rangle}}\sum_{i\ne k}  \frac{c_i \langle \ddot \u, \u_i\rangle }{z - c_i^2}  \v_i,\qquad \mathbf b=-\frac{z \mathfrak Q_\alpha(z)}{{\sqrt{\langle \ddot \u,\ddot \u\rangle}}} \sum_{i\ne k} \frac{\langle \ddot \u, \u_i\rangle}{z-c_i^2}\u_i.
\end{equation*}
The denominators $(z-c_i^2)$ in the last formula are bounded away from $0$ and therefore both $\mathbf a$ and $\mathbf b$ are vectors of length $O(1)$. Recalling the definition \eqref{eq_x40}, and using self-adjointness of the projection operators, we further rewrite the left-hand side of \eqref{eq_x44} as
\begin{equation}
 \frac{\alpha_0}{\|\x^q- P_{\widetilde \U} \x^q\|}\langle \x^q,\, \mathbf a- P_{\widetilde \U} \mathbf a\rangle +  \frac{\alpha_0}{\|\y^q- P_{\widetilde \V} \y^q\|}\langle \y^q,\, \mathbf b- P_{\widetilde \V} \mathbf b\rangle.
\end{equation}
Imagine that $z$ in the definition of $\mathbf a$ and $\mathbf b$ is $\lambda_k$, rather than the solution of \eqref{eq_x43} close to it.
Then, the vectors $\mathbf a- P_{\widetilde \U} \mathbf a$ and $\mathbf b- P_{\widetilde \V} \mathbf b$ would be measurable with respect to the randomness coming from the first $(q-1)$ iterations and have lengths $O(1)$. The vectors $\x^q$ and $\y^q$ are independent from them and have i.i.d.\ Gaussian components. Hence, as in the proof of Lemma \ref{Lemma_asymptotic_approx}, the scalar products are $O(1)$. On the other hand, the prefactors in front of the scalar products are $O(S^{-1})$ and we arrive at the desired bound \eqref{eq_x44}.

If we now switch from $z=\lambda_k$ to $z$ solving \eqref{eq_x43}, then we introduce an additional error term, but since $z-\lambda_k=O(S^{-1})$ by Step 2, this error is small enough so that the bound \eqref{eq_x44} continues to hold.

\section{Discussion on the choice of angles}\label{Section_appendix_x_vs_alpha}

Here we continue the discussion at the end of Section \ref{sec_implic_basic}. Let us explain that the angle between $\widehat \ba$ and $\ba$ (or $\widehat \bb$ and $\bb$) depends on the covariance structure on the noise, in contrast to our angles $\theta_x$ and $\theta_y$. Hence, without introducing additional restrictive assumptions on the covariance matrices of the vectors $\u$ and $\v$, the former angle can not be estimated.

From the technical perspective, computing an angle between $\widehat \ba$ and $\ba$ is equivalent to decomposing $\widehat \ba = p \ba + p' \ba^\perp$, where $p,p'\in\mathbb R$ and $\ba^\perp$ is orthogonal to $\ba$ (for simplicity, assume that $\widehat \ba$, $\ba$, and $\ba^\perp$ are unit vectors). To obtain the angle between $\widehat \x$ and $\x$ of Definition \ref{Definition_sample_setting}, one multiplies the decomposition by $\U^\T$, leading to $\widehat \x=p\x + p' \U^\T \ba^\perp$. The problem is that $\U^\T$ is not an orthogonal transformation, and there is no reason for $\U^\T \ba^\perp$ to be orthogonal to $\x$. A special case is when the covariance matrix of $\u$ is proportional to the identity matrix, so that all $K\times S$ matrix elements of $\U$ are i.i.d. Then by the law of large numbers, for any deterministic unit vector $\mathbf d$, the length of $\U^\T \mathbf d$ is $\sqrt{S}+o(\sqrt{S})$ as $S\to\infty$. Therefore, by the polarization identity, multiplication by $\U^\T$ preserves the angles between the vectors, up to a small error\footnote{In our setting $\widehat \ba$ is not deterministic and the analysis of error needs more careful arguments.}. In contrast, when the covariance matrix of $\u$ is generic and unknown to us, the knowledge of the angle between $\widehat \x$ and $\x$ does not provide information about the angle between $\widehat \ba$ and $\ba$. Indeed, if we multiply $\u$ by an arbitrary $K\times K$ matrix $\Upsilon$ (therefore transforming the covariance of $\u$), Assumption \ref{ass_basic} still holds, hence, the conclusion \eqref{eq_vector_limit1} continues to hold with no changes. But the angle between $\widehat \ba$ and $\ba$ becomes the angle between $\Upsilon^{-\T} \widehat \ba$ and $\Upsilon^{-\T} \ba$, which can be made arbitrary by an appropriate choice of $\Upsilon$, depending only on $\ba$. Similarly, since the sample canonical correlations are invariant with respect to multiplication of the data by $\Upsilon$, consistent estimation of the angle between $\widehat \ba$ and $\ba$ is impossible without additional information beyond these correlations.

We supplement this discussion with two numeric simulations shown in Fig.~\ref{fig_angles_x_alpha} where $\ba=(1,0,0,\dots)^\T$, $\bb=(1,0,0,\dots)^\T$ and vector $\v$ has i.i.d.~$\mathcal N(0,1)$ components. In the first one the vector $\u$ has i.i.d.~$\mathcal N(0,1)$ components. In the second one the components are independent, the first component is $\mathcal N(0,4)$ and all other components are $\mathcal N(0,1)$. In line with Theorem \ref{Theorem_basic_setting}, the angles between  $\widehat \x$ and $\x$ are very close to each other in two simulations. However, the angles between $\widehat \ba$ and $\ba$ are different, and only in the first case they are close to the ones between $\widehat \x$ and $\x$.


\begin{figure}[t]
\begin{subfigure}{.49\textwidth}
  \centering
  \includegraphics[width=1.0\linewidth]{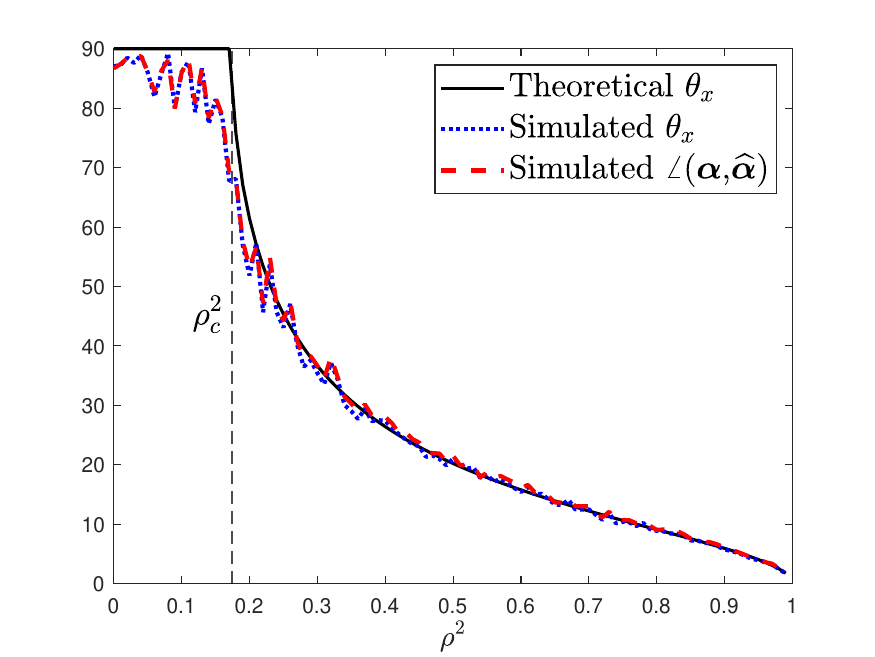}
  \caption{$u_k\thicksim$ i.i.d.~$\mathcal N(0,1)$, $k=1,\ldots,K$.}
  \label{angles_x_a_I}
\end{subfigure}%
\begin{subfigure}{.49\textwidth}
  \centering
  \includegraphics[width=1.0\linewidth]{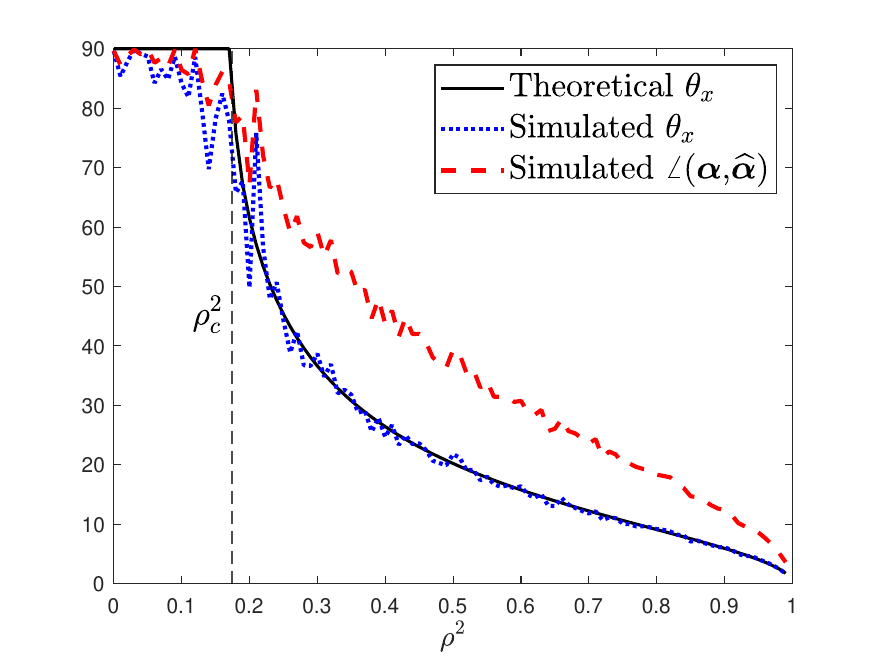}
  \caption{$u_1\thicksim\mathcal N(0,4),\,u_k\thicksim\mathcal N(0,1),k>1$.}
  \label{angles_x_a_nonI}
\end{subfigure}
\caption{Angles: theoretical (black solid) and simulated (blue dotted) between $\widehat \x$ and $\x$ and simulated (red dashed) between $\widehat \ba$ and $\ba$ for different covariance matrices. ${K=500}$, ${M=2500}$, $S=8000$. Simulated curves are based on one simulation from a fixed value $\rho^2$.}
\label{fig_angles_x_alpha}
\end{figure}

\section{Additional Monte Carlo simulations}\label{Section_appendix_MC}

\begin{figure}[t]
\begin{subfigure}{.49\textwidth}
  \centering
  \includegraphics[width=1.0\linewidth]{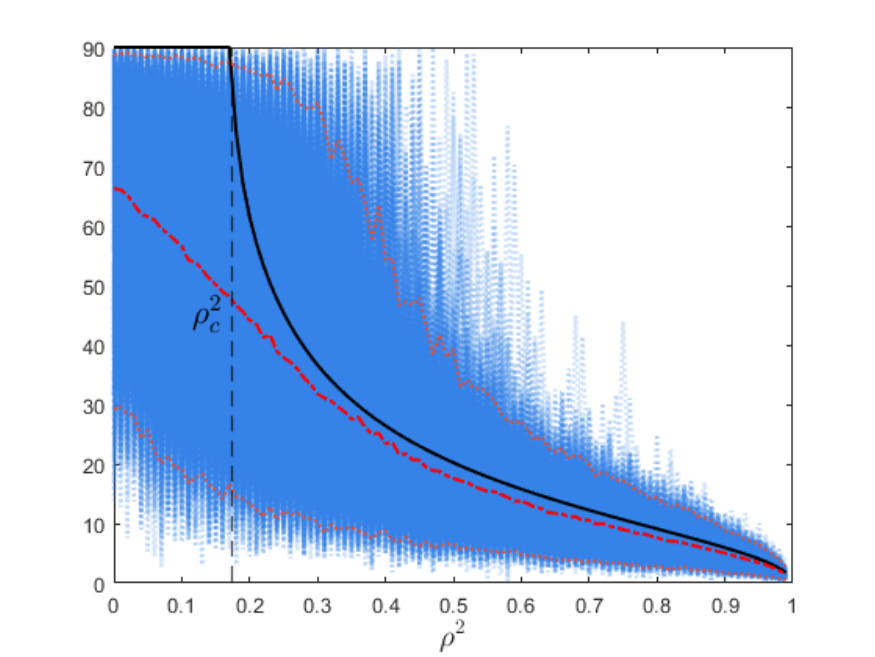}
  \caption{Values of $\theta_x$ (in degrees): \\theoretical (black) vs.~simulated (blue) angles.}
  \label{thetax_pic100}
\end{subfigure}%
\begin{subfigure}{.49\textwidth}
  \centering
  \includegraphics[width=1.0\linewidth]{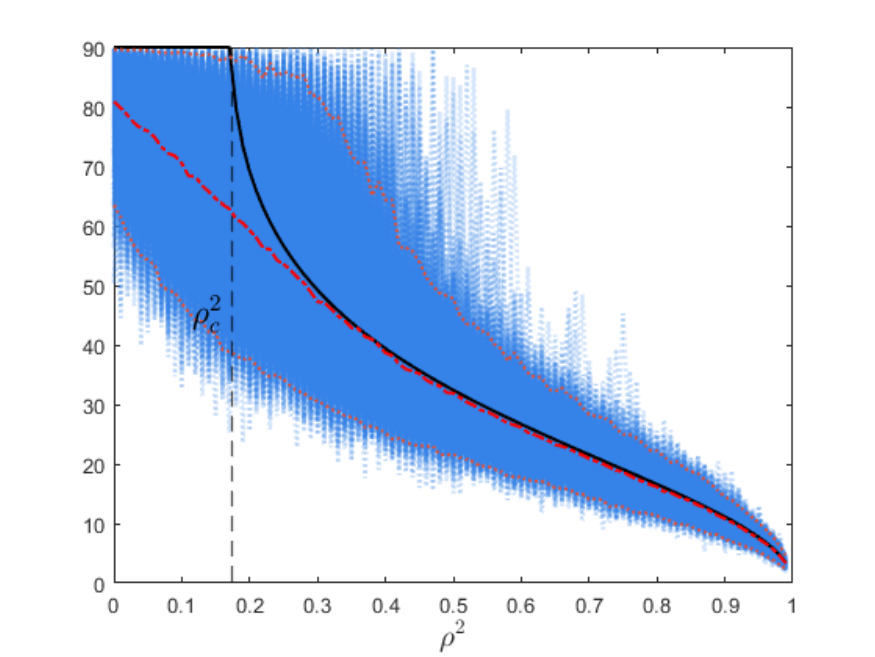}
  \caption{Values of $\theta_y$ (in degrees): \\theoretical (black) vs.~simulated (blue) angles.}
  \label{thetay_pic100}
\end{subfigure}
\caption{Comparison of theoretical (straight black) and simulated results (blue dotted) for ${K=5}$, ${M=25}$, $S=80$, $MC=1000$ simulations. Blue dotted clouds are $1000$ different MC realizations. Red dotted curves represent a $95\%$ empirical confidence interval for simulated angles. Red dash-dotted curves are sample averages of simulated angles.}
\label{fig_angles_theor_simul100}
\end{figure}

Our first set of simulations investigates the speed of convergence of the angles $\theta_x$ and $\theta_y$ to their asymptotic values given by \eqref{eq_vector_limit1} and \eqref{eq_vector_limit2} in the Gaussian setting of Theorem \ref{Theorem_basic_setting}. Figures \ref{fig_angles_theor_simul100}-\ref{fig_angles_theor_simul_many} show the effect of dimension increase on the distributions of the angles. In all these figures we plot angles as functions of $\rho^2$. As we can see, for very small samples, as in Figure \ref{fig_angles_theor_simul100}, the variance is quite high, yet the average values for the angles still match theoretical predictions of Theorem \ref{Theorem_basic_setting}. As we scale all dimensions by $10$ in Figure \ref{fig_angles_theor_simul10}, the realizations become much tighter. Scaling all dimensions again by $10$ in Figure \ref{fig_angles_theor_simul_many} brings the Monte Carlo simulations very close to the theoretical limit.

\begin{figure}[t]
\begin{subfigure}{.49\textwidth}
  \centering
  \includegraphics[width=1.0\linewidth]{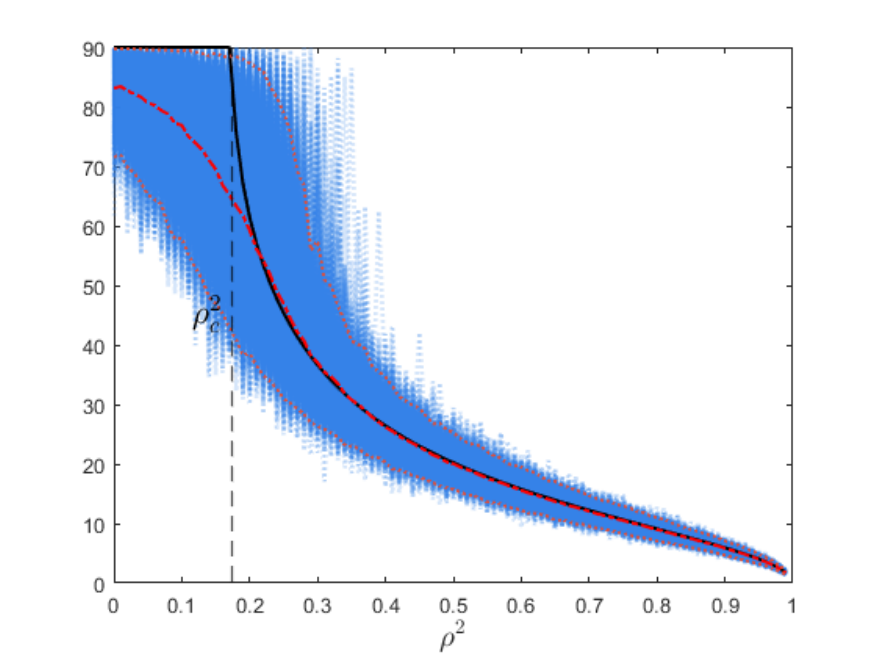}
  \caption{Values of $\theta_x$ (in degrees): \\theoretical (black) vs.~simulated (blue) angles.}
  \label{thetax_pic10}
\end{subfigure}%
\begin{subfigure}{.49\textwidth}
  \centering
  \includegraphics[width=1.0\linewidth]{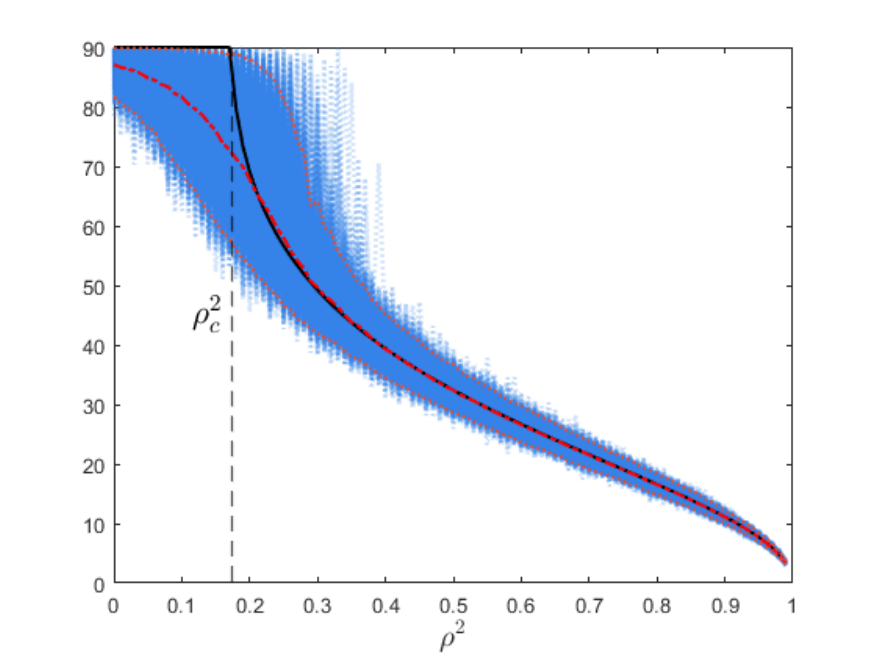}
  \caption{Values of $\theta_y$ (in degrees): \\theoretical (black) vs.~simulated (blue) angles.}
  \label{thetay_pic10}
\end{subfigure}
\caption{Comparison of theoretical (straight black) and simulated results (blue dotted) for ${K=50}$, ${M=250}$, $S=800$, $MC=1000$ simulations. Blue dotted clouds are $1000$ different MC realizations. Red dotted curves represent a $95\%$ empirical confidence interval for simulated angles. Red dash-dotted curves are sample averages of simulated angles.}
\label{fig_angles_theor_simul10}
\end{figure}

\begin{figure}[t]
\begin{subfigure}{.49\textwidth}
  \centering
  \includegraphics[width=1.0\linewidth]{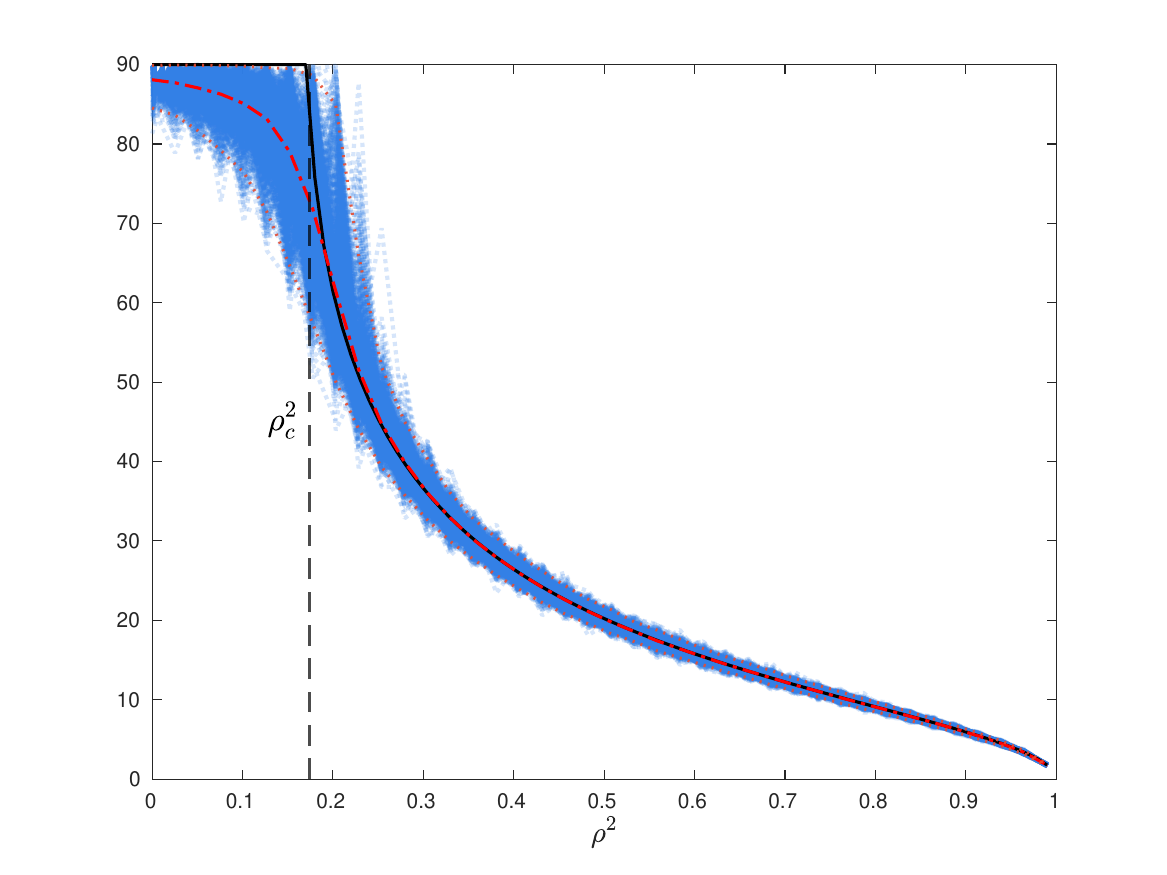}
  \caption{Values of $\theta_x$ (in degrees): \\theoretical (black) vs.~simulated (blue) angles.}
  \label{thetax_pic_many}
\end{subfigure}%
\begin{subfigure}{.49\textwidth}
  \centering
  \includegraphics[width=1.0\linewidth]{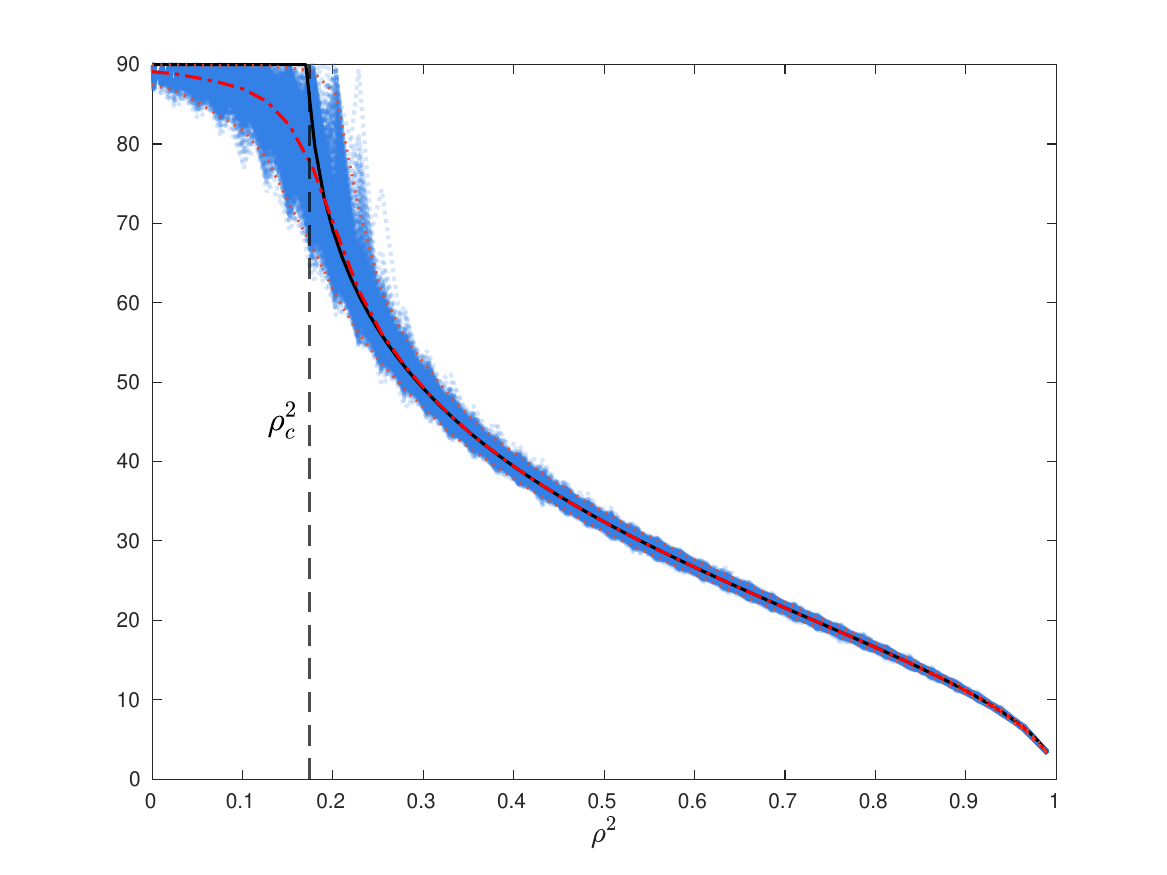}
  \caption{Values of $\theta_y$ (in degrees): \\theoretical (black) vs.~simulated (blue) angles.}
  \label{thetay_pic_many}
\end{subfigure}
\caption{Comparison of theoretical (straight black) and simulated results (blue dotted) for ${K=500}$, ${M=2500}$, $S=8000$, $MC=1000$ simulations. Blue dotted clouds are $1000$ different MC realizations. Red dotted curves represent a $95\%$ empirical confidence interval for simulated angles. Red dash-dotted curves are sample averages of simulated angles.}
\label{fig_angles_theor_simul_many}
\end{figure}

\FloatBarrier

\bigskip

Our next set of simulations accompanies Theorem \ref{Theorem_joint_CLT_main}. Based on $10000$ MC realizations we compute numerically  the matrix $\Sigma$, as defined in \eqref{eq_covariance_def}, for fixed $S=800$ and three noise distributions: Gaussian, a four-moment Gaussian distribution supported on three points\footnote{$P(X=\pm\sqrt{3})=\frac{1}{6}$, $P(X=0)=\frac{2}{3}$, leading to $\E X^2=1$, $\E X^4=3$.} $\{-\sqrt{3},0,\sqrt{3}\}$, and $t$ distribution with five degrees of freedom, whose fourth moment is non-Gaussian. We start from $\U$ and $\V$ matrices filled with i.i.d.\ elements, and then replace the first row of $\V$ by a linear combination of the first rows of $\U$ and $\V$ to introduce a correlation as in Example \ref{ex_10}. We use centering \eqref{eq_Z_def}, rather than by sample means. Figure \ref{Figure_CLT_check} compares the numerical covariances as functions of $\rho^2$ against the formulas of Section \ref{Section_detailed_CLT}. The match for Gaussian noises is very good, except for deterioration of the quality as $\rho^2$ approaches the critical value $\rho^2_c$. The results for four-moment Gaussian noises are indistinguishable from the Gaussian case, which suggest a conjectural extension of Theorem \ref{Theorem_joint_CLT_main} to the setting of Section \ref{Section_non_Gaussian}. The results for the $t$ distribution with five degrees of freedom are significantly further away from theoretical curves, which indicates the importance of the fourth moment.

We also run experiments for the $t$ distribution with $2.5$ degrees of freedom, which has the second, but not the third moment. We do not show the results, as the match is dramatically worse than Figure \ref{Figure_CLT_check}: the variances are much larger. We expect that heavy-tailed data leads to slower speed of convergence in the law of large numbers of Theorem \ref{Theorem_basic_setting}. We also do not expect fluctuations to remain Gaussian in this case.

\begin{figure}[t]
 \includegraphics[width=0.45\linewidth]{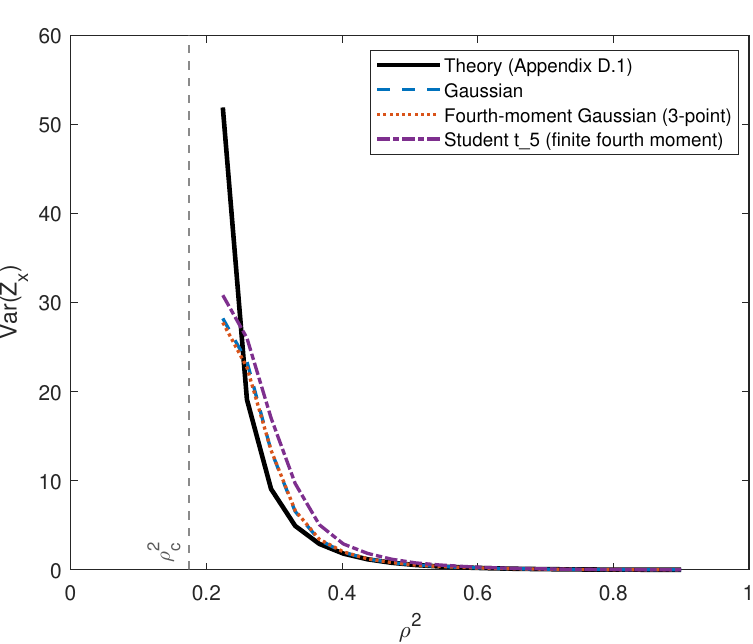}
  \hfill
 \includegraphics[width=0.45\linewidth]{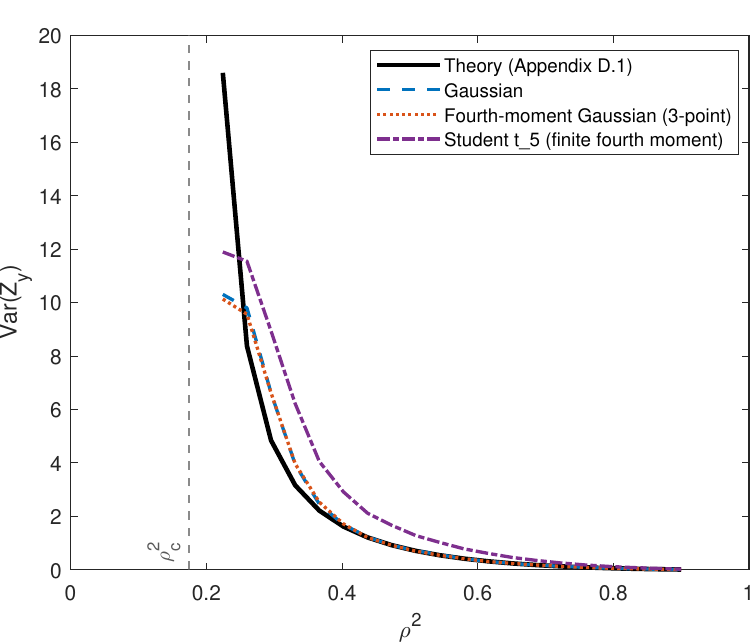}

 \includegraphics[width=0.45\linewidth]{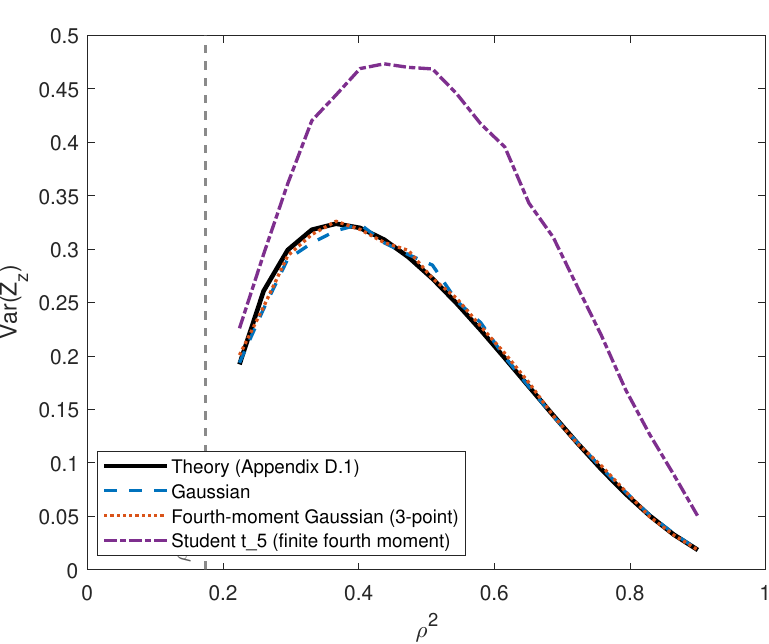}
  \hfill
 \includegraphics[width=0.45\linewidth]{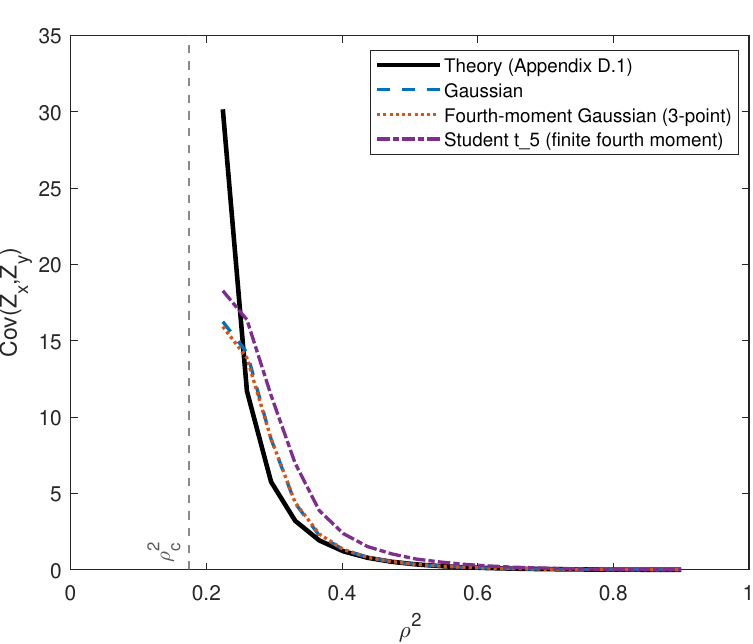}

 \includegraphics[width=0.45\linewidth]{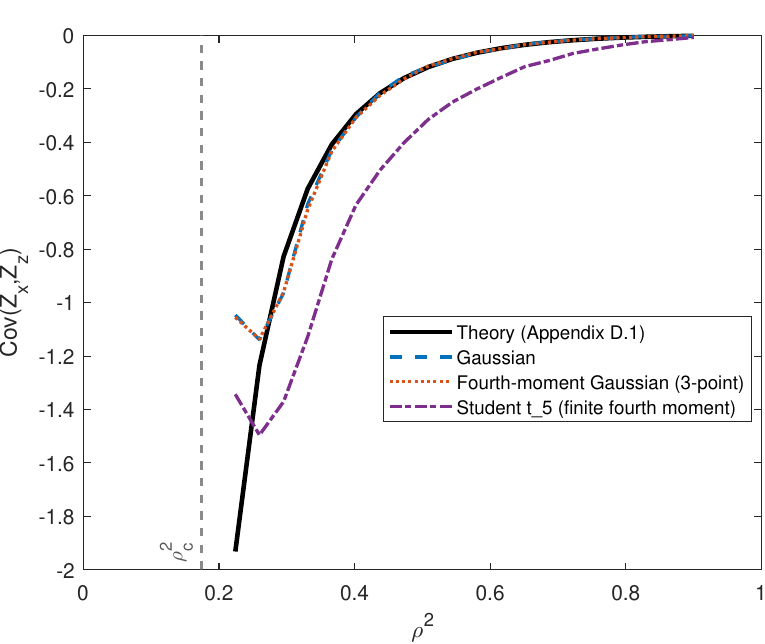}
  \hfill
  \includegraphics[width=0.45\linewidth]{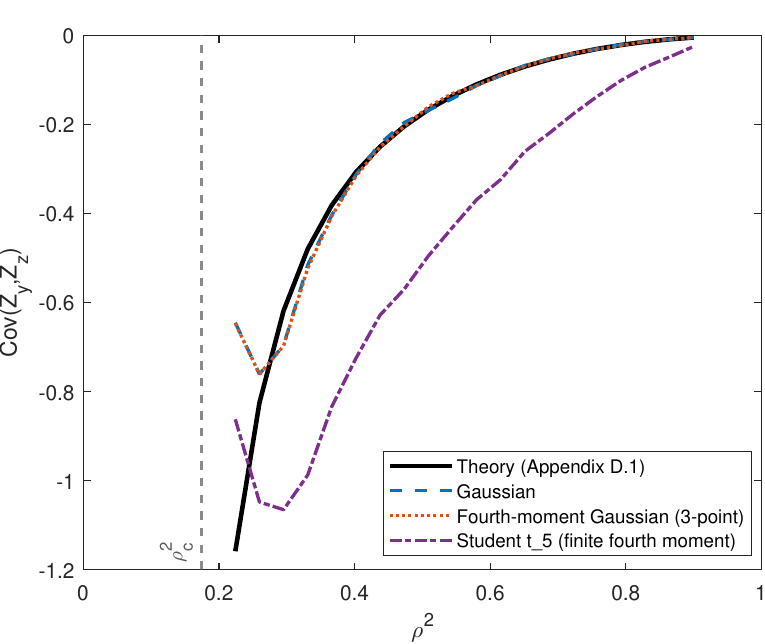}
  \caption{Comparison of theoretical (Section \ref{Section_detailed_CLT}) and simulated results for ${K=50}$, ${M=250}$, $S=800$, $MC=10000$ simulations. Six figures correspond to six elements of the second moment matrix $\Sigma$ of \eqref{eq_covariance_def}. \label{Figure_CLT_check} }
\end{figure}

\section{Data}
\label{Section_appendix_data}
We use 80 largest (as of July 2023) companies in each category. The list of cyclical stocks consists of AMZN, AVY, AZO,	BALL, BBWI,	BBY, BC, BKNG, BWA,	CASY, CCK, CCL,	CHDN, CHH, CMG,	CROX, CUK, DECK, DHI, DKS, DPZ,	DRI, EBAY, EXPE, F,	GNTX, GPC, GPK,	H, HAS, HD, HMC, IGT, IHG, IP, KMX, LAD, LEA, LEN, LKQ, LOW, LULU, LVS, MAR, MAT, MCD, MELI, MGA, MGM, MHK, MTN, NKE, NVR, ORLY, PAG, PHM, PII, PKG, RCL, RL, ROL, ROST, SBUX, SCI, SEE, SKX, TCOM, TJX, TM, TOL, TPR, TPX, TSCO, TXRH, ULTA, VFC, WHR, WSM,	WYNN, YUM.

The list of non-cyclical stocks consists of ABEV, ADM, ANDE, ATGE, BF-B, BG, BGS, BIG, BTI,	BUD, CAG, CALM,	CHD, CL, CLX, COST,	CPB, CVGW, DAR,	DEO, DG, DLTR, EL, EPC,	FDP, FIZZ, FLO,	FMX, GHC, GIS, HAIN, HLF, HRL, HSY,	IFF, IMKTA,	INGR, JJSF,	K, KDP,	KMB, KO,	KR,	LANC, LMNR,	MDLZ, MGPI,	MKC, MNST, MO, NUS,	NWL, PEP, PG, PM, PPC, PRDO, PRMW, SAM,	SENEA, SJM,	SPB, SPTN, STKL, STRA, STZ, SYY, TAP, TGT, THS, TR,	TSN, UL, UNFI, USNA, UVV, VLGEA, WBA, WMK, WMT.

Results of PCA separately applied to both stock categories are shown in Figure \ref{fig_PCA_stocks}.

\begin{figure}[t]
\begin{subfigure}{.49\textwidth}
  \centering
  \includegraphics[width=1.0\linewidth]{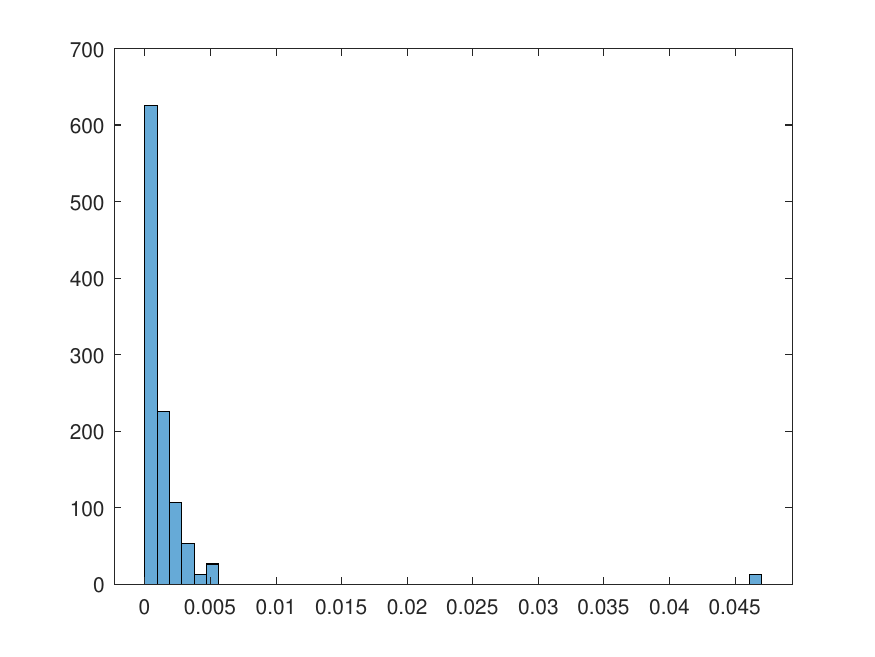}
  \caption{PCA: Cyclical stock returns.}
  \label{fig_PCA_cycl}
\end{subfigure}
\begin{subfigure}{.49\textwidth}
  \centering
  \includegraphics[width=1.0\linewidth]{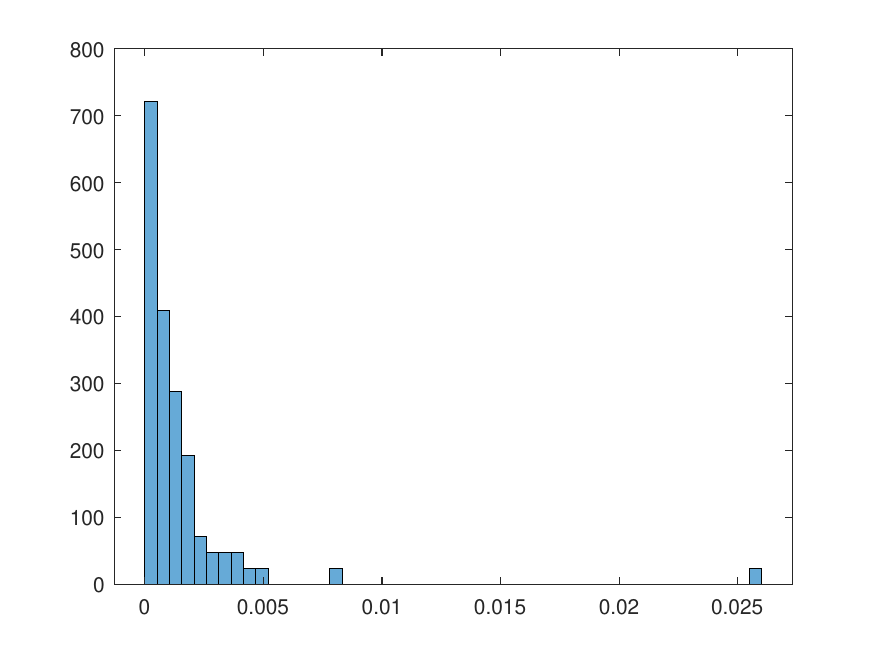}
  \caption{PCA: Non-cyclical stock returns.}
  \label{fig_PCA_noncycl}
\end{subfigure}
\caption{Eigenvalues from PCA for cyclical and non-cyclical groups of stock returns.}
\label{fig_PCA_stocks}
\end{figure}

\setlength{\bibsep}{0.0pt}

\bibliographystyle{abbrvnat}

\bibliography{CCA_biblio}

\end{document}